\documentclass[conference]{IEEEtran}
\IEEEoverridecommandlockouts

\def\BibTeX{{\rm B\kern-.05em{\sc i\kern-.025em b}\kern-.08em
    T\kern-.1667em\lower.7ex\hbox{E}\kern-.125emX}}

\usepackage[ruled,vlined]{algorithm2e}
\usepackage{algpseudocode}
\usepackage{setspace}
\usepackage{mathtools}
\usepackage{enumitem}
\usepackage{amsmath, amsthm, bm}

\newtheorem{myExa}{Example}
\newtheorem{myLemma}{Lemma}

\usepackage{multirow}
\usepackage{stackengine}
\usepackage[normal, footnotesize]{caption}
\usepackage{pifont}
\usepackage{hyperref}
\usepackage{color}

\newcommand{\name}{$\mathsf{MUST}$}
\newcommand{\bi}{$\mathsf{MR}$}
\newcommand{\bii}{$\mathsf{JE}$}
\newcommand{\problem}{MSTM}
\newcommand{\tabincell}[2]{\begin{tabular}{@{}#1@{}}#2\end{tabular}}

\newcommand{\squishlist}{
	\begin{list}{$\bullet$}
		{ \setlength{\itemsep}{1pt}
			\setlength{\parsep}{1pt}
			\setlength{\topsep}{2.5pt}
			\setlength{\partopsep}{0.5pt}
			\setlength{\leftmargin}{1em}
			\setlength{\labelwidth}{1em}
			\setlength{\labelsep}{0.6em}
		}
	}
	\newcommand{\squishend}{
	\end{list}
}

\makeatletter
  \newcommand\figcaption{\def\@captype{figure}\caption}
  \newcommand\tabcaption{\def\@captype{table}\caption}
\makeatother
    
\begin{document}

\title{\textsf{MUST}: An Effective and Scalable Framework for Multimodal Search of Target Modality}


\author{Mengzhao Wang\textsuperscript{1},
        Xiangyu Ke\textsuperscript{1},
        Xiaoliang Xu\textsuperscript{2},
        Lu Chen\textsuperscript{1},
        Yunjun Gao\textsuperscript{1},
        Pinpin Huang\textsuperscript{2},
        Runkai Zhu\textsuperscript{2}\\
\IEEEauthorblockN{\textsuperscript{1}Zhejiang University, Hangzhou, China \textsuperscript{2}Hangzhou Dianzi University, Hangzhou, China}
\IEEEauthorblockA{\textit{\{wmzssy,xiangyu.ke,luchen,gaoyj\}@zju.edu.cn; \{xxl,hpp,runkai.zhu\}@hdu.edu.cn}
}
}

\maketitle

\begin{abstract}
We investigate the problem of multimodal search of target modality, where the task involves enhancing a query in a specific target modality by integrating information from auxiliary modalities. The goal is to retrieve relevant objects whose contents in the target modality match the specified multimodal query. The paper first introduces two baseline approaches that integrate techniques from the Database, Information Retrieval, and Computer Vision communities. These baselines either merge the results of separate vector searches for each modality or perform a single-channel vector search by fusing all modalities. However, both baselines have limitations in terms of efficiency and accuracy as they fail to adequately consider the varying importance of fusing information across modalities. 
To overcome these limitations, the paper proposes a novel framework, \underline{\textbf{Mu}}ltimodal \underline{\textbf{S}}earch of \underline{\textbf{T}}arget Modality, called {\name}. Our framework employs a hybrid fusion mechanism, combining different modalities at multiple stages. Notably, we leverage vector weight learning to determine the importance of each modality, thereby enhancing the accuracy of joint similarity measurement. Additionally, the proposed framework utilizes a fused proximity graph index, enabling efficient joint search for multimodal queries. {\name} offers several other advantageous properties, including pluggable design to integrate any advanced embedding techniques, user flexibility to customize weight preferences, and modularized index construction. Extensive experiments on real-world datasets demonstrate the superiority of {\name} over the baselines in terms of both search accuracy and efficiency. Our framework achieves over 10$\times$ faster search times while attaining an average of 93\% higher accuracy. Furthermore, {\name} exhibits scalability to datasets containing more than 10 million data elements.
\end{abstract}


\section{Introduction}
\label{sec: introduction}
Multimodal search \cite{tautkute2019deepstyle,etzold2012context,jandial2022sac,wen2021comprehensive} represents a cutting-edge approach to information retrieval that revolutionizes how we interact with vast and diverse data sources. Traditional search engines have relied predominantly on textual queries to deliver results \cite{10.5555/1796408}. However, with the proliferation of the even expressive multimedia contents, such as images, videos, and audio \cite{etzold2012context,baltruvsaitis2018multimodal}, the need for a more comprehensive search paradigm emerged \cite{vo2019composing,PatelGBY22}. Multimodal search aims to address this challenge by integrating information from multiple modalities, unlocking the potential to provide richer and more contextually relevant search results, surpassing traditional search paradigms in various retrieval tasks \cite{wen2021comprehensive,google_mum,yu2021multi}.
By leveraging advanced techniques in natural language processing \cite{DevlinCLT19}, computer vision \cite{he2016deep}, and data fusion \cite{CLIP2021}, multimodal search systems have the capacity to revolutionize user experiences and enable applications that span industries, from e-commerce and healthcare to smart home systems and autonomous vehicles \cite{10.1145/3539597.3570423,levy2023chatting,ramanishka2020describing,ma2022multimodal,MajumdarSLAPB20}.

\begin{figure}[!tb]
  \centering
  \setlength{\abovecaptionskip}{0cm}
  \setlength{\belowcaptionskip}{-0.4cm}
  \includegraphics[width=\linewidth]{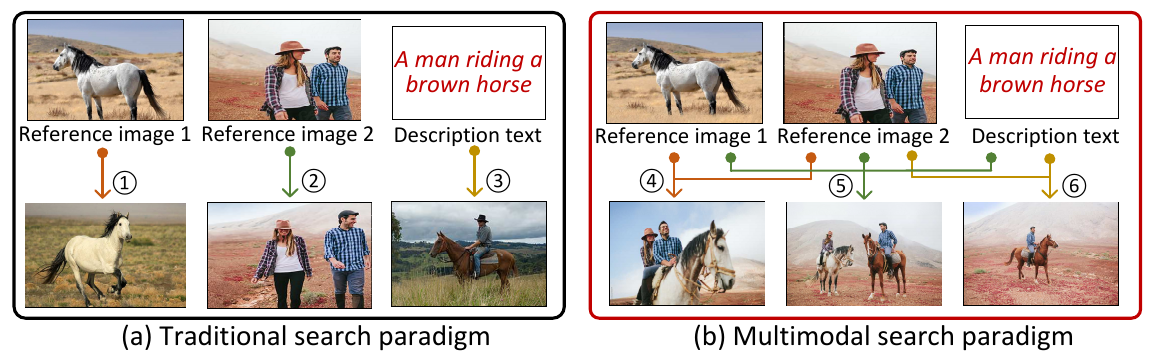}
  \caption{{An example of image search based on two different paradigms. The queries consist of two reference images and a description text. The goal is to search for images that not only resemble the input images, but also modify certain aspects according to the description text.}}
  \label{fig: intro_example}
  \vspace{-0.2cm}
\end{figure}

We investigate a targeted variant of multimodal search, known as {\bf M}ultimodal {\bf S}earch of {\bf T}arget {\bf M}odality (\problem), which enhances data in a specific target modality by integrating information from other modalities. In \problem, the query input includes multiple modalities: one target modality and several auxiliary modality inputs. The target modality offers implicit context, while the auxiliary modalities introduce new traits to the target modality input. The goal is to retrieve relevant objects whose contents in the target modality align with the specified multimodal query. A key feature of {\problem} is its ability to iteratively use a returned target modality example, like an image, as a reference and express differences through auxiliary modalities like text or additional images, allowing users to precisely define search criteria and obtain more tailored results based on preferences and needs.

\begin{myExa}
 In image search tasks, using a single modality as a query may not fully capture users' intentions \cite{zhao2022progressive}. Fig. \ref{fig: intro_example}(a) illustrates an image search using two reference images and a simple text description as queries. Each single modal query yields distinct outputs, highlighting the limitations of relying solely on one modality. 

 In contrast, the right-hand side of Fig. \ref{fig: intro_example} shows that different combinations of query inputs result in diverse outputs, each conveying a wealth of information. Additionally, when setting image 1 as the target modality in query \ding{175}, the emphasis is on the horse, while in query \ding{176}, the auxiliary information in the text guides the focus on the human-horse pairing while preserving all elements in the images. This exemplifies the power of utilizing {\problem} queries to precisely express user preferences and obtain more comprehensive search results.
\end{myExa}

\noindent\textbf{Other Applications}. Apart from its direct use in image search for e-commerce, {\problem} finds diverse applications across various domains. In healthcare, {\problem} enhances decision-making by augmenting medical images with patient electronic health records and symptom descriptions. This enables searching for past medical images with known decision labels, providing medical professionals with a comprehensive view of a patient's condition and facilitating more informed decision-making.
Smart home systems benefit from targeted multimodal fusion, as it allows for tailored responses to users' voice commands, contextual information from sensors, and user profiles and configuration histories. This integration results in more personalized and efficient interactions with smart home devices.
In each scenario, {\problem} empowers these applications to deliver more sophisticated, relevant, and personalized outputs, significantly enhancing the user experience in the digital era.

\noindent\textbf{Possible Solutions and Limitations}. To tackle the {\problem} problem, we propose two baselines: \textit{Multi-streamed Retrieval} ({\bi}) from the Database (DB) and Information Retrieval (IR) communities, and \textit{Joint Embedding} ({\bii}) from the Computer Vision (CV) community. Both baselines utilize advanced embedding techniques to transform multimodal inputs into high-dimensional vectors and subsequently perform vector searches to retrieve results. The main difference between the two lies in how they handle the embedding of a query, leading to distinct implementations of vector search in the context of {\problem}.
In {\bi}, the query is divided into smaller subqueries and separate solutions are applied for each modality \cite{tautkute2019deepstyle}. The results from all candidate sets are merged to obtain the final query result. While this approach can use established single-modal search methods, it still suffers from accuracy and efficiency issues: The candidate sets may be too large or irrelevant due to incomplete, noisy, or ambiguous information in the target modality or auxiliary modalities \cite{zagoris2010mmretrieval,ADBV,zhang2014query}. The evaluation on million-scale data shows that it requires more than $10^4$ candidates per modality to achieve the best top-100 results, yet the recall rate remains low, being less than 0.2 ({Fig. \ref{fig: efficiency evaluation}}).
On the other hand, {\bii} addresses the problem through multimodal learning. This approach embeds the features of all modality inputs into a single vector, allowing for vector search \cite{graph_survey_vldb2021,HNSW,NSG} on a corpus of vectors for the target modality. However, {\bii} faces challenges in synergistically understanding multimodal information. The modality gap introduces ambiguity in determining what information is essential and what can be disregarded \cite{zhang2020joint}, making joint embedding still an open problem \cite{jandial2022sac,DelmasRCL22}. Notably, even with the best joint embedding approach, the top-1 recall rate barely surpasses 0.4 (\textbf{\S \ref{subsec: accuracy evaluation}}).

\noindent\textbf{Our Solution}. We present a novel framework for \underline{\textbf{Mu}}ltimodal \underline{\textbf{S}}earch of \underline{\textbf{T}}arget Modality, named {\name}. This framework utilizes a hybrid fusion mechanism to combine various modalities at multiple stages, enhancing search accuracy by minimizing similarity measurement errors. Additionally, it constructs a fused proximity graph index encompassing all modal information and performs an efficient joint search. Indeed, {\name} distinguishes itself from the two baselines in three main aspects.
\underline{First}, {\name} enables the fusion of multiple modalities through a composition vector generated using multimodal learning models like CLIP \cite{baldrati2022conditioned}. Meanwhile, {\name} still supports the separate embedding of different modalities. Note that the embedding component in {\name} is {\em pluggable}, allowing seamless integration of any newly-devised encoder into the system.
\underline{Second}, {\name} provides a vector weight learning model to obtain the relative weights of different modalities, which projects an object to a unified high-dimensional vector space by concatenating different modal vectors with these weights. The loss function pulls the anchor closer to the positive example and pushes it away from the negative examples, based on the joint similarity in the unified space. Importantly, the learned weights capture the significance of different modalities, not their specific contents, leading to improved generalization across various query workloads. Despite the learned weights, {\name} still allows users to customize their weight preferences if desired.
\underline{Third}, {\name} builds a fused index for all modal information (not a separate index for each modality). Based on this index, it implements a joint search strategy for the multimodal query to obtain results efficiently. {Notably, {\name} employs a general pipeline to construct the fused index by amalgamating fined-grained components, enabling \textit{flexibility} to seamlessly integrate these components from current proximity graphs.} Furthermore, {\name} enhances indexing and search performance by re-assembling index components and optimizing the multi-vector computations.

\noindent\textbf{Contributions}. To the best of our knowledge, this is the first work that systematically explores the {\problem} problem in data embedding, importance mining, indexing, and search strategies.
The main contributions are:

\squishlist

\item We explore the {\problem} problem, which enhances a query in a specific target modality by combining information from auxiliary modalities (\textbf{\S \ref{sec: pre}}). We embed objects using various encoders and build two baselines for {\problem} by integrating the techniques from DB, IR, and CV communities (\textbf{\S \ref{subsec: baselines}}).

\item We present~\name{}, a new framework that uses a hybrid fusion mechanism to improve search accuracy and efficiency for any modality combination of {\problem} (\textbf{\S \ref{sec: framework overview}}). {\name} supports pluggable unimodal and multimodal embedding methods (\textbf{\S \ref{subsec: accuracy evaluation}}) and various graph indexes (\textbf{\S \ref{subsec: ablation}}).

\item We provide a multi-vector representation method for multimodal objects and queries (\textbf{\S \ref{sec: embedding}}). Each modality of an object or query is transformed into a high-dimensional vector via unimodal or multimodal encoders. This way, we can describe an object or query more fully with multiple vectors.

\item We present a lightweight and effective vector weight learning model, to get the relative weights of different modalities (\textbf{\S \ref{sec: vector weight learning}}). The learned weights capture the importance of different modalities and adapt to various query workloads.

\item We provide a component-based index construction pipeline to build a fused index for all modal information and execute a joint search of the multimodal query (\textbf{\S \ref{sec: indexing and search algorithms}}). Our pipeline achieves better performance by re-assembling existing components and optimizing multi-vector computations.

\item We implement~\name{} and evaluate it on five real-world datasets and four extended datasets to verify its accuracy, efficiency, and scalability (\textbf{\S \ref{sec: experiments}}). We show that search accuracy can be improved significantly by multi-stage fusion (\textbf{\S \ref{subsec: accuracy evaluation}}) or combining more modalities (\textbf{\S \ref{subsec: scalability}}).

\squishend

\section{Preliminaries}\label{sec: pre}
In this section, we present the essential terminology and formally define the {\problem} problem. The frequently-used notations are summarized in Table~\ref{tab:notations}.

\begin{table}[!tb]
  \centering
  \setlength{\abovecaptionskip}{0cm}
  \fontsize{8pt}{3.3mm}\selectfont
  \caption{Frequently used notations}
  \label{notations}
  \begin{tabular}{|p{40pt}|p{180pt}|}
    \hline
    \textbf{Notations} & \textbf{Descriptions}\\
    \hline
    $\mathcal{S},o$ & A set of multimodal objects, an object in $\mathcal{S}$\\
    \hline
    $q$ & {A multimodal query input}\\
    \hline
    $o^i,q^i$ & {The data part of $o$, $q$ in the $i$-th modality} \\
    \hline
    $m$ & {The number of modalities in $o$ ($o\in \mathcal{S}$)}\\
    \hline
    $t$ & {The number of modalities in $q$ ($t\leq m$, usually $t= m$)}\\
    \hline
    $\phi_i(\cdot)$ & The encoder for the $i$-th modality\\
    \hline
    $\Phi(\cdot,\cdot,\cdots)$ & A multimodal encoder \\
    \hline
    $IP(\cdot,\cdot)$ & The inner product (\textit{IP}) of two vectors\\
    \hline
    $SME$ & The similarity measure error (Eq. \ref{equ: sme}) \\
    \hline
    $\omega_i$ & The vector weight in the $i$-th modality \\
    \hline
    $\boldsymbol{\hat{q}},\boldsymbol{\hat{o}}$ & The concatenated vectors of $q$, $o$ \\
    \hline
  \end{tabular}
  \vspace{-0.4cm}
  \label{tab:notations}
\end{table}

\noindent\textbf{Object Set.} The object set $\mathcal{S}$ consists of $n$ objects, each $o\in \mathcal{S}$ possessing $m$ modalities ($m\geq 1$), represented as $o^i$ ($0\leq i\leq m-1$). In our focus, $m>1$, indicates that each object in the set has multiple modalities.
The versatility of an object set allows it to represent various types of data. For example, in the context of movies, each object may comprise three modalities—video, image, and text—corresponding to the movie itself, its poster, and introduction, respectively.

\noindent\textbf{Query.} A query $q$ consists of $t$ modalities, each represented by $q^i$ ($0\leq i \leq t-1$, $1\leq t\leq m$). We focus on the case where $t>1$ indicates a multimodal query input. In this context, we specify one of the query modalities as the target, which is used for rendering the search results\footnote{{We can also fuse other modalities into the target modality to form a composition vector, i.e., Option 2 in Fig. \ref{fig: our_framework}(f) and discussion in {\S \ref{sec: framework overview}}.}}. For simplicity, throughout the following discussion, we assume that $q^0$ represents the target modality.
However, it is important to note that users may not always provide a multimodal query input. Our solution for solving {\problem} is designed to accommodate such cases when certain modalities, including the targeted modality, might be absent. Further details are provided in \textbf{\S \ref{sec: discussion}}.

\noindent\textbf{Problem Statement.} Given an object set $\mathcal{S}$, a query input $q$, and a positive integer $k$, the goal of {\bf M}ultimodal {\bf S}earch of the {\bf T}arget {\bf M}odality ({\problem}) problem is to find $k$ objects from $\mathcal{S}$ that best match the query. Specifically, the target modality part of each object in the result set $R$ should closely resemble $q^0$, while also adhering to certain aspects specified by the set $\{q^i|1\leq i \leq t-1\}$, which comprises the auxiliary modalities of the query. 
In the image search example of Fig. \ref{fig: intro_example}, the output image of query \ding{176} contains all elements present in two reference images and also matches the provided text description. This exemplifies a successful search result that accurately captures the user's query intent.

\noindent\textbf{Performance Metric.} To measure the accuracy of the search results, we use the recall rate as the evaluation metric.
Suppose there are $k^{\prime}$ ground-truth objects of $q$ (denoted by $\mathcal{G}$) in the object set $\mathcal{S}$. The recall rate at $k$ is formally defined as below:
\begin{equation}
  \label{equ: recall}
  Recall@k(k^{\prime})=\frac{|{R}\cap\mathcal{G}|}{k^{\prime}}\quad.
\end{equation}

\section{Baselines}\label{subsec: baselines}
To tackle the {\problem} problem, we propose two baselines, namely {\bi} and {\bii}, leveraging current advancements.

\noindent\textbf{Basic Idea.} Given an object set $\mathcal{S}$ and a query $q$, we transform the different modalities into high-dimensional feature vectors through embedding\footnote{We apply state-of-the-art embedding techniques for each modality (please refer to Appendix \ref{appendix: encoders} for the specific encoders used in this paper).
}. These vectors capture the essence of each modality and allow for efficient comparison and retrieval \cite{vo2019composing,CLIP2021}.
To solve {\problem}, we conduct a vector search procedure using one or more vector indexes constructed from the feature vectors of objects in $\mathcal{S}$. The similarity between the query vector and potential result vectors from $\mathcal{S}$ is evaluated using the inner product (\textit{IP}). As such, we can find objects in $\mathcal{S}$ whose target modality parts closely match the multimodal query.

\begin{figure}[!tb]
  \centering
  \setlength{\abovecaptionskip}{0.1cm}
  \setlength{\belowcaptionskip}{-0.6cm}
  \includegraphics[width=\linewidth]{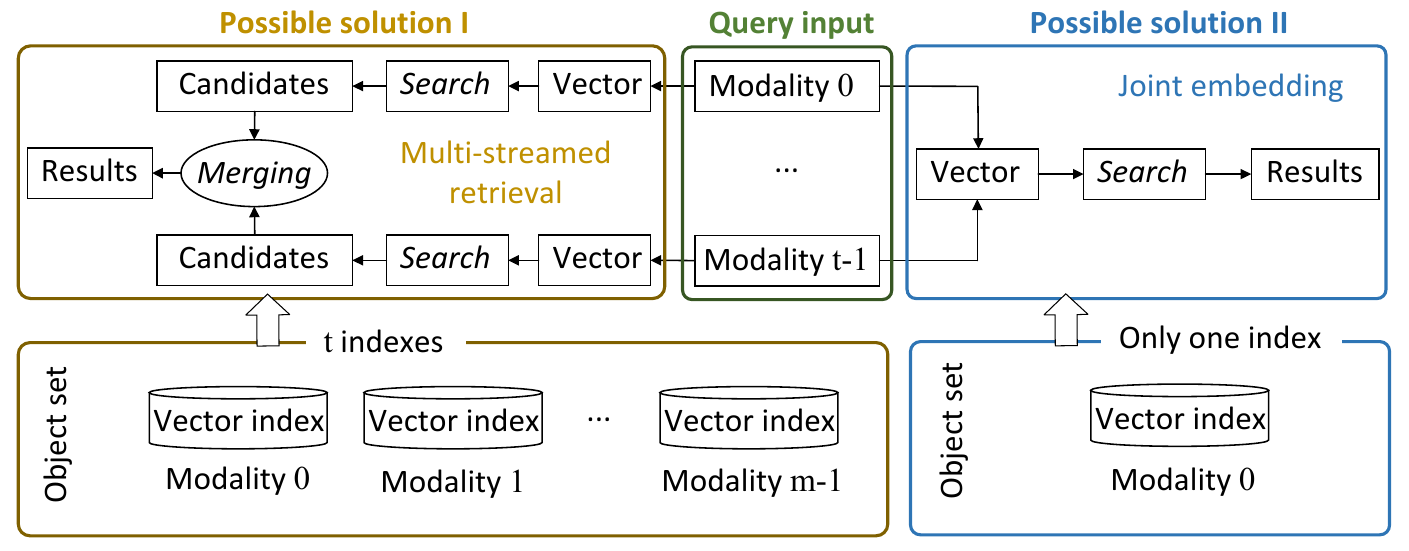}
  \caption{{Overview of two possible baselines for the {\problem} problem.}}
  \label{fig: current_framework}
\end{figure}

\noindent\textbf{Similarity Measurement Error.} Unless stated otherwise, we compute the similarity between vectors using the \textit{IP} metric, and all vectors are normalized. For a given object $o$ in $\mathcal{S}$ and a query input $q$, we calculate the \textit{IP} between their corresponding vector representations $\phi_i(q^i)$ and $\phi_i(o^i)$ as follows:
\begin{equation}
  \label{equ: sm_p}
  {IP(\phi_i(q^i),\phi_i(o^i)) = \phi_i(q^i) \odot \phi_i(o^i)\quad,}
\end{equation}
where $\odot$ denotes element-wise multiplication. The value of $IP(\phi_i(q^i),\phi_i(o^i))$ $\in$ $[0,1]$ indicates the similarity between $o$ and $q$ in the $i$-th modality. A higher value signifies a greater similarity.
For the composition vector $\Phi(q^0, \cdots, q^{t-1})$, we compute the \textit{IP} w.r.t the target modality by
\begin{equation}
  \label{equ: sm_pa}
  {IP(\Phi(q^0,\cdots,q^{t-1}),\phi_0(o^0)) = \Phi(q^0,\cdots,q^{t-1}) \odot \phi_0(o^0)\quad.}
\end{equation}
Current embedding methods can ensure that $\Phi(q^0,\cdots,q^{t-1})$ and $\phi_0(o^0)$ share the same vector space \cite{vo2019composing, baldrati2022conditioned}. Eq. \ref{equ: sm_pa} illustrates the similarity between the query $q$ and the target modality content of object $o$. The query result is the object whose target modality content exhibits the highest similarity to $q$. In our assumption of exact vector search, which entails no errors in the similarity computation, the query accuracy is solely dependent on the similarity measurement error ($SME$). For a result object $r$ and the ground-truth result $a$ w.r.t $q$, the $SME$ is computed as follows:
\begin{equation}
  \label{equ: sme}
  {SME(a,r) = 1-IP(\phi_0(a^0),\phi_0(r^0))\quad.}
\end{equation}
The $SME$ reflects the encoder loss and how well the exact vector search can retrieve the ground-truth object.

\noindent\textbf{Baseline 1: Multi-streamed Retrieval ({\bi}).} As depicted in Fig. \ref{fig: current_framework} (upper left), {\bi} divides the {\problem} into $t$ separate sub-queries, each focusing on a different modality. These individual sub-queries are processed independently to obtain candidate sets of potential results for each modality. To achieve this, {\bi} uses a unimodal encoder $\phi_{i}(\cdot)$ to embed each query element $q^i$ into a vector space, resulting in the feature vector $\phi_{i}(q^i)$ \cite{kennedy2008query}. Subsequently, it builds $m$ vector indexes on $\mathcal{S}$ for all modalities and performs $t$ separate vector search procedures \cite{tautkute2019deepstyle}. Finally, it merges all candidates from individual sub-queries and returns the final results \cite{zagoris2010mmretrieval}. This framework is a common practice in research related to DB and IR \cite{Milvus_sigmod2021, ADBV, wang2022navigable}, and it efficiently handles hybrid queries by effectively merging multiple constraints {\em with known importance}, making it a possible baseline to address \problem.

In hybrid queries for vector similarity search with attribute constraints \cite{Milvus_sigmod2021, ADBV}, the attribute holds higher importance than the feature vector. This allows for straightforward candidate merging by identifying objects that (1) match the attribute of the query and (2) are more similar to the feature vector of the query. However, in {\problem}, the importance of each modality is unknown, making it challenging to directly merge candidates based on their importance.
We take the intersection of all candidates as the final results in {\problem}, and further optimize this framework by replacing $\phi_{0}(q^0)$ with $\Phi(q^0,\cdots,q^{t-1})$ obtained from the joint embedding.

\noindent\textbf{Baseline 2: Joint Embedding ({\bii}).} {\bii} leverages the advancements in multimodal representation learning \cite{vo2019composing,CLIP2021} to address the {\problem} problem. It processes the multimodal query by fusing the target modality and auxiliary modality inputs into a single query vector. In Fig. \ref{fig: current_framework} (upper right), both the target modality and auxiliary modality inputs are jointly embedded to create a unified vector representation $\Phi(q^0,\cdots,q^{t-1})$.
Once the composition vector is obtained, vector search is performed on the vector index constructed from the target modality vectors $\{\phi_0(o^0)|o\in \mathcal{S} \}$. Recently, the CV community has extensively explored multimodal encoders, leading to the design of various joint embedding networks that aim to enhance the quality of embeddings. For instance, TIRG (Text-Image Residual Gating) \cite{vo2019composing} employs a gating-residual mechanism to fuse multiple modalities effectively. Another notable work \cite{baldrati2022conditioned} introduces a combiner network that combines features from multiple modalities, derived from the OpenAI CLIP network \cite{CLIP2021}. However, these existing multimodal encoders also fail to capture the importance of different modalities, which is crucial in solving \problem.

\noindent\textbf{Summary.} Both baselines employ vector search to efficiently retrieve query results, as depicted in Fig. \ref{fig: current_framework}. However, they differ in how they process the embedding of the query $q$, resulting in distinct implementations of the vector search procedure in the context of {\problem}. For detailed explanations of the high-dimensional vector search, please refer to Appendix \ref{appendix: vector search}.

\section{Proposed {\name} Framework: An Overview}
\label{sec: framework overview}

\begin{figure}[!tb]
  \centering
  \setlength{\abovecaptionskip}{0.1cm}
  \setlength{\belowcaptionskip}{-0.6cm}
  \includegraphics[width=\linewidth]{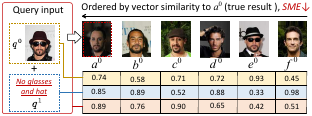}
  \caption{A face retrieval example on CelebA with image and text modalities \cite{celeba}. $a$–$f$ are the returned objects by different methods, only the target modality is shown. The table shows the \textit{IP} between different query vectors and the upper face vectors. Three query vectors are $\phi_0(q^0)$, $\phi_1(q^1)$, and $\Phi(q^0,q^1)$, corresponding to the yellow, blue, and red rows, respectively.}
  \label{fig: similarity_error}
\end{figure}

\begin{figure*}[!tb]
  \centering
  \setlength{\abovecaptionskip}{0cm}
  \setlength{\belowcaptionskip}{-0.6cm}
  \includegraphics[width=\linewidth]{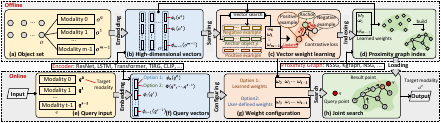}
  \caption{Overview of the working flow of {\name}.}
  \label{fig: our_framework}
\end{figure*}

In the following, we discuss the limitations of the two baselines, {\bi} and {\bii}, in addressing {\problem}. These baselines, while effective in other contexts, face challenges when it comes to handling the importance of different modalities in the multimodal query. We will highlight their shortcomings and propose a new framework, \underline{\textbf{Mu}}ltimodal \underline{\textbf{S}}earch of \underline{\textbf{T}}arget Modality ({\name}), that overcomes these issues. Our new approach aims to adapt to the varying significance of different modalities, enabling more accurate and context-aware multimodal search results in {\problem}.

\begin{myExa}
\label{example: sme}
In Fig. \ref{fig: similarity_error}, we illustrate an example on a real-world face dataset, CelebA \cite{celeba}. The query input $q$ consists of a reference face $q^0$ and a textual constraint $q^1$. The image $a^0$ is the ground-truth face that perfectly matches $q$, and $b^{0}$–$f^{0}$ are other candidate images\footnote{We use the ResNet \cite{he2016deep} to encode these images and calculate their $SME$ w.r.t. $a^0$. We also transform the textual constraint and image-text pair into vectors using the Encoding \cite{wang2015exploring} and CLIP \cite{baldrati2022conditioned}, respectively.}. These images are displayed in $SME$-descending order from right to left, indicating the increasing similarity with $a^0$.  
Additionally, the table in Fig. \ref{fig: similarity_error} presents the \textit{IP} values between different query vectors and the face vectors. For example, the first column of the table provides the \textit{IP} values between $\phi_0(q^0)$ and $\phi_0(a^0)$, $\phi_1(q^1)$ and $\phi_1(a^1)$, and $\Phi(q^0,q^1)$ and $\phi_0(a^0)$, from top to bottom.

The two baselines, {\bi} and {\bii}, differ primarily in how they encode and utilize the multimodal query $q$ and the object set $\mathcal{S}$—employing early or late fusion, respectively.
In {\bi}, $q^0$ and $q^1$ are separately encoded. Fig. \ref{fig: similarity_error} shows that it returns two top-$1$ faces, $e$ and $f$, by searching for the most similar vectors concerning image and text vectors, respectively. However, these images are dissimilar to the ground-truth $a^0$ due to incomplete query intent. Even considering the intersection of the top-$3$ candidate sets $\{ e, a, d \}$ for images and $\{ f, b, d \}$ for text, it still returns image $d$ instead of $a$, indicating the limitation of late fusion.
For {\bii}, $q^0$ and $q^1$ are combined into a composition vector $\Phi(q^0,q^{1})$. By searching for the most similar vector to $\Phi(q^0,q^{1})$ on $\{o^0|o\in \mathcal{S}\}$, it erroneously returns the face $c$, despite using the most advanced joint embedding technique \cite{baldrati2022conditioned}. Even trying to obtain two top-$3$ candidate sets $\{ c, a, b \}$ and $\{ f, b, d \}$ by searching for the closest vectors to $\Phi(q^0,q^{1})$ on $\{o^0|o\in \mathcal{S} \}$ and $\phi_1(q^1)$ on $\{o^1|o\in \mathcal{S} \}$, respectively, and then merging them leads to suboptimal results (please refer to \textbf{\S \ref{subsec: accuracy evaluation}} for evaluation). This approach combines different fusion stages from the two baselines but still fails to account for the importance of different modalities, limiting accuracy and efficiency due to the merging operation.
This motivates us to explore more sophisticated and comprehensive ways to fuse multiple modalities, considering the varying importance of different modalities in the multimodal query.
\end{myExa}

{\name} is a novel framework tailored for addressing the {\problem}, offering high accuracy, efficiency, and scalability. Unlike the two baselines discussed earlier, {\name} adopts a more comprehensive approach by incorporating three {\em pluggable} components that {\em fuse multiple modalities at different levels}. This allows {\name} to leverage all available modality information, taking advantage of the complementary benefits offered by different fusion levels. By doing so, {\name} accurately captures the importance of different modalities while efficiently executing joint search operations on a fused index.
Fig. \ref{fig: our_framework} provides a high-level overview of the {\name} framework, showcasing its key components, as elaborated below:

\noindent\textbf{Embedding.} As depicted in Fig. \ref{fig: our_framework} (left), the {\name} framework offers remarkable flexibility in representing objects or queries using multiple high-dimensional vectors obtained from various unimodal or multimodal encoders. It can seamlessly accommodate any encoder for any combination of modalities, such as using LSTM \cite{greff2016lstm} for text, ResNet \cite{he2016deep} for images, CLIP \cite{baldrati2022conditioned} for text-image pairs, and more. 
For an object set $\mathcal{S}$, {\name} transforms each object $o$ into $m$ vectors from $m$ different modalities, and each query $q$ into $t$ query vectors. Notably, {\name} allows for flexible encoding of the target modality input. It can be encoded independently (Option 1 in Fig. \ref{fig: our_framework}(f)), or fused with other modalities using a multimodal encoder (Option 2 in Fig. \ref{fig: our_framework}(f)).
By default, $\Phi(q^0,\cdots,q^{t-1})$ is represented in the same vector space as $\phi_0(q^0)$ \cite{vo2019composing}, ensuring compatibility and coherence within the framework. This adaptability empowers {\name} to effectively handle diverse multimodal scenarios, making it a powerful and versatile solution for addressing the {\problem}.

\noindent\textbf{Vector Weight Learning.} Innovatively, {\name} introduces a vector weight learning model that discerns the importance of different modalities for similarity measurement between objects. Considering a pair of objects $p$ and $o$, {\name} assigns specific weights to the vector spaces of each modality, effectively adjusting the influence of each vector. As illustrated in Fig. \ref{fig: our_framework}(c), {\name} incorporates the weight $\omega_i$ into $\phi_i(p^i)$ to create a virtual anchor (colored green). This virtual anchor is represented by a concatenated vector $\boldsymbol{\hat{p}}=[ \omega_{0}\cdot \phi_0(p^0),\cdots, \omega_{m-1}\cdot \phi_{m-1}(p^{m-1}) ]$. Similarly, {\name} generates the virtual point for object $o$ by $\boldsymbol{\hat{o}}=[ \omega_{0}\cdot \phi_0(o^0),\cdots, \omega_{m-1}\cdot \phi_{m-1}(o^{m-1}) ]$. The joint similarity between $p$ and $o$ is then computed by the \textit{IP} between $\boldsymbol{\hat{p}}$ and $\boldsymbol{\hat{o}}$.
To achieve the learning of vector weights, {\name} utilizes vector search to identify negative examples that share a high joint similarity with $p$. By employing a contrastive loss function, {\name} moves the virtual anchor away from the virtual points of negative examples and closer to the virtual point of the positive example. Through this process, the weights are adaptively adjusted to reflect the relative importance of different modalities in the similarity measurement.
Ultimately, {\name} outputs the learned weights, which can be effectively used for indexing and search operations

\noindent\textbf{Indexing and Searching.} {\name} constructs a fused proximity graph index based on the joint similarity between objects in the object set $\mathcal{S}$. The weights of different modalities, acquired from the model shown in Fig. \ref{fig: our_framework}(c), are utilized in this process. For a query input $q$ with $t$ query vectors, {\name} employs a merging-free joint search procedure to find the ground-truth object on the fused index. In the fused index $G=(V,E)$ (Fig. \ref{fig: our_framework}(d)), the objects in $\mathcal{S}$ correspond to vertices in $V$, and edges in $E$ represent similar object pairs in terms of their joint similarity. When processing a query input $q$, {\name}'s search procedure initiates from either a random or fixed vertex (e.g., $g$ in Fig. \ref{fig: our_framework}(h)) and explores neighboring vertices in $G$ that are closer to $q$ (e.g., $e,c,a$). The procedure continues until it reaches a vertex that has no neighbors closer to $q$ than itself (e.g., $a$). Throughout this procedure, {\name} calculates the distance of vertices from $q$ using joint similarity.
Regarding the weight options, {\name} provides two choices: (1) learned weights obtained from the offline model (Option 1 in Fig. \ref{fig: our_framework}(g)), and (2) user-defined weights (Option 2 in Fig. \ref{fig: our_framework}(g)). This flexibility allows users to either leverage weights learned from the vector weight learning model or manually specify their own weights for a more customized search experience.

\begin{myExa}
\label{example: our framework}
In the face retrieval example shown in Fig. \ref{fig: similarity_error}, our vector weight learning model outputs the weights $\omega_0=0.80$ and $\omega_1=0.33$ for the two modalities. Leveraging these learned weights, we compute the joint similarity between the query $q$ and the candidate objects. The concatenated vector representation of $q$ is computed as $\boldsymbol{\hat{q}}=[\omega_0 \cdot \Phi(q^0,q^1), \omega_1 \cdot \phi_1(q^1)]$. By using this joint similarity computation, we find that object $a$ has the highest joint similarity to the query $q$ compared to the other candidates. According to Lemma \ref{lemma: joint similarity}, we calculate $IP(\boldsymbol{\hat{q}}, \boldsymbol{\hat{a}})$=0.6622. As a result, {\name} achieves a significantly improved query result by effectively capturing the importance of different modalities and accurately evaluating the joint similarity between objects.
\end{myExa}
\section{Embedding}
\label{sec: embedding}
Deep representation learning has revolutionized the use of various encoders to transform information into high-dimensional vectors, benefiting different downstream tasks \cite{Milvus_sigmod2021}. Traditional encoders represent objects using single vectors from individual modalities. For example, ResNet \cite{he2016deep} encodes face images into vectors. However, recent progress in multimodal learning has introduced encoders that can fuse multiple modalities, such as the CLIP model, which can embed both face and text as a unified vector \cite{jandial2022sac}. Despite these advancements, research indicates that a single-vector representation may be inadequate for unimodal encoders, capturing only partial object information \cite{Milvus_sigmod2021}, and may introduce significant encoder errors for multimodal encoders \cite{DelmasRCL22}. Our experiments confirm that relying on a single-vector representation leads to notably low search accuracy (see Tab. \ref{tab: accuracy mitstates}–\ref{tab: accuracy coco}).

In {\name}, we introduce a novel multi-vector representation method for multimodal objects and queries (refer to Fig. \ref{fig: our_framework}(b) and (f)). This approach generates distinct vector representations for different modalities of an object or query. Importantly, {\name} offers flexibility in encoding the target modality input. It can either be independently encoded (Option 1 in Fig. \ref{fig: our_framework}(f)) or fused with other modalities using a multimodal encoder (Option 2 in Fig. \ref{fig: our_framework}(f)). By default, $\Phi(q^0,\cdots,q^{t-1})$ is represented in the same vector space as $\phi_0(q^0)$ \cite{vo2019composing}, ensuring compatibility and coherence within the framework.

This approach enables us to describe an object or query comprehensively using multiple vectors, resulting in strong generalization capabilities for {\name}. In scenarios where multimodal query input is unavailable, users can still perform conventional single-modal search initially and then achieve improved results through {\problem}. Additionally, the embedding component in {\name} is pluggable, allowing seamless integration of any newly-devised encoders into the system. Further details about the encoders are provided in Appendix \ref{appendix: encoders}.
\section{Vector Weight Learning}
\label{sec: vector weight learning}

In {\name}, we combine all vectors of an object using a set of weights to form a {\em concatenated} vector. These weights serve as indicators of the importance of different modalities in representing the object. By doing so, we achieve the mapping of each object into a {\em unified high-dimensional vector space}, facilitating similarity computation between object pairs through the Inner Product (\textit{IP}) of their concatenated vectors. To determine these weights, we introduce a lightweight vector weight learning model based on {\em contrastive learning}.
To begin, given an anchor object, we acquire its positive and negative examples. Subsequently, we construct a contrastive loss function and minimize it to learn the relative weights. The training pipeline of the model is depicted in Fig. \ref{fig: our_framework}(c).

\subsection{Positive and Negative Examples}

The training data consists of two parts: the anchor set $Q$ (i.e., queries) and a set of their true resultant objects $T$. For each anchor $p\in Q$, there is a corresponding true object in $T$. Positive and negative examples for $p$ are created as follows.

\noindent\textbf{Positive Example}. In $T$, the true  object corresponding to the anchor $p$ is directly assigned as a positive example $p^{+}$.

\noindent\textbf{Negative Examples}. We focus on identifying hard negative examples that are easily confused with the true object of $p$. Using a weight combination ${\omega_0, \omega_1, \cdots, \omega_{m-1} }$, we map $p$ and objects in $T$ into a unified vector space (shadow region in Fig. \ref{fig: our_framework}(c)). In this space, we generate virtual anchor and object points based on their concatenated vectors. Then, through vector search, we obtain the top-$k$ result objects denoted by a set $R$ with the highest similarity to $p$. $R$ is defined as follows:
\begin{equation}
  \label{equ: negative_example}
  R=\arg \max_{R\subset T\wedge |R|=k} \sum_{r \in R} IP(\boldsymbol{\hat{r}},\boldsymbol{\hat{p}})\quad,
\end{equation}
where $\boldsymbol{\hat{r}}=[ \omega_{0}\cdot \phi_0(r^0),\cdots, \omega_{m-1}\cdot \phi_{m-1}(r^{m-1}) ]$ and $\boldsymbol{\hat{p}}=[ \omega_{0}\cdot \phi_0(p^0),\cdots, \omega_{m-1}\cdot \phi_{m-1}(p^{m-1}) ]$ are the concatenated vectors of $r$ and $p$, respectively. We designate false objects in $R$ as negative examples, denoted by $N^{-}=R\setminus \{p^{+}\}$.

\subsection{Loss Function}
Our training objective is to push the virtual anchor away from the virtual points of objects in $N^{-}$ and pull it closer to the virtual point of the object $p^+$. To achieve this, we devise a loss function $L$ based on the well-known contrastive loss \cite{RobinsonCSJ21}. Let $Q$ be a training minibatch of $M$ anchors, and we define the loss function as follows:
\begin{equation}
  \label{equ: loss}
  L=\frac{1}{M} \sum_{p\in Q}-\log \frac{e^{IP(\boldsymbol{\hat{p}},\boldsymbol{\hat{p}^+})}}{e^{IP(\boldsymbol{\hat{p}},\boldsymbol{\hat{p}^+})}+\sum_{p^-\in N^-}e^{IP(\boldsymbol{\hat{p}},\boldsymbol{\hat{p}^-})}}\quad.
\end{equation}
We aim to minimize $L$ to learn the relative weights, starting with a random initialization of weights ${\omega_0, \omega_1, \cdots, \omega_{m-1} }$. These weights are used to compute concatenated vectors of an anchor $p$ and its positive and negative examples, with negative examples obtained through vector search under the current weights. The top-$k$ result objects $R$ are obtained using Eq. \ref{equ: negative_example}. If the positive example $p^+$ is not in $R$, and for all $p^{-}\in N^-$, it holds that $IP(\boldsymbol{\hat{p}},\boldsymbol{\hat{p}^-}) > IP(\boldsymbol{\hat{p}},\boldsymbol{\hat{p}^+})$, the loss $L$ is significant. To minimize this loss, we update the weights using gradient descent in a way that increases $e^{IP(\boldsymbol{\hat{p}},\boldsymbol{\hat{p}^+})}$ while decreasing $\sum_{p^-\in N^-}e^{IP(\boldsymbol{\hat{p}},\boldsymbol{\hat{p}^-})}$.
This weight update encourages the positive example to have a higher \textit{IP} w.r.t. $p$, while pushing the negatives to have a lower \textit{IP}. Subsequently, we can obtain new negatives using the updated weights and continue the weight optimization process. We eventually arrive at a set of  learned weights, under which the true object is more likely to be retrieved as the top result in the search process.
We have the following lemma:

\begin{myLemma}
\label{lemma: joint similarity}
The joint similarity of an object pair is the weighted sum of the similarity of each modality.
\end{myLemma}

\begin{proof}
For two objects $a$ and $b$, their concatenated vectors are $\boldsymbol{\hat{a}}$ and $\boldsymbol{\hat{b}}$. 
The \textit{IP} between $\boldsymbol{\hat{a}}$ and $\boldsymbol{\hat{b}}$ can be computed by 
\begin{equation}
\begin{split}
    IP(\boldsymbol{\hat{a}},\boldsymbol{\hat{b}}))=\boldsymbol{\hat{a}} \odot \boldsymbol{\hat{b}}
    & =\sum_{i=0}^{m-1}\omega_{i}^2 \cdot IP(\phi_i(a^i),\phi_i(b^i))\quad,
\end{split}
\end{equation}
where $IP(\phi_i(a^i),\phi_i(b^i))$ indicates the similarity between $a$ and $b$ in the $i$-th modality.
\end{proof}

The weight learning process described above plays a crucial role in capturing the significance of different modalities for representing objects. By learning the relative weights, we can effectively incorporate information from multiple modalities into the similarity computation. The search process then benefits from a holistic view of object representations, considering the diverse user intentions captured by different modalities. In the face retrieval case (Example \ref{example: our framework}) and our experimental studies (e.g., Fig. \ref{fig: mitstates_case}), this more comprehensive representation of objects leads to more meaningful and precise search results.

\subsection{Generalization Analysis of Weight}
\label{subsec: weight analysis}

The weight-learning component in our approach eliminates the need for specific weights for each query input. This is achieved by {\em learning query-independent weights that are associated with the modalities themselves}, rather than the specific content within each modality. As a result, we can employ a fixed set of weights to compute the joint similarity between any query and object in the dataset.
Consider two extreme query cases on a dataset containing image and text modalities. In \underline{Case 1}, the text describes what is already present in the given image, while in \underline{Case 2}, the text describes something not depicted in the given image. In both scenarios, our system, {\name}, consistently embeds the image with the text semantics using Option 2 in Fig. \ref{fig: our_framework}(f), while also separately embedding the text semantics using a unimodal encoder. For any object $o$, {\name} computes the joint similarity between $o$ and both types of queries using the same weights (refer to the {\em Learned Weights} section in {\S \ref{subsec: query requirement}}). The similarity value is determined by the inner product (\textit{IP}) between the modalities, as stated in Lemma \ref{lemma: joint similarity}.
Indeed, capturing the differences between image and text contents can be effectively achieved using specific vectors rather than weights, allowing us to represent the unique characteristics of each modality while avoiding the impracticality of assigning individual weights to each object in large-scale scenarios. By employing weights to capture the importance of different modalities instead of contents, we achieve better generalization across various query workloads.

In {\name}, users have the option to use custom weights for specific purposes, such as giving more weight to the text modality. In this case, the learned weights can be replaced by user-defined weights, which would prioritize objects that are more similar to the emphasized modality. The evaluation of this option is provided in Tab. \ref{tab: query type user-defined weights}. Note that assigning proper weights manually can be challenging in practice. Based on our experiments (\textbf{\S~\ref{subsec: ablation}}), we observe that different weights significantly affect the recall rate of {\problem} (Fig. \ref{fig: ablation weight learning}).
\section{Indexing and Searching}
\label{sec: indexing and search algorithms}

To address the efficiency and scalability challenges associated with enumerating potential objects, {\name} adopts an approximate method that balances accuracy and efficiency. This involves constructing a fused index based on the similarity of concatenated vectors. Specifically, we utilize a proximity graph index \cite{graph_survey_vldb2021}, which is a sota method in the vector search domain\footnote{{Please refer to Appendix \ref{appendix: pg-based index} for detailed related work discussion about vector search}.}.
In the fused index, $G=(V,E)$, each vertex $v \in V$ represents an object $v \in \mathcal{S}$\footnote{We use the same symbol for an object and its corresponding vertex.}, and each edge $(v,u) \in E$ captures a closely related object pair $(v,u)$ via joint similarity. 
The index can reduce the search space and navigate us to a true object by visiting only a few objects in $\mathcal{S}$, leading to better efficiency.
We further improve indexing and search performance by re-assembling index components and optimizing the multi-vector computations, respectively. This ensures that our system remains efficient in handling large-scale datasets.

\subsection{Index Construction}
\label{subsec: index construction}
We present a general pipeline (Algorithm \hyperref[alg: index]{1}) for constructing fine-grained proximity graphs on CGraph\footnote{CGraph refers to the Directed Acyclic Graph framework \cite{CGraph}.}. The pipeline is composed of five flexible components (\ding{172}–\ding{176}). By decomposing any current proximity graph \cite{graph_survey_vldb2021} into these components, we can seamlessly integrate them into our pipeline\footnote{In our evaluations, we implemented some representative proximity graph algorithms, which are detailed in \textbf{\S \ref{subsec: ablation}}.}. 
Furthermore, to enhance the capabilities of our pipeline, we amalgamate components from several state-of-the-art algorithms in the context of concatenated vectors, culminating in the creation of a new indexing algorithm.

\setlength{\textfloatsep}{0pt}
\begin{algorithm}[t]
\label{alg: index}
  \caption{\textsc{Construct Fused Index}}
  \LinesNumbered
  \KwIn{Object set $\mathcal{S}$, maximum number of neighbors $\gamma$, maximum iterations $\varepsilon$}
  \KwOut{{Fused Index} ${G}=({V},{E})$ and seed vertex $g$}
  
  ${V} \gets \mathcal{S}$; ${E} \gets \emptyset $
  
  \ForAll(\tcc*[f]{\ding{172}}){$o \in {V}$}{
    $N(o) \gets$ $\gamma$ random objects {\Comment{\texttt{\textcolor{blue}{Neighbor set}}}}
  }
  \While({\Comment{\texttt{\textcolor{blue}{NNDescent}}}}){iterations $\leq$ $\varepsilon$}{
  \ForAll{$o$$\in$${V}$ and $v$$\in$$N(o)$ and $u$$\in$ $N(v) \setminus N(o)$}{
        $z\gets \arg\min_{z\in N(o)}IP(\boldsymbol{\hat{o}},\boldsymbol{\hat{z}})$
        
        \If{$IP(\boldsymbol{\hat{o}},\boldsymbol{\hat{u}})> IP(\boldsymbol{\hat{o}},\boldsymbol{\hat{z}})$}{
          $N(o)\gets N(o)\setminus \{z\} \cup \{u\}$
        }
  }
  }
  \ForAll({\tcc*[f]{\ding{173}}}){$o$ $\in$ ${V}$ and $v$ $\in$ $N(o)$}{
    $C(o) \gets N(o) \cup N(v)$ {\Comment{\texttt{\textcolor{blue}{Candidate neighbor}}}}
  }
  
  \ForAll({\tcc*[f]{\ding{174}}}){$o \in {V}$}{
    $v$$\gets$$\arg\max_{v\in C(o)}IP(\boldsymbol{\hat{o}},\boldsymbol{\hat{v}})$; $N(o)$$\gets$$N(o)\cup \{v\}$
    
    \While{$C(o)\neq \emptyset$ and $|N(o)|< \gamma$}{
     $v$$\gets$$\arg\max_{v\in C(o)}IP(\boldsymbol{\hat{o}},\boldsymbol{\hat{v}})$; $C(o)$$\gets$$C(o)$$\setminus$$\{v\}$
    
    \ForAll({\Comment{{\texttt{\textcolor{blue}{MRNG strategy\cite{NSG}}}}}}){$u\in N(o)$}{
      \If{$IP(\boldsymbol{\hat{o}},\boldsymbol{\hat{v}}) > IP(\boldsymbol{\hat{u}},\boldsymbol{\hat{v}})$}{
        $N(o) \gets N(o) \cup \{v\}$
      }
    }
  }
  }
  
  $g\gets$ nearest vertex to $\frac{1}{|V|}\sum_{o\in V}\boldsymbol{\hat{o}}$ {\tcc*[f]{\ding{175}}}
  
  Ensure connectivity by BFS from $g$ {\tcc*[f]{\ding{176}}}
  
  return ${G}=({V}, {E})$ and $g$ {\Comment{{\texttt{\textcolor{blue}{$E=\bigcup_{o\in V} N(o)$}}}}}
\end{algorithm}

\noindent\textbf{\ding{172} Initialization.} This component is responsible for generating the initial neighbors for each object in $\mathcal{S}$. We start with randomly selecting a set of objects as neighbors $N(o)$ for any given object $o \in \mathcal{S}$ (Lines 2-3). The neighbor set $N(o)$ is then updated iteratively by visiting the neighbors $N(v)$ of each object $v$ in $N(o)$ (Lines 4-8).
During the update process, for each object $u$ in $N(v)$ where $u \notin N(o)$, we find the object $z$ that minimizes $IP(\boldsymbol{\hat{o}},\boldsymbol{\hat{z}})$.
If $IP(\boldsymbol{\hat{o}},\boldsymbol{\hat{u}}) \geq IP(\boldsymbol{\hat{o}},\boldsymbol{\hat{z}})$, we replace $z$ with $u$ in $N(o)$. This iterative process is performed for all objects in $\mathcal{S}$ to create a high-quality initial graph.
The evaluation conducted indicates that only three iterations are sufficient to achieve a graph quality of over 90\% (for detailed evaluation, refer to {Appendix \ref{appendix: index parameters}}).

\noindent\textbf{\ding{173} Candidate Acquisition.} This component obtains some candidate neighbors $C(o)$ for each vertex $o$ in $V$ from the initial graph. These candidates will serve as the potential final neighbors.
For each vertex $o$ in $V$, we get $C(o)$ by combining $o$'s initial neighbors and their neighbors (Lines 9-10).

\noindent\textbf{\ding{174} Neighbor Selection.} In this component, we apply a filtering process to the candidate neighbors $C(o)$ and carefully select the final neighbors $N(o)$ for each vertex $o$ in $V$. The primary objective is to diversify the distribution of neighbors, which is essential for ensuring search efficiency.
To achieve this, we employ the MRNG strategy \cite{NSG} (Lines 11-17). For each vertex $o$ in $V$, we first clear its set of final neighbors $N(o)$, then extract the vertex $v$ that is closest to $o$ from the candidate set $C(o)$ and include it in $N(o)$. Subsequently, we iteratively select vertices from $C(o)$ that are closest to $o$ and satisfy the condition $IP(\boldsymbol{\hat{o}},\boldsymbol{\hat{v}}) > IP(\boldsymbol{\hat{u}},\boldsymbol{\hat{v}})$ for all $u$ in the current set $N(o)$. If this condition is met for a vertex $v$, we add it to $N(o)$. This selection process ensures a diversified distribution of neighbors (as demonstrated in \textbf{Lemma \ref{lemma: neighbor selection}}).
Notably, this approach has been widely acknowledged in the literature \cite{HNSW, NSG, DPG}, and our experiments further corroborate its effectiveness in the context of {\problem} (\textbf{\S \ref{subsec: efficiency evaluation}}).
The parameter $\gamma$ is carefully evaluated, and additional details regarding its assessment can be found in Appendix \ref{appendix: index parameters}.

\begin{myLemma}
\label{lemma: neighbor selection}
For any two neighbors $u$ and $v$ in $N(o)$, the angle $\angle uov$ (denoted by $\theta(u,o,v)$) is at least $60^{\circ}$.
\end{myLemma}

\begin{proof}
\textit{(Sketch.)} 
Assuming $\theta(u,o,v) < 60^\circ$ for two neighbors $u$ and $v$ in $N(o)$, we find that the sum of the angles $\theta(o,v,u)$ and $\theta(o,u,v)$ exceeds $120^\circ$ in the triangle $\triangle uov$. Here, the inner product (\textit{IP}) of two vertices is used to measure the side length between them, and smaller \textit{IP} values imply longer sides. By comparing the \textit{IP}, either $\theta(o,v,u) > \theta(o,u,v)$ (i.e., $\theta(o,v,u) > 60^\circ$) or $\theta(o,u,v) > \theta(o,v,u)$ (i.e., $\theta(o,u,v) > 60^\circ$). \underline{Case 1:} if $\theta(o,v,u) > \theta(o,u,v)$, it implies $IP(\boldsymbol{\hat{o}},\boldsymbol{\hat{v}}) > IP(\boldsymbol{\hat{o}},\boldsymbol{\hat{u}})$, resulting in vertex $v$ being added to $N(o)$ before $u$. Consider the assumption, we have $IP(\boldsymbol{\hat{u}},\boldsymbol{\hat{v}}) > IP(\boldsymbol{\hat{o}},\boldsymbol{\hat{u}})$. Therefore, vertex $u$ cannot be added to $N(o)$. \underline{Case 2:} If $\theta(o,u,v) > \theta(o,v,u)$, we can swap the positions of $u$ and $v$ in $\triangle uov$ and arrive at the same conclusion as in Case 1. We put the detailed proof in Appendix \ref{appendix: proof of lemma 2}.
\end{proof}

\noindent\textbf{\ding{175} Seed Preprocessing.} We select a fixed seed as a start vertex for searching of different queries. We first compute the centroid of all vertices in $V$ with their concatenated vectors. We then compute the \textit{IP} between each vertex and centroid to find the vertex closest to the centroid as the seed (Line 18).

\noindent\textbf{\ding{176} Connectivity.} We perform a breadth-first search (BFS) from the seed. In case the BFS cannot reach all vertices in $V$ from the seed, a connection is established between a visited vertex and an unvisited vertex. This connection bridges the gap between previously unexplored regions of the graph and the BFS is continued until all vertices are reachable from the seed (Line 19), thereby enhancing the search accuracy.

\subsection{Joint Search}
\label{subsec: joint search}
Upon receiving a multimodal query input $q$, {\name} conducts a joint search across all modalities using the fused index. Initially, $q$ is transformed into $t$ query vectors (Fig. \ref{fig: our_framework}(f)) and concatenated with a set of weights to be a virtual query point (Fig. \ref{fig: our_framework}(g)). When $t=m$, $q$ is mapped into the same vector space as the objects in $\mathcal{S}$, enabling the computation of the inner product (\textit{IP}) between $q$ and the objects in $\mathcal{S}$ based on Lemma \ref{lemma: joint similarity}. However, if $t\neq m$, the concatenated vectors compute the \textit{IP} by setting $\omega_i=0$ for $t\leq i \leq m-1$.
Next, the search process begins at the seed and employs greedy routing within the fused index to obtain approximate top-$k$ results (Fig. \ref{fig: our_framework}(h)).

\begin{algorithm}[t]
\label{alg: search}
  \caption{\textsc{Joint Search}}
  \LinesNumbered
  \KwIn{Fused index $G=(V,E)$, multimodal query $q$, seed vertex $g$, number of results $k$, result set size $l$ ($>k$)}
  \KwOut{approximate top-$k$ results of $q$}
  
  ${R}$$\gets$$\{g\}$; ${H}$$\gets$$\emptyset$
  
  ${C}$ $\gets$ $l-1$ random vertices
  
  ${R}$ $\gets$ $R\cup C$ {\Comment{{\texttt{\textcolor{blue}{sorted by \textit{IP} to $q$}}}}}
  
  \While({{\Comment{{\texttt{\textcolor{blue}{unvisited vertices}}}}}}){($R$ $\setminus$ $H$) $\neq$ $\emptyset$}{
  $v\gets$ unvisited nearest vertex to $q$ in $R$ {\Comment{{\texttt{\textcolor{blue}{$v\notin H$}}}}}
  
  $H$ $\gets$ $H$ $\cup$ $\{v\}$ {\Comment{{\texttt{\textcolor{blue}{mark $v$ as visited}}}}}
  
  \ForAll{$u \in {N(v)}$ and $u\notin H$}{
  
    $z$ $\gets$ $\arg \min_{z\in R}IP(\boldsymbol{\hat{q}},\boldsymbol{\hat{z}})$
    
    \If{$IP(\boldsymbol{\hat{q}},\boldsymbol{\hat{z}})<IP(\boldsymbol{\hat{q}},\boldsymbol{\hat{u}})$}{
      $R\gets R\setminus \{z\} \cup \{u\}$ {\Comment{{\texttt{\textcolor{blue}{update $R$}}}}}
    }
  }
  }
  
  return top-$k$ nearest vertices in $R$
\end{algorithm}

Algorithm \hyperref[alg: search]{2} presents the joint search procedure of {\name}. During the greedy routing, two key data structures, $R$ and $H$, are utilized (Line 1). $R$ represents the result set with a fixed size of $l$ and is initialized with the seed vertex $g$ and $l-1$ randomly chosen vertices. On the other hand, $H$ is a set that keeps track of visited vertices, effectively avoiding redundant vector computations. 
The iterative greedy routing (Lines 4-10) selects unvisited vertices from $R$ closest to the query point $q$. It calculates the \textit{IP} between each neighbor of $v$ and $q$, updating $R$ accordingly. The process continues until all vertices in $R$ are visited, yielding the top-$k$ nearest vertices.
In practice, users have the flexibility to balance accuracy and efficiency by tuning the parameter $l$. The value of $l$ determines the size of the result set $R$ and influences the trade-off between accuracy and efficiency.
We conduct evaluations of $l$ in {Appendix \ref{appendix: search parameters}}. Given the number of iterations $\eta$, we have the following lemma to ensure the joint similarity is non-decreasing during searching:
\begin{myLemma}
The sum of the \textit{IP} between the query $q$ and the vertices in $R$ is a monotonically non-decreasing function of $\eta$, denoted by $f(\eta)$.
\end{myLemma}
\begin{proof}
In the joint search process of {\name}, let's consider any two consecutive iterations $\eta=i$ and $\eta=j$ ($i<j$). We denote $R_i$ and $R_j$ as the $R$ sets after the $i$-th and $j$-th iterations, respectively. Additionally, let $z$ be the vertex farthest from $q$ in $R_i$.
During the $j$-th iteration, we encounter two cases for any neighbor $u$ of the current visited vertex.
\underline{Case 1:} If $IP(\boldsymbol{\hat{q}},\boldsymbol{\hat{z}})\geq IP(\boldsymbol{\hat{q}},\boldsymbol{\hat{u}})$, $R_i$ remains unchanged, resulting in $f(i)=f(j)$.
\underline{Case 2:} If $IP(\boldsymbol{\hat{q}},\boldsymbol{\hat{z}})<IP(\boldsymbol{\hat{q}},\boldsymbol{\hat{u}})$, $z$ is replaced by $u$ in $R_j$, leading to $f(j)=f(i)-IP(\boldsymbol{\hat{q}},\boldsymbol{\hat{z}})+IP(\boldsymbol{\hat{q}},\boldsymbol{\hat{u}})$, which implies $f(i)<f(j)$.
Thus, we can deduce that $i<j$ implies $f(i) \leq f(j)$, demonstrating that $f(\eta)$ is a monotonically non-decreasing function of $\eta$.
\end{proof}

\noindent\textbf{Optimizing Multi-vector Computation.} In our approach, the joint search, particularly the multi-vector computation, constitutes the most time-consuming part. For each object pair, we must compute $m$ similarities between high-dimensional vectors. It is well-documented in the literature \cite{SONG, adsampling} that vector computation can consume up to 90\% of the total search time in many real-world datasets.
When processing an object $u$ (Line 7 in Algorithm \hyperref[alg: search]{2}), we need to compute the inner product of $\boldsymbol{\hat{u}}$ and $\boldsymbol{\hat{q}}$. The resulting inner product value is then used for the similarity comparison (Line 9 in Algorithm 2) with the most dissimilar object $z$ in $R$. If $u$ is more similar to $q$ than $z$, we update $R$ with $u$ based on this inner product value. However, if $u$ is less similar to $q$ than $z$, we can simply discard $u$. In this case, there is no need to compute the exact value of $IP(\boldsymbol{\hat{q}},\boldsymbol{\hat{u}})$. Since the vectors are normalized, we have 
\begin{equation}
  \label{equ: distance_relationship}
  IP(\boldsymbol{\hat{q}},\boldsymbol{\hat{u}})=1 - \frac{1}{2} \cdot ||\boldsymbol{\hat{q}},\boldsymbol{\hat{u}}||^2 \quad,
\end{equation}
where $||\boldsymbol{\hat{q}},\boldsymbol{\hat{u}}||$ is the Euclidean distance between $\boldsymbol{\hat{q}}$ and $\boldsymbol{\hat{u}}$. As the $||\boldsymbol{\hat{q}},\boldsymbol{\hat{u}}||^2$ increases, $IP(\boldsymbol{\hat{q}},\boldsymbol{\hat{u}})$ decreases. Therefore, we scan the vectors of $\boldsymbol{\hat{u}}$ incrementally and compute the partial square Euclidean distance based on the scanned $x$ vectors as
\begin{equation}
  \label{equ: partial_distance}
  \widetilde{||\boldsymbol{\hat{q}},\boldsymbol{\hat{u}}||}^2 = \sum_{i=0}^{x-1}\omega_{i}^2 \cdot ||\phi_i(u^i),\phi_i(q^i)||^2 \quad,
\end{equation}
where $||\phi_i(u^i),\phi_i(q^i)||$ is the Euclidean distance between the vectors in the $i$-th modality. Then, we can compute the partial \textit{IP} $\widetilde{IP(\boldsymbol{\hat{q}},\boldsymbol{\hat{u}})}$ by applying Eq. \ref{equ: partial_distance} to Eq. \ref{equ: distance_relationship}. We check whether $IP(\boldsymbol{\hat{q}},\boldsymbol{\hat{z}})\geq \widetilde{IP(\boldsymbol{\hat{q}},\boldsymbol{\hat{u}})}$, if it holds, we can discard $u$ immediately, otherwise, we continue to consider the next vector until scanning all vectors (i.e., $x=m$) or $\widetilde{IP(\boldsymbol{\hat{q}},\boldsymbol{\hat{u}})}$ is not more than $IP(\boldsymbol{\hat{q}},\boldsymbol{\hat{z}})$. As we will demonstrate in our experiment, this optimization significantly improves the search efficiency without incurring any accuracy loss (Lemma \ref{lemma: multi-vector computation}).

\begin{myLemma}
\label{lemma: multi-vector computation}
By utilizing the multi-vector computation optimization, we can safely discard the object $u$ that satisfies $IP(\boldsymbol{\hat{q}},\boldsymbol{\hat{z}})\geq IP(\boldsymbol{\hat{q}},\boldsymbol{\hat{u}})$ by the partial \textit{IP} $\widetilde{IP(\boldsymbol{\hat{q}},\boldsymbol{\hat{u}})}$. Furthermore, when $IP(\boldsymbol{\hat{q}},\boldsymbol{\hat{z}})<IP(\boldsymbol{\hat{q}},\boldsymbol{\hat{u}})$, we can obtain the exact value of $IP(\boldsymbol{\hat{q}},\boldsymbol{\hat{u}})$.
\end{myLemma}
\begin{proof}
According to Eq. \ref {equ: distance_relationship}, a larger value of $||\boldsymbol{\hat{q}}, \boldsymbol{\hat{u}} ||^2$ corresponds to a smaller value of $IP(\boldsymbol{\hat{q}}, \boldsymbol{\hat{u}})$. As we incrementally scan the vectors, $\widetilde {||\boldsymbol{\hat{q}}, \boldsymbol{\hat{u}}||}^2$ gradually increases while $\widetilde{IP (\boldsymbol{\hat{q}}, \boldsymbol{\hat{u}})}$ gradually decreases. Let $x$ be the number of scanned vectors. Once $IP (\boldsymbol {\hat{q}}, \boldsymbol {\hat{z}}) \geq \widetilde {IP (\boldsymbol{\hat{q}}, \boldsymbol{\hat{u}})}$ is true for the first time, it remains true for larger values of $x$. Therefore, we can safely terminate the multi-vector computation when $IP(\boldsymbol{\hat{q}},\boldsymbol{\hat{z}})\geq \widetilde{IP(\boldsymbol{\hat{q}},\boldsymbol{\hat{u}})}$. In the case where $IP(\boldsymbol{\hat{q}},\boldsymbol{\hat{z}})<IP(\boldsymbol{\hat{q}},\boldsymbol{\hat{u}})$, we have $IP(\boldsymbol{\hat{q}},\boldsymbol{\hat{z}})<\widetilde{IP(\boldsymbol{\hat{q}},\boldsymbol{\hat{u}})}$ for any value of $x$. Hence, in this case,  we scan all $m$ vectors and obtain the exact value of $IP(\boldsymbol{\hat{q}},\boldsymbol{\hat{u}})$.
\end{proof}
\section{Experiments}
\label{sec: experiments}
We thoroughly evaluate {\name} across six key aspects: (1) accuracy, (2) case study, (3) efficiency, (4) scalability, (5) query workloads, and (6) ablation studies. Kindly refer to our GitHub repository: https://github.com/ZJU-DAILY/MUST for our source code, datasets, and additional evaluations.

\setlength{\textfloatsep}{1.55\baselineskip plus 0.2\baselineskip minus 0.4\baselineskip}
\begin{table}[!tb]
  \centering
  \setlength{\abovecaptionskip}{0.05cm}
  \setstretch{0.8}
  \fontsize{6.5pt}{3.3mm}\selectfont
  \caption{Dataset statistics ($\star$ marks the target modality).}
  \label{tab: Dataset}
  \setlength{\tabcolsep}{.011\linewidth}{
  \begin{tabular}{|l|l|l|l|l|l|}
    \hline
    \textbf{Dataset} & \textbf{\# Modality} & \textbf{\# Object} & \textbf{\# Query} & \textbf{Type} & \textbf{Source}\\
    \hline
    CelebA \cite{celeba} & 2 & 191,549 & 34,326 & Image$^{\star}$,Text & real-world \\
    \hline
    MIT-States \cite{mit-states} & 2 & 53,743 & 72,732 & Image$^{\star}$,Text & real-world \\
    \hline
    Shopping \cite{shopping} & 2 & 96,009 & 47,658 & Image$^{\star}$,Text & real-world \\
    \hline
    MS-COCO \cite{neculai2022probabilistic} & 3 & 19,711 & 1237 & Image$^{\star}$ $\times$2,Text & real-world \\
    \hline
    CelebA+ \cite{celeba} & 4 & 191,549 & 34,326 & Image$^{\star}$ $\times$3,Text & real-world \\
    \hline
    ImageText1M \cite{sift1m} & 2 & 1,000,000 & 10,000 & Image$^{\star}$,Text & semi-synthetic \\
    \hline
    AudioText1M \cite{msong1m} & 2 & 992,272 & 200 & Audio$^{\star}$,Text & semi-synthetic \\
    \hline
    VideoText1M \cite{uqv1m} & 2 & 1,000,000 & 10,000 & Video$^{\star}$,Text & semi-synthetic \\
    \hline
    ImageText16M \cite{bigann} & 2 & 16,000,000 & 10,000 & Image$^{\star}$,Text & semi-synthetic \\
    \hline
  \end{tabular}
  }\vspace{-0.5cm}
\end{table}

\vspace{-0.2cm}
\subsection{Experimental Setting}
\noindent\textbf{Datasets.} We use nine datasets obtained from public sources, each with varying modalities and cardinalities, as shown in Tab. \ref{tab: Dataset}. Unless specified otherwise, the queries consist of the same number of modalities as the objects in each dataset (i.e., $t=m$). For more details, kindly refer to {Appendix \ref{appendix: datasets}}.

\noindent\textbf{Compared Methods.} We compare our proposed {\name} with two baselines: Multi-streamed Retrieval (abbr. {\bi}) and Joint Embedding (abbr. {\bii}). To ensure a fair comparison, we use the same encoders and proximity graph index in all competitors.

\noindent\textbf{Metrics.} We measure search accuracy for a batch of queries by mean recall rate ({\em $Recall@k(k^{\prime})$}, Eq. \ref{equ: recall}) and mean similarity measure error (\textit{SME}, Eq. \ref{equ: sme}). We use queries per second ($QPS$) to measure search efficiency. $QPS$ is the number of queries ($\#q$) divided by the total response time ($\tau$), i.e., $\#q/\tau$.

\begin{table}[!tb]
  \centering
  \setlength{\abovecaptionskip}{0.05cm}
  \setlength{\belowcaptionskip}{0.1cm}
  \setstretch{0.8}
  \fontsize{6.5pt}{3.3mm}\selectfont
  \caption{Search accuracy on MIT-States.}
  \label{tab: accuracy mitstates}
  \setlength{\tabcolsep}{.002\linewidth}{
  \begin{tabular}{|c|l|l|l|l|l|}
    \hline
    \textbf{Framework} & \textbf{Encoder} & \textbf{\textit{Recall}@1(1)} & \textbf{\textit{Recall}@5(1)} & \textbf{\textit{Recall}@10(1)} & \textbf{\textit{SME}} \\
    \hline
    \multirow{2}*{\boldsymbol{\bii}} & TIRG & 0.1181 & 0.3027 & 0.4175 & 0.1574 \\
    \cline{2-6}
    ~ & CLIP & 0.2236 & 0.4979 & 0.6187 & 0.1382 \\
    \hline
    \multirow{8}*{\boldsymbol{\bi}} & ResNet17+LSTM & 0.3998 & 0.6336 & 0.7106 & 0.1222 \\
    \cline{2-6}
    ~ & ResNet50+LSTM & 0.5401 & 0.7104 & 0.7639 & 0.1012 \\
    \cline{2-6}
    ~ & ResNet17+Transformer & 0.2435 & 0.4110 & 0.4931 & 0.1381 \\
    \cline{2-6}
    ~ & ResNet50+Transformer & 0.3112 & 0.4475 & 0.5142 & 0.1404 \\
    \cline{2-6}
    ~ & TIRG+LSTM & 0.3768 & 0.6574 & 0.7691 & 0.1283 \\
    \cline{2-6}
    ~ & TIRG+Transformer & 0.2830 & 0.4918 & 0.5834 & 0.1395 \\
    \cline{2-6}
    ~ & CLIP+LSTM & 0.4911 & 0.7619 & 0.8436 & 0.1108 \\
    \cline{2-6}
    ~ & CLIP+Transformer & 0.3707 & 0.5912 & 0.6751 & 0.1285 \\
    \hline
    \multirow{8}*{\tabincell{c}{\boldsymbol{\name}\\ \textbf{(ours)}}} & ResNet17+LSTM & 0.5275 & 0.7897 & 0.8780 & 0.0915 \\
    \cline{2-6}
    ~ & ResNet50+LSTM & \textbf{0.6655($\uparrow$23.2\%)} & \textbf{0.8558($\uparrow$12.3\%)} & \textbf{0.9127($\uparrow$8.2\%)} & \textbf{0.0738} \\
    \cline{2-6}
    ~ & ResNet17+Transformer & 0.3325 & 0.4828 & 0.5548 & 0.1272 \\
    \cline{2-6}
    ~ & ResNet50+Transformer & 0.3743 & 0.4866 & 0.5367 & 0.1344 \\
    \cline{2-6}
    ~ & TIRG+LSTM & 0.4202 & 0.7012 & 0.8137 & 0.1184 \\
    \cline{2-6}
    ~ & TIRG+Transformer & 0.3131 & 0.4800 & 0.5543 & 0.1333 \\
    \cline{2-6}
    ~ & CLIP+LSTM & 0.5376 & 0.7859 & 0.8678 & 0.1006 \\
    \cline{2-6}
    ~ & CLIP+Transformer & 0.4190 & 0.5262 & 0.5731 & 0.1229 \\
    \hline
  \end{tabular}
  }
\end{table}

\begin{table}[!tb]
  \vspace{-0.1cm}
  \centering
  \setlength{\abovecaptionskip}{0.05cm}
  \setlength{\belowcaptionskip}{-0.1cm}
  \setstretch{0.8}
  \fontsize{6.5pt}{3.3mm}\selectfont
  \caption{Search accuracy on CelebA.}
  \label{tab: accuracy celeba}
  \setlength{\tabcolsep}{.004\linewidth}{
  \begin{tabular}{|c|l|l|l|l|l|}
    \hline
    \textbf{Framework} & \textbf{Encoder} & \textbf{\textit{Recall}@1(1)} & \textbf{\textit{Recall}@5(1)} & \textbf{\textit{Recall}@10(1)} & \textbf{\textit{SME}} \\
    \hline
    \multirow{2}*{\boldsymbol{\bii}} & TIRG & 0.2725 & 0.5258 & 0.6220 & 0.1896 \\
    \cline{2-6}
    ~ & CLIP & 0.3644 & 0.7006 & 0.7789 & 0.1453 \\
    \hline
    \multirow{4}*{\boldsymbol{\bi}} & ResNet17+Encoding & 0.3337 & 0.5477 & 0.6233 & 0.1724 \\
    \cline{2-6}
    ~ & ResNet50+Encoding & 0.3098 & 0.5029 & 0.5717 & 0.2047 \\
    \cline{2-6}
    ~ & TIRG+Encoding & 0.3275 & 0.5707 & 0.6622 & 0.1875 \\
    \cline{2-6}
    ~ & CLIP+Encoding & 0.4578 & 0.7319 & 0.7990 & 0.1416 \\
    \hline
    \multirow{4}*{\tabincell{c}{\boldsymbol{\name}\\ \textbf{(ours)}}} & ResNet17+Encoding & 0.5701 & 0.7888 & 0.8446 & 0.1087 \\
    \cline{2-6}
    ~ & ResNet50+Encoding & 0.5423 & 0.7539 & 0.8106 & 0.1293 \\
    \cline{2-6}
    ~ & TIRG+Encoding & 0.4932 & 0.7377 & 0.8099 & 0.1433 \\
    \cline{2-6}
    ~ & CLIP+Encoding & \textbf{0.6388($\uparrow$39.5\%)} & \textbf{0.8583($\uparrow$17.3\%)} & \textbf{0.9024($\uparrow$12.9\%)} & \textbf{0.0952} \\
    \hline
  \end{tabular}
  }\vspace{-0.6cm}
\end{table}

\begin{table}[!tb]
  \centering
  \setlength{\abovecaptionskip}{0.05cm}
  \setlength{\belowcaptionskip}{0.4cm}
  \setstretch{0.8}
  \fontsize{6.5pt}{3.3mm}\selectfont
  \caption{Search accuracy on Shopping (T-shirt).}
  \label{tab: accuracy shopping}
  \setlength{\tabcolsep}{.003\linewidth}{
  \begin{tabular}{|c|l|l|l|l|l|}
    \hline
    \textbf{Framework} & \textbf{Encoder} & \textbf{\textit{Recall}@1(1)} & \textbf{\textit{Recall}@5(1)} & \textbf{\textit{Recall}@10(1)} & \textbf{\textit{SME}} \\
    \hline
    \boldsymbol{\bii} & TIRG & 0.1320 & 0.4005 & 0.5162 & 0.0964 \\
    \hline
    \multirow{2}*{\boldsymbol{\bi}} & ResNet17+Encoding & 0.0027 & 0.0190 & 0.0399 & 0.1379 \\
    \cline{2-6}
    ~ & TIRG+Encoding & 0.1320 & 0.4015 & 0.5206 & 0.0964 \\
    \hline
    \multirow{2}*{\tabincell{c}{\boldsymbol{\name}\\ \textbf{(ours)}}} & ResNet17+Encoding & 0.4208 & 0.6931 & 0.7973 & 0.0743 \\
    \cline{2-6}
    ~ & TIRG+Encoding & \textbf{0.4669($\uparrow$253.7\%)} & \textbf{0.7585($\uparrow$88.9\%)} & \textbf{0.8507($\uparrow$63.4\%)} & \textbf{0.0651} \\
    \hline
  \end{tabular}
  }\vspace{-0.5cm}
\end{table}

\noindent\textbf{Setup and Parameters.} All experiments are conducted on a Linux server equipped with an Intel(R) Xeon(R) Gold 6248R CPU running at 3.00GHz and 755G memory. We perform three repeated trials and report the average results for all evaluation metrics. Due to the space limitation, we put the detailed settings in Appendix \ref{appendix: setup and parameters}.

\subsection{Accuracy Evaluation}
\label{subsec: accuracy evaluation}

In our evaluation of all methods, we employ various encoders on four real-world datasets (cf. Appendix \ref{appendix: encoders}).
For {\bii}, we utilize multimodal encoders such as TIRG \cite{vo2019composing}, CLIP \cite{CLIP2021}, and MPC \cite{neculai2022probabilistic} to embed all modalities into the vector space of the target modality.
In the case of {\bi} and {\name}, we employ unimodal encoders, such as ResNet \cite{he2016deep} and Transformer \cite{DevlinCLT19}, to individually embed each modality. Additionally, we obtain a composition vector using a multimodal encoder (such as CLIP \cite{CLIP2021}), which is then used to replace the vector representation of the target modality.

\begin{table}[!tb]
  \centering
  \setlength{\abovecaptionskip}{0.05cm}
  \setlength{\belowcaptionskip}{0.35cm}
  \setstretch{0.8}
  \fontsize{6.5pt}{3.3mm}\selectfont
  \caption{Search accuracy on MS-COCO.}
  \label{tab: accuracy coco}
  \setlength{\tabcolsep}{.003\linewidth}{
  \begin{tabular}{|c|l|l|l|l|}
    \hline
    \textbf{Framework} & \textbf{Encoder} & \textbf{\textit{Recall}@10(1)} & \textbf{\textit{Recall}@50(1)} & \textbf{\textit{Recall}@100(1)} \\
    \hline
    \textbf{\boldsymbol{\bii}} & MPC & 0.0202 & 0.0865 & 0.1512 \\
    \hline
    \multirow{2}*{\boldsymbol{\bi}} & MPC+GRU+ResNet50 & 0.0647 & 0.1827 & 0.2741 \\
    \cline{2-5}
    ~ & ResNet50+GRU+ResNet50 & 0.0493 & 0.1633 & 0.2425 \\
    \hline
    \multirow{2}*{\tabincell{c}{\boldsymbol{\name}\\ \textbf{(ours)}}} & MPC+GRU+ResNet50 & 0.0825 & 0.2272 & 0.3363 \\
    \cline{2-5}
    ~ & ResNet50+GRU+ResNet50 & \textbf{0.0914($\uparrow$41.3\%)} & \textbf{0.2498($\uparrow$36.7\%)} & \textbf{0.3711($\uparrow$35.4\%)} \\
    \hline
  \end{tabular}
  }\vspace{-0.5cm}
\end{table}

\begin{figure}[!t]
  \setlength{\abovecaptionskip}{0cm}
  \setlength{\belowcaptionskip}{0cm}
  \centering
  \includegraphics[width=\linewidth]{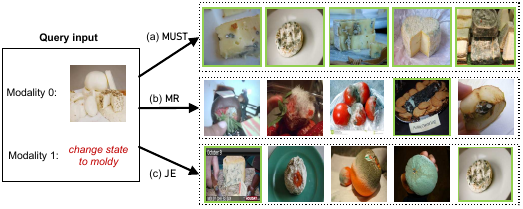}
  \caption{Top-5 examples of different frameworks on MIT-States. The green box marks the ground-truth objects.}
  \label{fig: mitstates_case}
  \vspace{-0.4cm}
\end{figure}

\begin{figure*}[!th]
\setstretch{0.9}
\fontsize{8pt}{4mm}\selectfont
\begin{minipage}{0.06\textwidth}
  \setlength{\abovecaptionskip}{0cm}
  \setlength{\belowcaptionskip}{0cm}
  \centering
  \scriptsize
  \stackunder[0.5pt]{\includegraphics[scale=0.13]{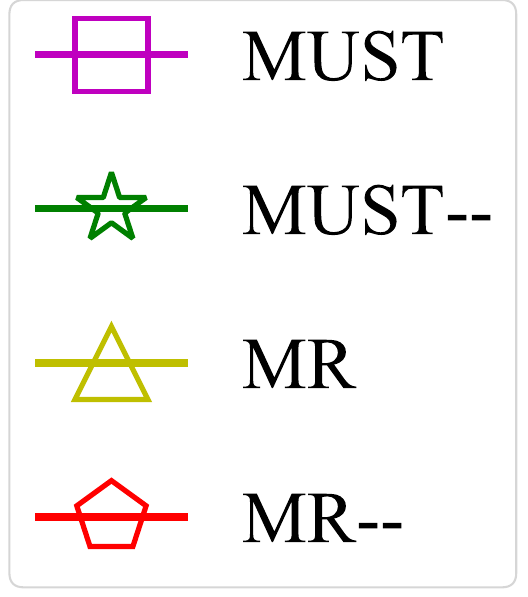}}{}\vspace{0.8cm}
\end{minipage}
\begin{minipage}{0.45\textwidth}
  \setlength{\abovecaptionskip}{0cm}
  \setlength{\belowcaptionskip}{0cm}
  \centering
  \scriptsize
  \stackunder[0.5pt]{\includegraphics[scale=0.14]{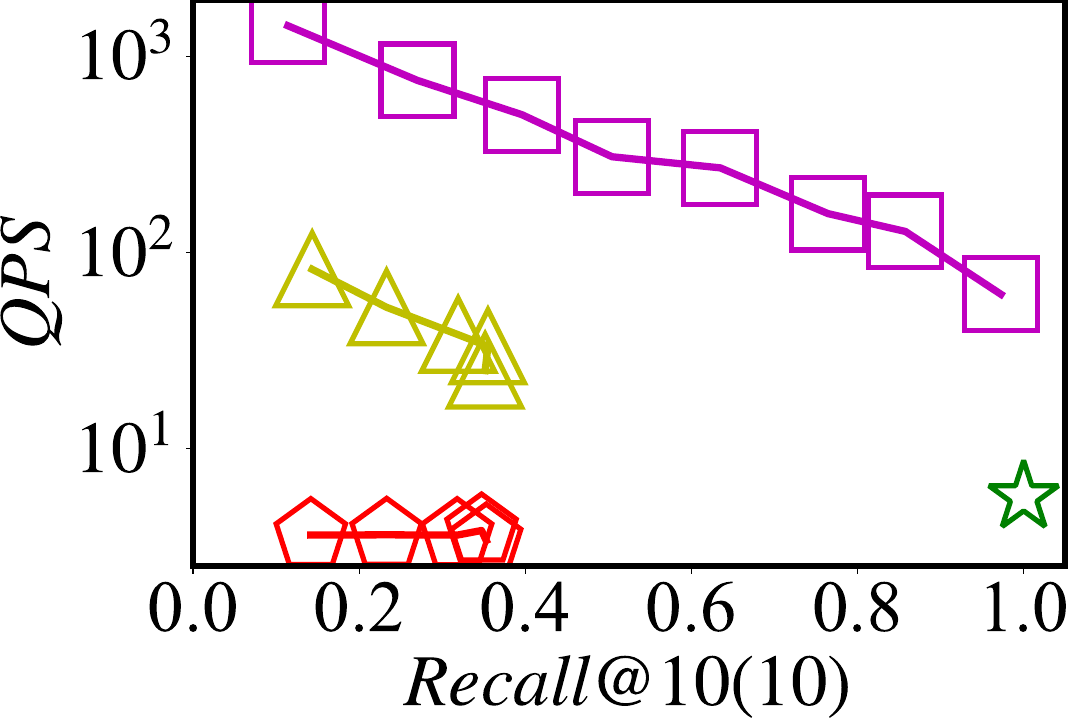}}{(a) ImageText1M}
  \stackunder[0.5pt]{\includegraphics[scale=0.14]{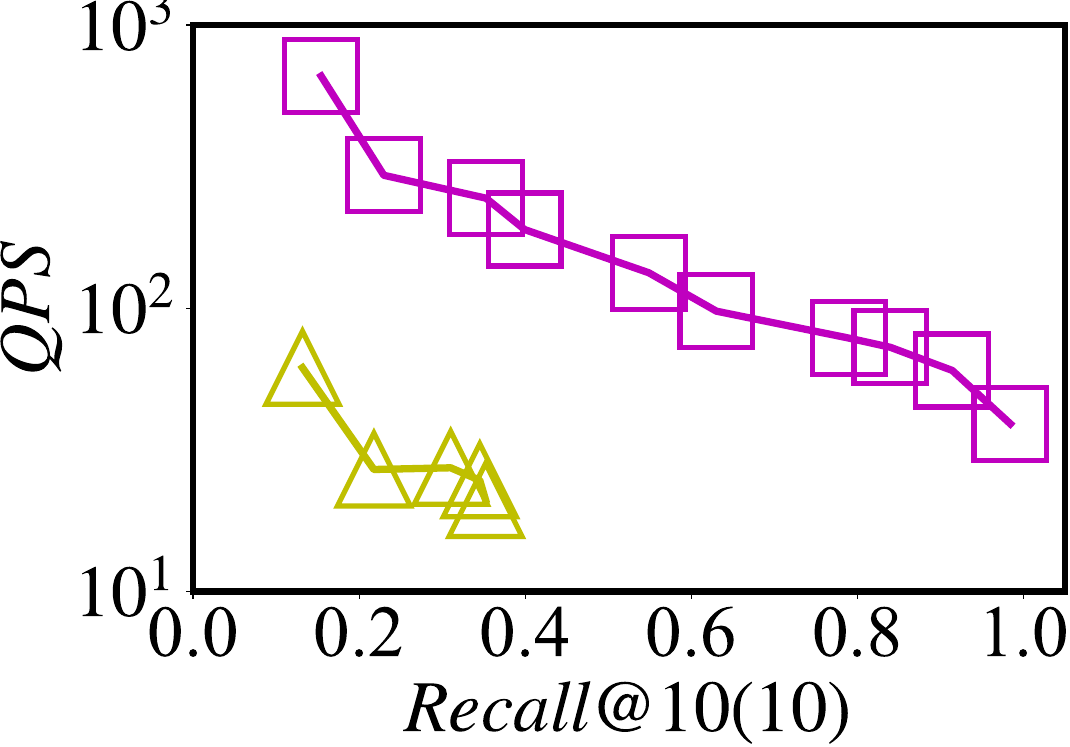}}{(b) AudioText1M}
  \stackunder[0.5pt]{\includegraphics[scale=0.14]{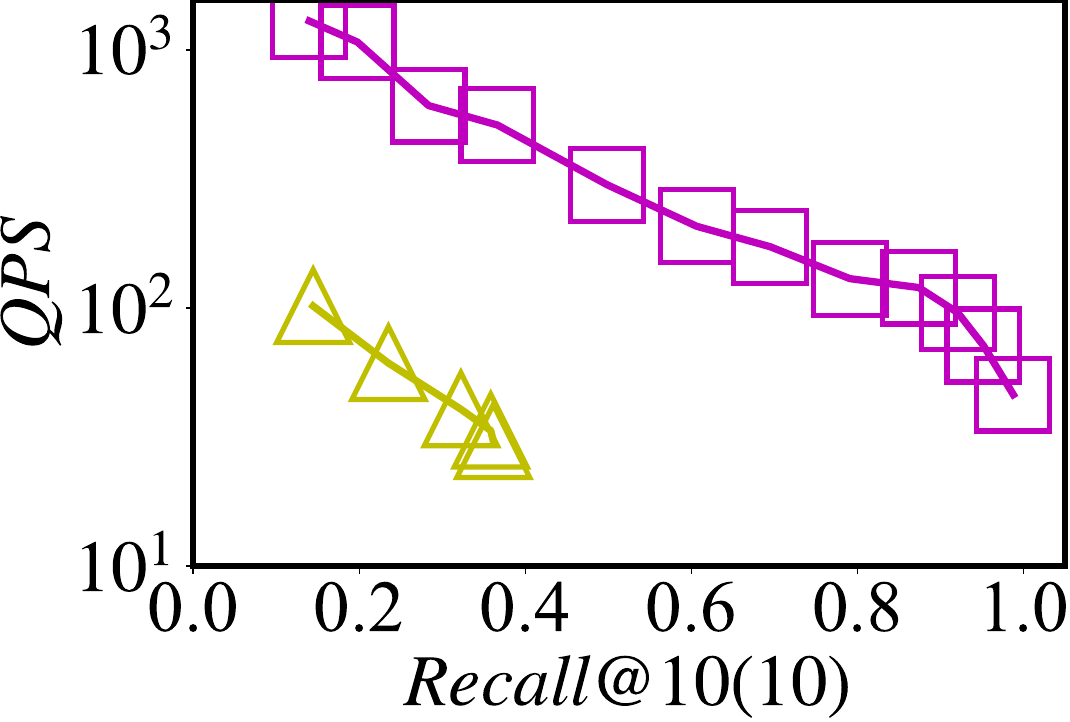}}{(c) VideoText1M}
  \newline
  \caption{Efficiency evaluation of different methods.}
  \label{fig: efficiency evaluation}
\end{minipage}\hspace{0.2cm}
\begin{minipage}{0.484\textwidth}
  \centering
  \setlength{\abovecaptionskip}{0cm}
  \setlength{\belowcaptionskip}{0cm}
  \setstretch{0.8}
  \fontsize{6.5pt}{3.3mm}\selectfont
  \tabcaption{Response time comparison (in seconds) of {\name}{-}{-} and {\name} when ${Recall@10(10)} > 0.99$ under different data volumes. The value in parentheses shows the percentage of response time decrease by using {\name}.}
  \label{tab: scalability search time}
  \setlength{\tabcolsep}{.008\linewidth}{
  \begin{tabular}{|l|l|l|l|l|l|}
    \hline
    \textbf{Scale} & 1M & 2M & 4M & 8M & 16M \\
    \hline
    {\name}{-}{-} & 15.4 & 32.8 & 67.5 & 129.9 & 266.9 \\
    \hline
    {\name} & 2.7 ($\downarrow$82.5\%) & 2.7 ($\downarrow$91.8\%) & 3.4 ($\downarrow$95.0\%) & 3.4 ($\downarrow$97.4\%) & 4.4 ($\downarrow$98.4\%) \\
    \hline
  \end{tabular}
  }
\end{minipage}
\vspace{-0.3cm}
\end{figure*}

\begin{figure*}[!th]
\setlength{\abovecaptionskip}{0cm}
\setstretch{0.9}
\fontsize{8pt}{4mm}\selectfont
\begin{minipage}{0.3\textwidth}
  \setlength{\abovecaptionskip}{0cm}
  \setlength{\belowcaptionskip}{-0.2cm}
  \centering
  \scriptsize
  \stackunder[0.5pt]{\includegraphics[scale=0.14]{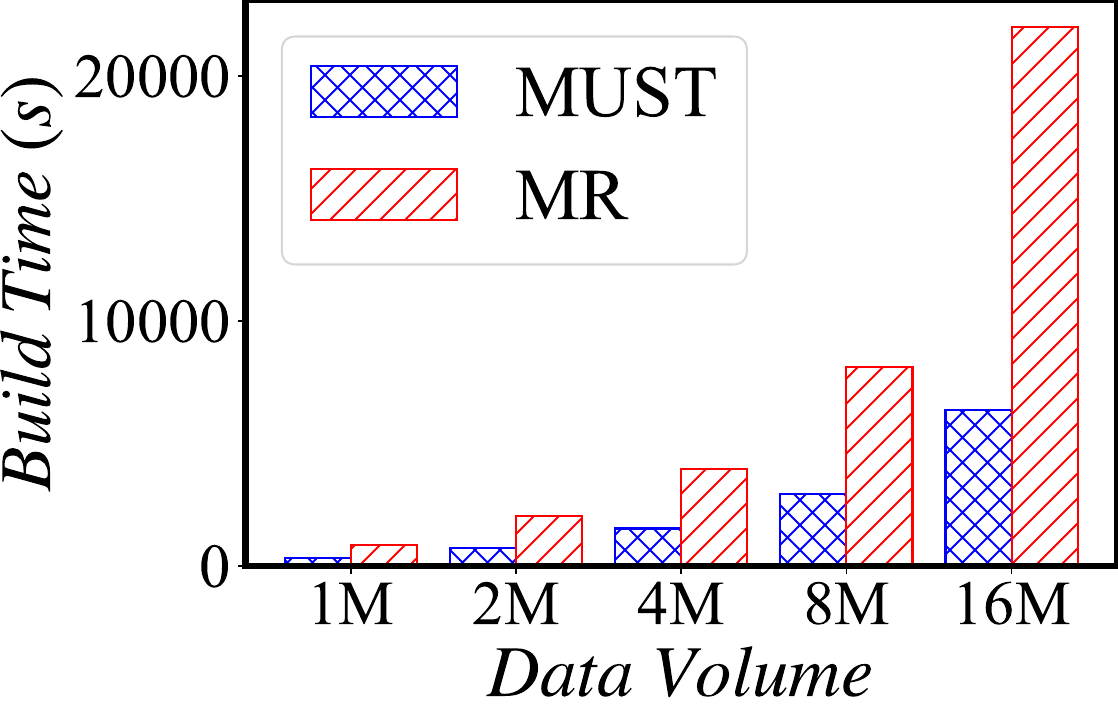}}{(a) Build time (s)}
  \stackunder[0.5pt]{\includegraphics[scale=0.14]{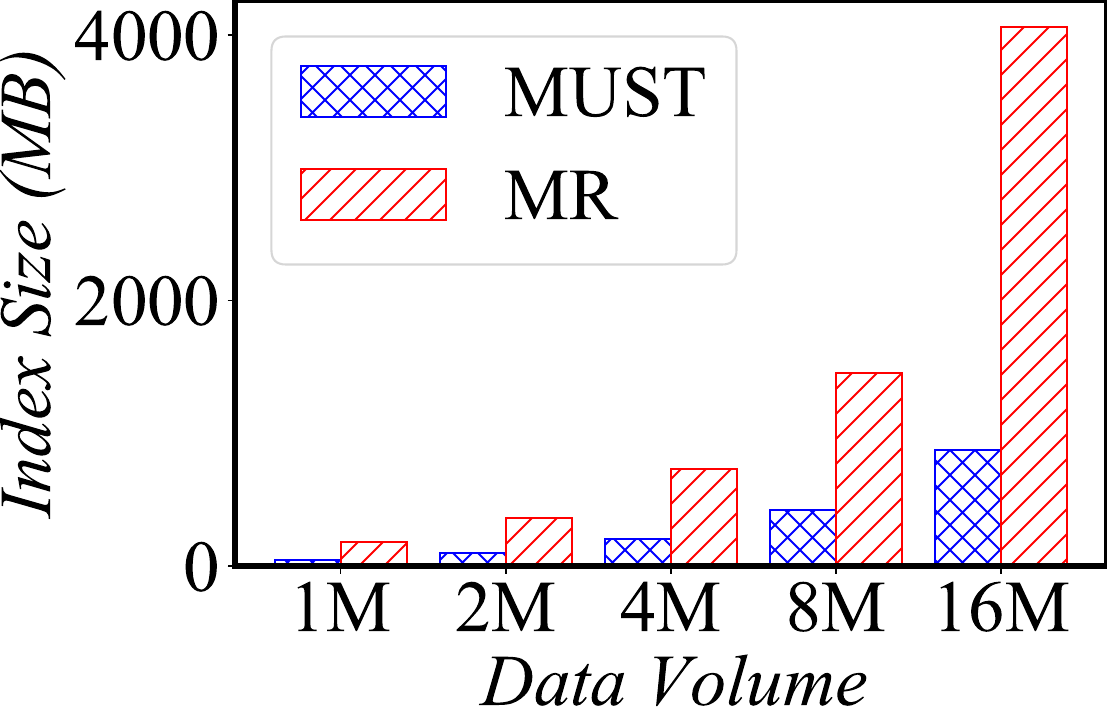}}{(b) Index size (MB)}
  \newline
  \caption{Effect of different data volumes.}
  \label{fig: scalability index build}
\end{minipage}\hspace{0.2cm}
\begin{minipage}{0.24\textwidth}
  \centering
  \setlength{\abovecaptionskip}{0.2cm}
  \setlength{\belowcaptionskip}{0cm}
  \setstretch{0.8}
  \fontsize{6.5pt}{3.3mm}\selectfont
  \tabcaption{Recall rates with different numbers of modalities on CelebA+.}
  \label{tab: scalability number of modalities}
  \setlength{\tabcolsep}{.008\linewidth}{
  \begin{tabular}{|l|l|l|l|}
    \hline
    \textbf{\# Modality} ($m$) & 2 & 3 & 4 \\
    \hline
    \textbf{{\bi}} (\textit{Recall}@1(1)) & 0.4578 & 0.4613 & 0.4599 \\
    \hline
    {\name} (\textit{Recall}@1(1)) & 0.6388 & 0.6771 & 0.6956 \\
    \hline
  \end{tabular}
  }
\end{minipage}\hspace{0.2cm}
\begin{minipage}{0.43\textwidth}
  \setlength{\abovecaptionskip}{0cm}
  \setlength{\belowcaptionskip}{-0.2cm}
  \centering
  \scriptsize
  \stackunder[0.5pt]{\includegraphics[scale=0.14]{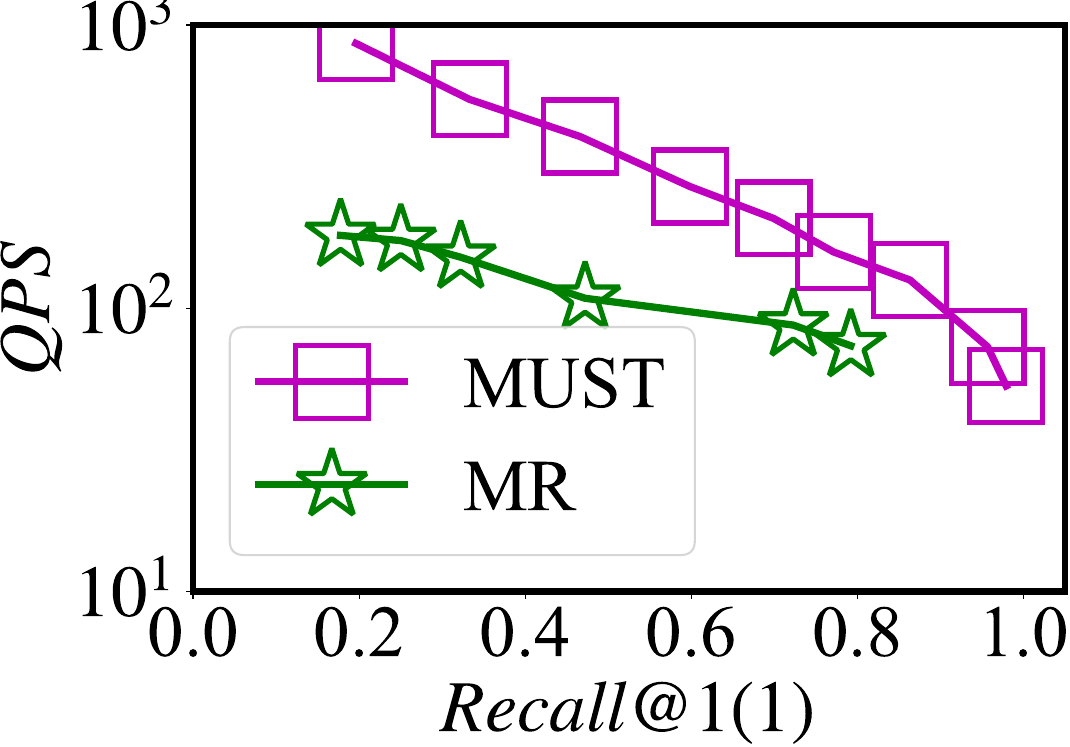}}{(a) $k=1$}
  \stackunder[0.5pt]{\includegraphics[scale=0.14]{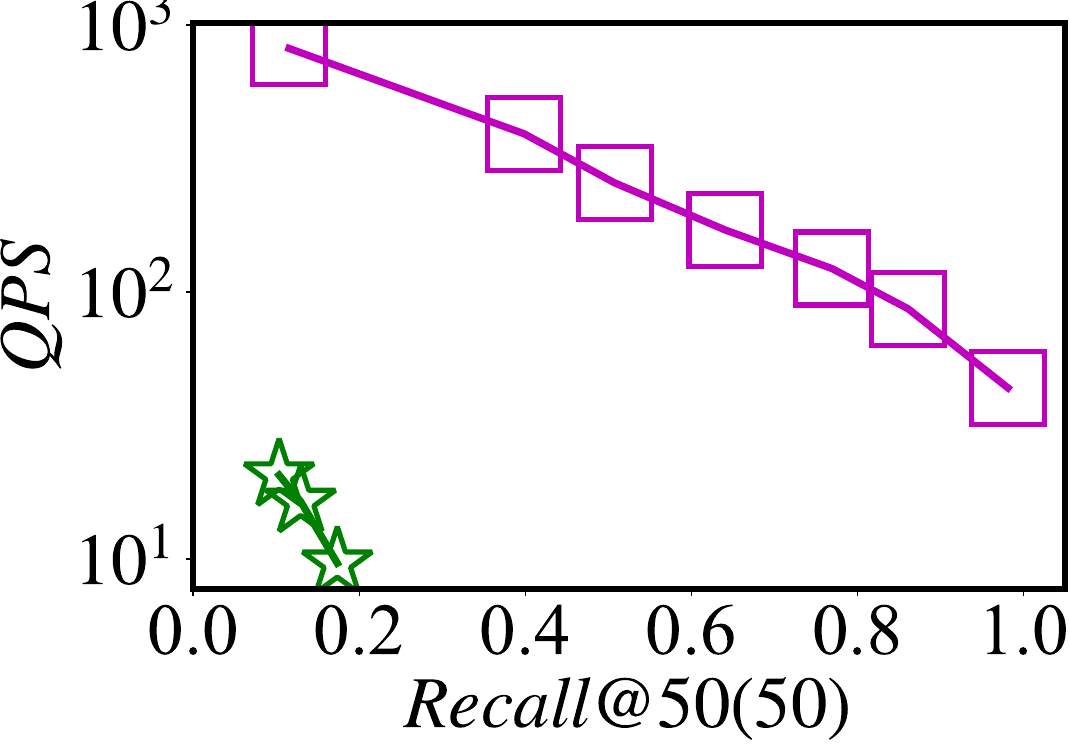}}{(b) $k=50$}
  \stackunder[0.5pt]{\includegraphics[scale=0.14]{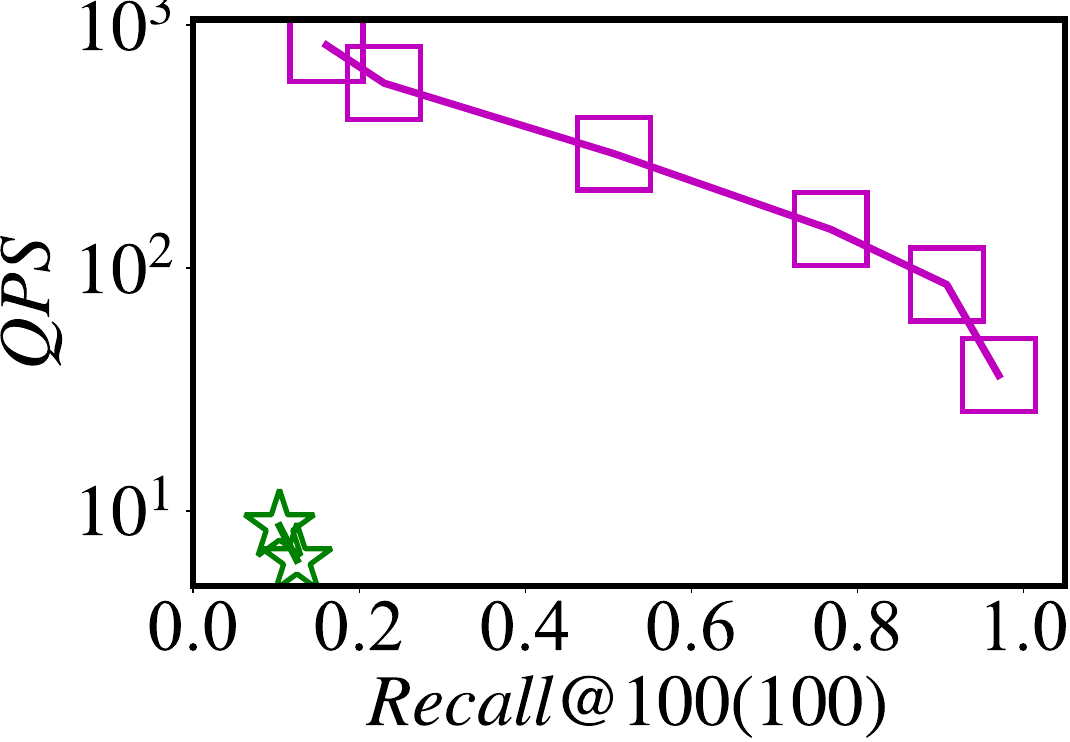}}{(c) $k=100$}
  \newline
  \caption{Effect of different $k$.}
  \label{fig: query type k}
\end{minipage}
\vspace{-0.3cm}
\end{figure*}

Tab. \ref{tab: accuracy mitstates}–\ref{tab: accuracy coco} show the search accuracy and $SME$ of the three frameworks.
We have three major observations as summarized below: \underline{First}, {\name} significantly outperforms its competitors on all experimental datasets. Notably, {\name} achieves at least 198\% and 23\% improvement over {\bii} and {\bi} respectively for their best $Recall@1(1)$ on MIT-States. {\name} also reduces the $SME$ on all datasets. 
\underline{Second}, different encoders yield varying recall rates, e.g., CLIP, being a state-of-the-art multimodal encoder, achieves the highest accuracy in single-vector representation. This underscores the importance of encoder selection in achieving optimal performance. An advantage of {\name} is its pluggable embedding component, which allows seamless integration of newly-devised encoders.
\underline{Third}, multi-vector representation exhibits higher recall rates. For example, {\bi} (CLIP+LSTM) and {\name} (CLIP+LSTM) are better than {\bii} (CLIP) on MIT-States. Even with the same multi-vector representation, {\name} consistently achieves larger improvements compared to {\bi}. On the most challenging MS-COCO dataset, both {\bi} and {\name} demonstrate impressive performance compared to {\bii}, which struggles due to fusing three modalities, leading to larger embedding errors.

\subsection{Case Study}
\label{subsec: case study}
Fig. \ref{fig: mitstates_case} shows some case studies of the {\problem} problem on MIT-States. We use the best encoder for each framework based on Table \ref{tab: accuracy mitstates} (CLIP for {\bii}, ResNet50+LSTM for {\bi} and {\name}). We give a query input with an \textit{image of fresh cheese} and a \textit{text description of ``change state to moldy''}, and show the top-5 search results from different frameworks. The results show that {\name} outperforms its competitors. The objects returned by {\name} all satisfy the multimodal constraints, while most objects from {\bi} and {\bii} only match some of the requirements. We provide more recall examples of other queries and datasets in Appendix \ref{appendix: recall examples}.

\subsection{Efficiency Evaluation}
\label{subsec: efficiency evaluation}
We evaluate the performance of {\name}'s fused index and joint search strategy on three million-scale datasets. We conduct comparisons with {\bi}\footnote{For fairness, we exclude {\bii} as it only utilizes a single-vector representation and exhibits much lower accuracy.}, which applies the same index and search strategy to each vector set. Additionally, we implement brute-force versions for vector search of both {\name} and {\bi}, labeled as {\name}{-}{-} and {\bi}{-}{-}, respectively. 

Fig. \ref{fig: efficiency evaluation} presents the QPS vs Recall comparison, where we adjust the parameter $l$ in Algorithm \hyperref[alg: search]{2} to achieve different recall rates. For clarity, we exclude {\name}{-}{-} and {\bi}{-}{-} on AudioText1M and VideoText1M due to their slow performance. The results yield the following observations:
\underline{First}, {\name} is 10$\times$ faster than {\bi} with the same recall rate, which can be attributed to the time-consuming merging operation employed by {\bi}.
\underline{Second}, the recall rate of {\bi} is less than 0.4. We find that the merging operation causes the accuracy of {\bi} to become non-increasing with increasing $l$. Initially, the recall rate of {\bi} increases as $l$ grows, owing to the increased chances of finding the target when intersecting results from each modality. However, as $l$ further increases, the size of the intersection often exceeds $k$, making it challenging to identify the top-$k$ results. Note that the importance of different modalities in {\problem} is unknown, which may necessitate additional optimization for selecting the top-$k$ objects from a large set.
\underline{Third}, both {\name} and {\bi} are more than 10$\times$ faster than their brute-force counterparts, which indicates the effectiveness of our indexing and searching strategies.

\subsection{Scalability}
\label{subsec: scalability}

\noindent\textbf{Data Volume ($n$).} Tab. \ref{tab: scalability search time} shows the response time of {\name} and {\name}{-}{-} with varying $n$.
We find that the response time of {\name}{-}{-} increases linearly with the growth of $n$. In contrast, {\name} exhibits only a slight increase in response time even with large $n$ and reduces the response time by up to 98.4\% when $n$ is 16 million. Fig. \ref{fig: scalability index build} illustrates that {\name}'s build time and index size are significantly lower than {\bi}, affirming its efficiency and scalability in large-scale scenarios.

\noindent\textbf{Number of Modalities ($m$).} Tab. \ref{tab: scalability number of modalities} reports the recall rates of {\name} and {\bi} with different $m$. Overall, the recall rate increases with $m$ for both methods, as more information leads to more accurate results. However, the challenge of merging becomes more pronounced in {\bi} as the number of modalities increases. As a result, the recall rate of $m=4$ in {\bi} is even lower than that of $m=3$. This highlights {\name}'s capability to effectively handle multiple modalities.

\subsection{Query Workloads}
\label{subsec: query requirement}

\begin{figure*}[!th]
\setstretch{0.9}
\fontsize{8pt}{4mm}\selectfont
\begin{minipage}{0.33\textwidth}
\centering
  \setlength{\abovecaptionskip}{0.3cm}
  \setlength{\belowcaptionskip}{0cm}
  \setstretch{0.8}
  \fontsize{6.5pt}{3.3mm}\selectfont
  \tabcaption{Effect of different user-defined weights. ${q}$ is the query input and ${r}$ is the returned result.}
  \label{tab: query type user-defined weights}
  \setlength{\tabcolsep}{.008\linewidth}{
  \begin{tabular}{|l|l|l|l|l|l|l|}
    \hline
    \multirow{2}*{\textbf{Weights}} & $\omega_0^2$ & 0.5 & 0.6 & 0.7 & 0.8 & 0.9 \\
    \cline{2-7}
    ~ & $\omega_1^2$ & 0.5 & 0.4 & 0.3 & 0.2 & 0.1 \\
    \hline
     \multicolumn{2}{|c|}{${IP(\phi_0(q^0),\phi_0(r^0))}$} & {0.6915} & {0.7009} & {0.7440} & {0.8286} & {0.9301} \\
    \hline
    \multicolumn{2}{|c|}{${IP(\phi_1(q^1),\phi_1(r^1))}$} & 0.9999 & 0.9960 & 0.9748 & 0.9242 & 0.8525 \\
    \hline
  \end{tabular}
  }
\end{minipage}\hspace{0.4cm}
\begin{minipage}{0.27\textwidth}
  \centering
  \setlength{\abovecaptionskip}{0.2cm}
  \setlength{\belowcaptionskip}{0cm}
  \setstretch{0.8}
  \fontsize{6.5pt}{3.3mm}\selectfont
  \tabcaption{Effect of single query modality on MIT-States.}
  \label{tab: query modalities}
  \setlength{\tabcolsep}{.008\linewidth}{
  \begin{tabular}{|c|l|l|l|}
    \hline
    \textbf{Modality} & \textbf{Encoder} & \textbf{\textit{Recall}@1(1)} & \textbf{\textit{Recall}@5(1)} \\
    \hline
    \multirow{2}*{\textbf{Target}} & ResNet17 & 0.0268 & 0.1103 \\
    \cline{2-4}
    ~ & ResNet50 & 0.0363 & 0.1393 \\
    \hline
    \multirow{2}*{\textbf{Auxiliary}} & LSTM & 0.2747 & 0.4343 \\
    \cline{2-4}
    ~ & Transformer & 0.2601 & 0.2641 \\
    \hline
  \end{tabular}
  }
\end{minipage}\hspace{0.1cm}
\begin{minipage}{0.355\textwidth}
  \setlength{\abovecaptionskip}{0cm}
  \setlength{\belowcaptionskip}{-0.3cm}
  \centering
  \scriptsize
  \stackunder[0.5pt]{\includegraphics[scale=0.15]{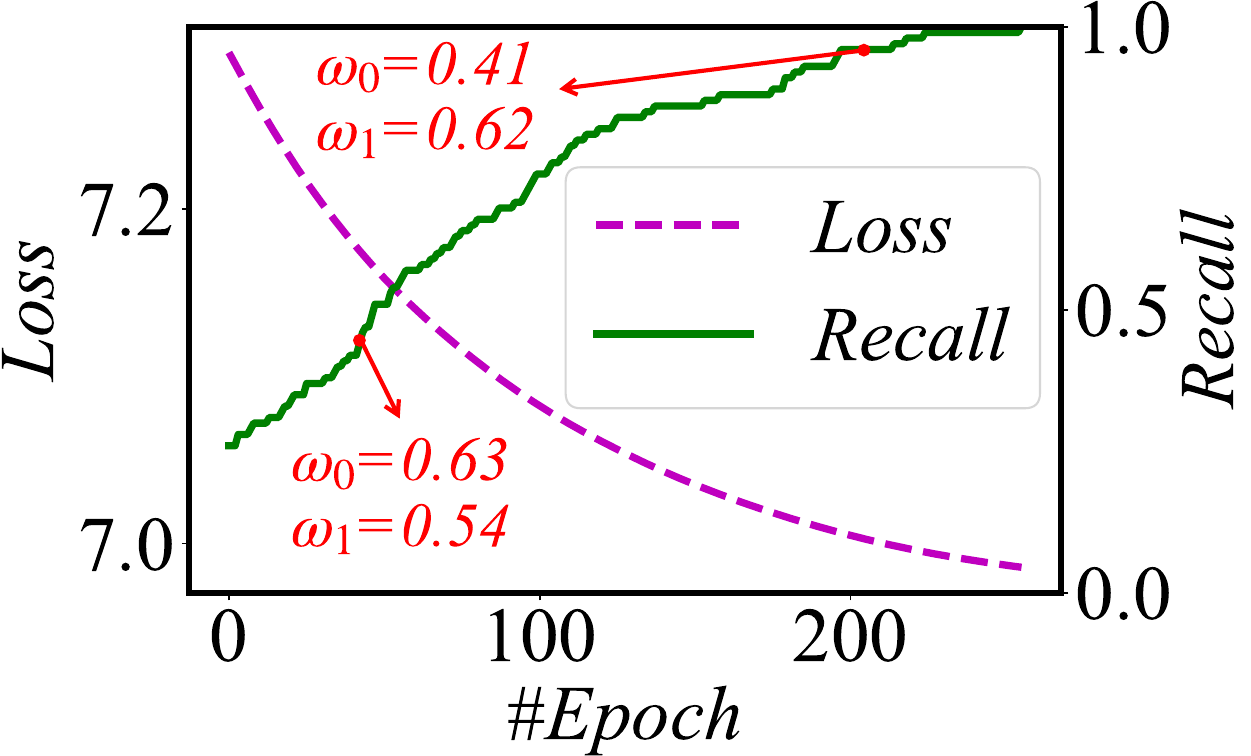}}{(a) Hard}\hspace{0.1cm}
  \stackunder[0.5pt]{\includegraphics[scale=0.15]{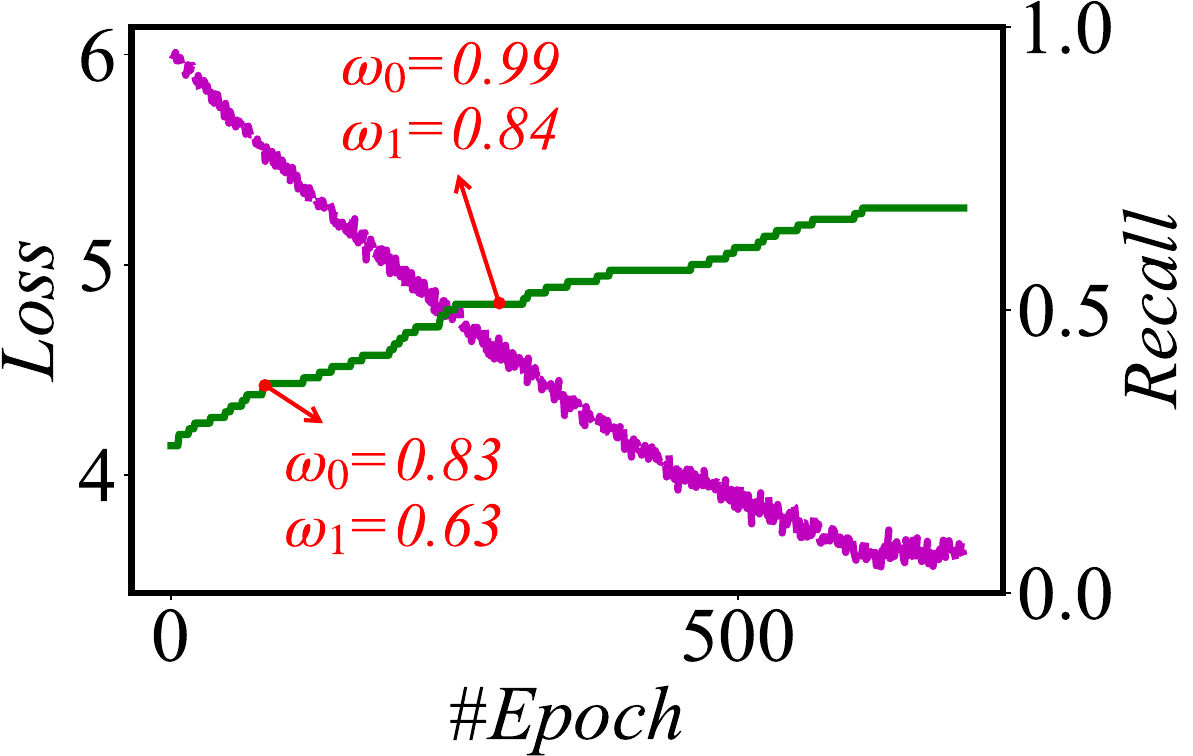}}{(b) Random}
  \newline
  \caption{Effect of different negatives.}
  \label{fig: ablation weight learning}
\end{minipage}
\vspace{-0.1cm}
\end{figure*}

\setlength{\textfloatsep}{0cm}
\setlength{\floatsep}{0cm}
\begin{figure*}[!th]
\setlength{\abovecaptionskip}{0cm}
\setstretch{0.9}
\fontsize{8pt}{4mm}\selectfont
\begin{minipage}{0.765\textwidth}
  \setlength{\abovecaptionskip}{0cm}
  \setlength{\belowcaptionskip}{-0.5cm}
  \centering
  \scriptsize
  \stackunder[0.5pt]{\includegraphics[scale=0.22]{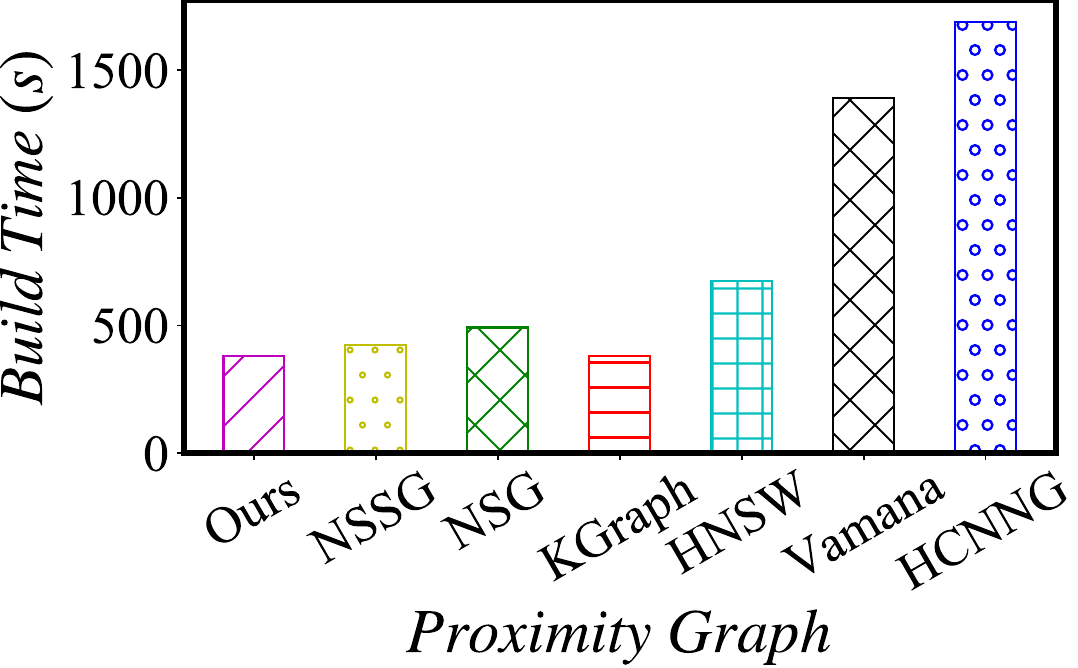}}{{(a) Construction time}}\hspace{0.15cm}
  \stackunder[0.5pt]{\includegraphics[scale=0.22]{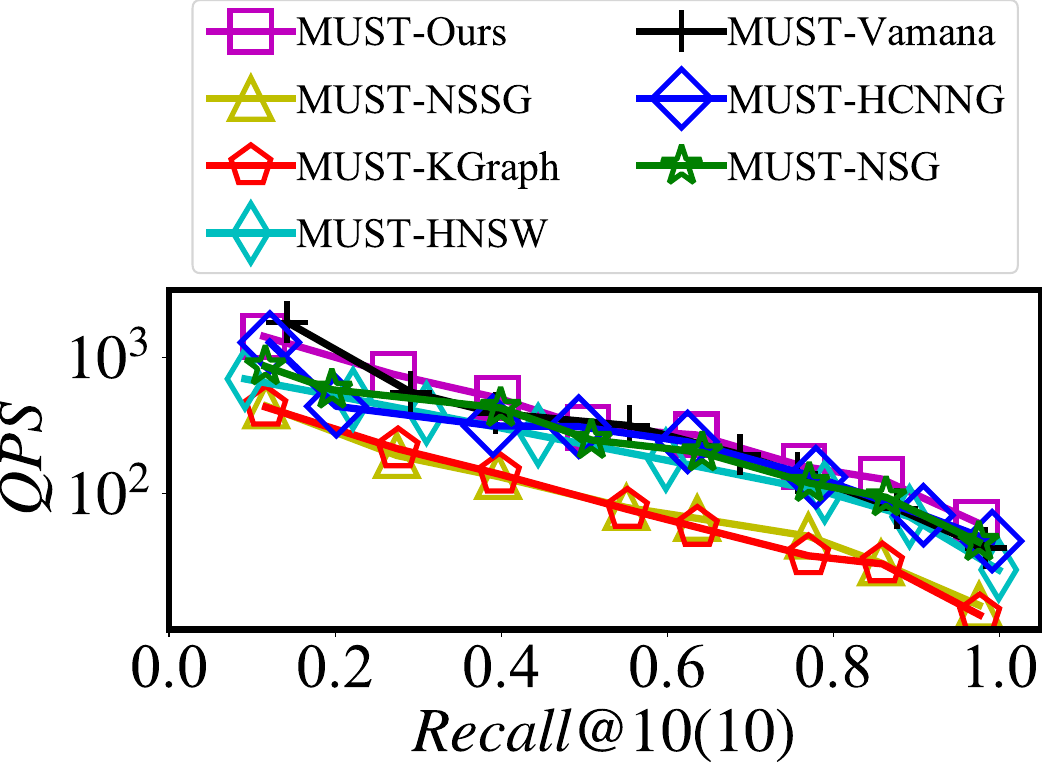}}{{(b) Search performance}}\hspace{0.2cm}
  \stackunder[0.5pt]{\includegraphics[scale=0.22]{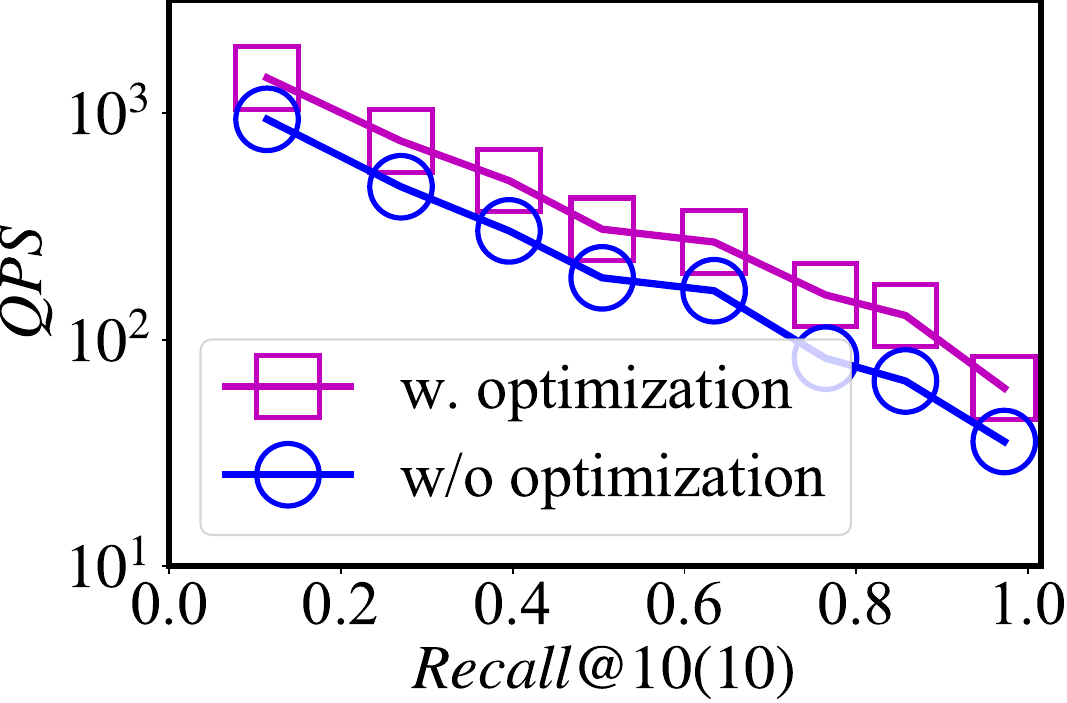}}{(c) Multi-vector computation}
  \newline
  \caption{Effect of index pipeline and search optimization.}
  \label{fig: ablation indexing}
\end{minipage}
\begin{minipage}{0.22\textwidth}
  \centering
  \setlength{\abovecaptionskip}{0cm}
  \setlength{\belowcaptionskip}{-0.5cm}
  \includegraphics[width=\linewidth]{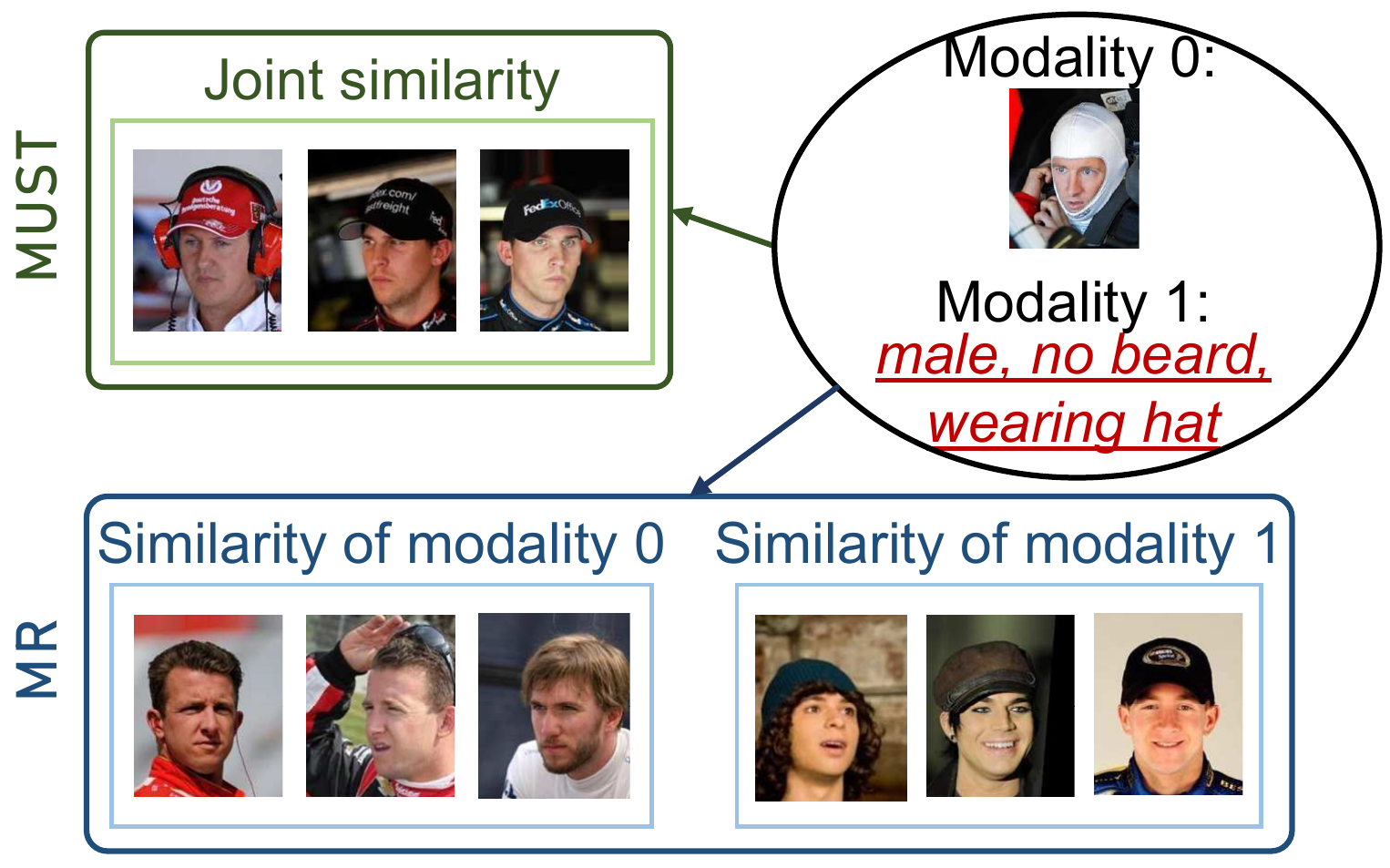}
  \caption{{An example of the top-3 neighbors on different indexes.}}
  \label{fig: index show}
\end{minipage}
\vspace{-0.1cm}
\end{figure*}

\noindent\textbf{Number of Results ($k$).} In Fig. \ref{fig: query type k}, we compare the search performance of {\name} and {\bi} with different $k$ on ImageText1M. The results show that {\name} consistently outperforms {\bi} for any $k$. Moreover, {\name} brings more improvements on larger $k$ while {\bi} has a limited recall rate and \textit{QPS} (cf. \textbf{\S \ref{subsec: efficiency evaluation}}). This is because {\bi} requires more candidates from each modality when $k$ is larger, which makes merging even more challenging, e.g., the number of candidates is 1,300 for the best $Recall@1(1)$, and 10,500 for the best $Recall@100(100)$.

\noindent\textbf{Learned Weights.} We investigate the impact of different queries with fixed learned weights on MIT-States. In Fig. \ref{fig: mitstates_case}, the original text describes something not present in the given image. To create a new query input, we retain the reference image while modifying the text description to \textit{``remove the fresh state''}. This query now describes what is already present in the given image.
Both queries are executed with the same learned weights on {\name}, and we observe that they yield identical query results. This compelling result verifies the generalization capability of the fixed learned weights, highlighting that the learned weights reflect the importance of different modalities, independent of their specific content.

\noindent\textbf{User-defined Weights.} {Tab. \ref{tab: query type user-defined weights} shows the effect of {\name} with different user-defined weights on MIT-States. We calculate the mean similarity over one modality for a batch of query inputs and returned objects. For example, when $\omega_0^2=\omega_1^2=0.5$, the mean \textit{IP} between modality 0 of query inputs and returned objects is 0.6915 and between modality 1 is 0.9999. To get an object whose modality 0 is more similar to the query input, users can increase the weight of modality 0. When $\omega_0^2=0.9$ and $\omega_1^2=0.1$, the returned object has higher similarity to the query input in modality 0. Thus, we can get customized results by adjusting the weight configuration.}

\noindent\textbf{Number of Query Modalities ($t$).} We study how different $t$ values in the queries affect the search accuracy on MIT-States. Tab. \ref{tab: query modalities} shows the search accuracy when only one modality is used in the queries (i.e., $t=1$). Compared with multimodal queries ($t=2$, cf. Tab. \ref{tab: accuracy mitstates}), the single-modal queries have lower search accuracy. Thus, using more query modalities is crucial for the quality of query results.

\subsection{Ablation Study}
\label{subsec: ablation}

\noindent\textbf{Vector Weight Learning Model.}
We compare the proposed hard negative acquisition strategy and the random selection. Fig. \ref{fig: ablation weight learning} shows the loss and recall rate w.r.t. the epoch on ImageText1M. We can observe that the model using the hard negatives converges faster compared to the model using the random ones. Additionally, the learned weights from the hard negatives lead to a higher recall rate, demonstrating the effectiveness of our strategy.
It is important to highlight that the weight learning model is remarkably efficient, as it takes less than 200 seconds to train on all datasets. In contrast, the embedding models used in the process require over 12 hours to train. As a result, the vector weight learning model is lightweight and imposes minimal additional training cost.

\noindent\textbf{Proximity Graph.}
{We implement six proximity graphs in {\name}: KGraph \cite{NNDescent}, NSG \cite{NSG}, NSSG \cite{NSSG}, HNSW \cite{HNSW}, Vamana \cite{DiskANN}, and HCNNG \cite{HCNNG}.} We also re-assemble KGraph, NSG, and NSSG according to \textbf{\S \ref{subsec: index construction}} to form our fused index. We evaluate their indexing and search performance on ImageText1M. Fig. \ref{fig: ablation indexing}(b) shows that our method is more efficient than the competitors. Fig. \ref{fig: ablation indexing}(a) shows the index construction time of different methods. Our optimized one is faster than the others. Therefore, our pipeline facilitates the design of proximity graph algorithms and can improve performance even without new optimization. 
Fig. \ref{fig: index show} shows the visualization of three neighbors for an object on CelebA. The vertex and its neighbors in {\name}'s index balance the importance of different modalities and have a better joint similarity. The vertex and its neighbors in {\bi}’s indexes only consider the similarity in one modality.

\noindent\textbf{Multi-vector Computation Optimization.} Fig. \ref{fig: ablation indexing}(c) shows the effect of multi-vector computation optimization on ImageText1M. This optimization improves the search efficiency without affecting the search accuracy. This is because we can skip some vector computations without losing accuracy by scanning the vectors of each object and query incrementally (cf. Lemma \ref{lemma: multi-vector computation}). This optimization is more significant in high-accuracy regions than in low-accuracy regions.
\section{Discussion}
\label{sec: discussion}
\noindent\textbf{Single Modality Inputs.}
{In scenarios where users provide single-modal query inputs but seek more personalized results, {\name} adapts by refining the query iteratively using a returned target modality example. For instance, in image retrieval, users may only provide text input. {\name} can then generate an output image based on the given text, serving as a reference for users to enhance their query by adding additional text. This interactive process enables users to create more complete multimodal query inputs and obtain the desired results using {\name}, even if the initial inputs are incomplete or imprecise.}

\noindent\textbf{Index Updates.}
{The index in {\name} relies on a proximity graph algorithm, and the efficacy of dynamic updates depends on the specific proximity graph employed. While certain algorithms, like KGraph \cite{NNDescent} and NSG \cite{NSG}, do not support dynamic updates, others, such as HNSW \cite{HNSW} and Vamana \cite{DiskANN}, adeptly handle dynamic updates by incrementally inserting data points.  For instance, upon the arrival of a new object, its embedding vector can be used to search for neighbors in the index, updating them accordingly. However, it is crucial to note that all existing proximity graph algorithms necessitate periodic reconstruction to maintain optimal performance \cite{ADBV}. For example, a deleted data point is not immediately removed from the index, and it can be marked with a data-status bitset. This is because the data point may be essential to ensure the connectivity of the proximity graph. The actual deletion takes place during the reconstruction process. However, this process is time-consuming for proximity graph algorithms, which affects the scalability of {\name} in dynamic data update scenarios. Therefore, supporting efficient index updates in {\name} remains an ongoing concern.}

\section{Conclusion}
\label{sec: discussion and conclusion}

In this study, we thoroughly investigate the {\problem} problem and proposed a novel and effective framework called {\name}. Our framework introduces a hybrid fusion mechanism that intelligently combines different modalities at multiple stages, capturing their relative importance and accurately measuring the joint similarity between objects. Additionally, we have developed a fused proximity graph index and an efficient joint search strategy tailored for multimodal queries. {\name} exhibits the capability to handle interactive multimodal search scenarios, wherein users may lack query inputs for certain modalities. The comprehensive experimental results demonstrate the superior performance of {\name} compared to the baselines, showcasing its advantages in terms of accuracy, efficiency, and scalability. For more in-depth discussions and analysis, we refer readers to Appendix \ref{appendix: discussion}.

Looking ahead, we plan to further enrich {\name} by incorporating additional encoders such as the OpenAI embeddings \cite{openai_embedding_api} and Hugging Face embeddings \cite{hugging_face_embedding_api}.

\section*{Acknowledgment}
This work was supported in part by the NSFC under Grants No. (62102351, 62025206, U23A20296) and the Yongjiang Talent Programme (2022A-237-G). Lu Chen is the corresponding author of the work.

\bibliographystyle{IEEEtran}
\bibliography{IEEEabrv, mybibli.bib}

\clearpage

\appendix

\subsection{Proof of Lemma \ref{lemma: neighbor selection}}
\label{appendix: proof of lemma 2}

\begin{proof}
We consider the scenario where there exist two vertices $u$ and $v$ in the set of final neighbors $N(o)$ for a given vertex $o$, such that the angle between them, denoted as $\theta(u,o,v)$, is less than $60^\circ$. In the triangle $\triangle uov$, the sum of the angles $\theta(o,v,u)$ and $\theta(o,u,v)$ exceeds $120^\circ$. Here, the inner product (IP) of two vertices is used to measure the side length between them, and smaller IP values imply longer sides in the triangle. Therefore, we can conclude that either $\theta(o,v,u) > \theta(o,u,v)$ (i.e., $\theta(o,v,u) > 60^\circ$) or $\theta(o,u,v) > \theta(o,v,u)$ (i.e., $\theta(o,u,v) > 60^\circ$).

\underline{Case 1:} If $\theta(o,v,u) > \theta(o,u,v)$, it follows that $IP(\boldsymbol{\hat{o}},\boldsymbol{\hat{v}}) > IP(\boldsymbol{\hat{o}},\boldsymbol{\hat{u}})$, indicating that the vertex $v$ will be added to $N(o)$ before $u$ (Line 14 in Algorithm \hyperref[alg: index]{1}). Since $\theta(u,o,v) < 60^\circ < \theta(o,v,u)$, we have $IP(\boldsymbol{\hat{u}},\boldsymbol{\hat{v}}) > IP(\boldsymbol{\hat{o}},\boldsymbol{\hat{u}})$. As a result, according to Line 16 in Algorithm \hyperref[alg: index]{1}, vertex $u$ cannot be added to $N(o)$, which contradicts the initial assumption that $u$ is in $N(o)$.

\underline{Case 2:} If $\theta(o,u,v) > \theta(o,v,u)$, we can swap the positions of $u$ and $v$ in the triangle $\triangle uov$ and arrive at the same conclusion as in Case 1.

These cases demonstrate that the assumption of having two neighbors with an angle less than $60^\circ$ in $N(o)$ is not feasible, and thus, it is ensured that the selected neighbors in $N(o)$ maintain an angle of at least $60^\circ$ between each other, as described in Algorithm \hyperref[alg: index]{1}. This property ensures the effectiveness and correctness of the pipeline in constructing the final neighbor sets.
\end{proof}

\subsection{Encoders Used in Our Experiments}
\label{appendix: encoders}

\noindent\textbf{ResNet.} ResNet is a type of deep neural network that utilizes residual blocks and skip connections to facilitate the training of deep networks and mitigate the issue of vanishing or exploding gradients. It was proposed by researchers at Microsoft Research in 2015 \cite{he2016deep} and achieved success in the ImageNet classification task with a 152-layer network. ResNet can also be applied to other visual recognition tasks, including object detection and segmentation. In our experiments, we employed ResNet17 and ResNet50 as encoders for the image modality. These are variations of ResNet with different numbers of layers. ResNet17 consists of 17 layers, while ResNet50 consists of 50 layers. Additionally, ResNet50 adopts a bottleneck design for its residual blocks, which reduces the parameter count and accelerates the training process. Both ResNet17 and ResNet50 can serve as feature extractors for tasks such as object detection or segmentation.

\vspace{0.1cm}
\noindent\textbf{LSTM.} LSTM stands for Long Short-Term Memory, which is a type of recurrent neural network (RNN) designed for processing sequential data, including speech and video \cite{greff2016lstm}. LSTM incorporates feedback connections and a specialized structure called a cell, enabling it to store and update information over long time intervals. It also employs three gates (input, output, and forget) to regulate the flow of information into and out of the cell. LSTM finds applications in various tasks such as speech recognition, machine translation, and handwriting recognition. Furthermore, LSTM can be combined with convolutional neural networks (CNNs) to form a convolutional LSTM network, which proves useful for tasks like video prediction or object tracking. In our experiments, we utilized LSTM as the encoder for the text modality.

\vspace{0.1cm}
\noindent\textbf{Transformer.} The Transformer is a deep learning model introduced in 2017 for natural language processing tasks, such as machine translation and text summarization \cite{DevlinCLT19}. Unlike recurrent neural networks, the Transformer does not process sequential data in a sequential manner. Instead, it employs attention mechanisms to capture dependencies between words or tokens. The Transformer comprises two main components: an encoder and a decoder. The encoder takes an input sentence and converts it into a sequence of vectors known as encodings. The decoder takes these encodings and generates an output sentence. Both the encoder and decoder consist of multiple layers, each containing a multi-head self-attention module and a feed-forward neural network module. Additionally, the Transformer utilizes positional encodings to incorporate positional information for each word in the sentence. In our experiments, we also employed the Transformer to encode the text modality.

\vspace{0.1cm}
\noindent\textbf{GRU.} 
A Gated Recurrent Unit (GRU) encoder is a type of recurrent neural network (RNN) that can encode input sequences of varying lengths into fixed-length feature vectors \cite{zhang2017gru}. It selectively updates and resets its hidden state based on the input and previous state, allowing it to capture both short-term and long-term dependencies within the sequence. This concise representation generated by the GRU encoder can be utilized for various tasks, including machine translation, speech recognition, and natural language understanding. In our experiments, we employed the GRU encoder to encode the text modality.

\vspace{0.1cm}
\noindent\textbf{Encoding.} 
In our experiments, we utilized ordinal encoding \cite{wang2015exploring} to encode the structured text description. This technique transforms categorical data into numerical data by assigning integer values to categories based on their rank or order. For instance, if a feature has three categories: ``low'', ``medium'', and ``high'', they can be encoded as 1, 2, and 3 respectively. Ordinal encoding is appropriate for categorical features that possess a natural ordering, such as grades, sizes, ratings, and so on. The original categories can be restored by reversing the ordinal encoding process, which involves mapping the integer values back to their respective categories.

\vspace{0.1cm}
\noindent\textbf{TIRG.} 
TIRG, which stands for Text-Image Residual Gating, is a method used to merge image and text features for image retrieval tasks \cite{vo2019composing}. It involves modifying the features of the query image using text, while maintaining the resulting feature vector within the same space as the target image. This is accomplished through a gated residual connection, which enhances the encoding and learning of representations. In our experiments, we utilized TIRG to encode image-text pairs as composition vectors.

\vspace{0.1cm}
\noindent\textbf{CLIP.} 
CLIP, which stands for Contrastive Language-Image Pre-Training, is a neural network model developed by OpenAI that has demonstrated remarkable achievements in multi-modal zero-shot learning \cite{CLIP2021}. It is trained on a large dataset of image-text pairs collected from the web and learns to associate images with their corresponding textual descriptions. CLIP exhibits impressive generalization capabilities and has been successfully applied to various tasks, including fine-grained art classification, image generation, zero-shot video retrieval, event classification, and visual commonsense reasoning. In our experiments, we also employed CLIP to encode image-text pairs as composition vectors.

\vspace{0.1cm}
\noindent\textbf{MPC.} 
MPC, which stands for Multimodal Probabilistic Composer, is a model designed to encode multiple modalities from diverse visual and textual sources \cite{neculai2022probabilistic}. It utilizes a probabilistic rule to combine probabilistic embeddings and employs a probabilistic similarity metric to measure the distance between them. The functioning of MPC is as follows: given information from various visual or textual modalities, it first learns probabilistic embeddings for each modality. These embeddings are then merged using a probabilistic composer, resulting in a probabilistic compositional embedding. This embedding is subsequently matched with the probabilistic embedding of the desired image by minimizing a probabilistic distance metric. MPC is capable of processing more than two queries by applying the probabilistic composer to a set of probabilistic embeddings. In our experiments, we employed MPC to encode image-text-image triples as composition vectors.

\subsection{High-dimensional Vector Search}
\label{appendix: vector search}
Vector search is a fundamental task with applications across various domains \cite{DPG, LiZAH20}, and it has received significant attention in recent years due to advancements in representation learning methods \cite{HNSW, Manu_zilliz}.
However, exact vector search can be computationally expensive, prompting researchers to focus on developing approximate techniques that strike a balance between accuracy and efficiency using vector indexes \cite{zhang2022leqat, PQ, HVS, li2022deep}.
Current vector search methods can be categorized into four types based on how the index is constructed: tree-based methods \cite{DasguptaF08, LuWWK20, MujaL14}, quantization-based methods \cite{PQ, ScaNN, AndreKS15}, hashing-based methods \cite{HuangFZFN15, GongWOX20, LiZSWT020}, and proximity graph-based methods \cite{HNSW, NSG, NSSG}. Recent works \cite{DPG, NSSG, zhao2023towards} have demonstrated that proximity graph-based methods achieve a favorable trade-off between accuracy and efficiency, making them well-suited for handling large-scale vector search tasks.

\begin{figure}
  \setlength{\abovecaptionskip}{0cm}
  \setlength{\belowcaptionskip}{0cm}
  \centering
  \footnotesize
  \stackunder[0.5pt]{\includegraphics[scale=0.28]{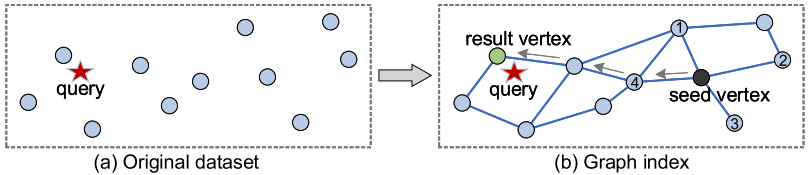}}{}
  \newline
  \caption{An example of indexing and search based on proximity graph \cite{graph_survey_vldb2021}.}
  \label{fig: pg-based index}
\end{figure}

\subsection{Proximity Graph-Based Index Algorithm}
\label{appendix: pg-based index}
Proximity graph-based algorithms have gained popularity for vector similarity search, particularly in high-dimensional spaces, due to their ability to strike a balance between efficiency and accuracy by capturing neighbor relationships between vectors \cite{NSG}. Major high-tech companies like Microsoft \cite{DiskANN} and Alibaba \cite{NSG} utilize these algorithms. To enable online query serving, an offline proximity graph index needs to be built on the dataset of feature vectors. This graph consists of vertices representing the vector data points and edges representing pairwise similarities or distances between vectors. Different algorithms, such as NSG \cite{NSG} or KGraph \cite{NNDescent}, employ various graph construction methods.

In a recent survey \cite{graph_survey_vldb2021}, a comprehensive analysis of proximity graph-based index algorithms is provided, including their performance, strengths, and potential pitfalls. Fig. \ref{fig: pg-based index} illustrates an example of finding the nearest vertex to a query vector $q$ using a proximity graph index. The process begins with a seed vertex (the black vertex), which can be randomly selected or fixed \cite{graph_survey_vldb2021}. It then visits its neighbors and computes their distances to $q$. Vertex 4 is chosen as the next visiting vertex because it is the closest among the seed's neighbors. This process continues until it reaches the green vertex, which has no neighbors closer to $q$ than itself. The search process relies on various factors, such as the seed acquisition strategy \cite{HVS} and the routing technique \cite{LiZAH20}. 

It is worth noting the following remarks:

\noindent\textit{(1) Lack of Theoretical Guarantee}. Although state-of-the-art proximity graph index algorithms lack theoretical guarantees \cite{graph_survey_vldb2021, HNSW}, their superiority in real-world scenarios has been validated by numerous research works \cite{graph_survey_vldb2021, HNSW, NSG, DPG, HVS, NSSG, DiskANN} and industrial applications \cite{annbenchmark, lanns}.

\noindent\textit{(2) Flexibility and Customization in {\name}}. In {\name}, we have designed a general pipeline that allows components from existing proximity graph algorithms to be easily integrated. Moreover, our pipeline supports custom-optimized components, which can inspire further research and experimentation.

Overall, proximity graph-based indexes provide an effective solution for vector similarity search, and their practical performance has been well-established in real-world scenarios.

\subsection{Motivation for Vector Weight Learning}
\label{appendix: motivation for vector weight learning}
In our framework, we aim to combine $m$ vectors of an object with $m$ modalities by assigning weights to each vector, resulting in a concatenated vector. This concept is inspired by similar cases, such as calculating multi-metric distance in multi-metric spaces \cite{zhu2022desire}. However, determining the importance or relevance of different vectors is a challenging task. Current methods often rely on user-defined weights \cite{zhu2022desire, franzke2016indexing}, which has two limitations.

\noindent\textit{(1) Lack of User-friendliness}. Assigning proper weights to different vectors is not user-friendly since users may not have the necessary knowledge or understanding to determine appropriate weights. This manual weight assignment process can be subjective and may not reflect the true importance of each vector. Based on our experiments, we have observed that different weights significantly affect the recall rate of {\problem}, as shown in Fig. \ref{fig: ablation weight learning} of the main text.

\noindent\textit{(2) Inapplicability for Offline Index Construction}. To build a fused index, as proposed in our paper, the weights for measuring the similarity between objects need to be known in advance during the offline index construction phase. However, relying on user-defined weights obtained online is not suitable for this purpose. Online weight assignment may introduce inconsistency and hinder the offline index construction process, which requires a consistent set of weights.

To address these limitations, there is a need for an automated approach to learn the vector weights that overcomes the user-defined weight assignment challenge and enables efficient offline index construction. By automatically learning the weights, we can ensure that the similarity computation process captures the true relevance and contribution of each vector, making it more objective and reliable.

\begin{figure}
  \setlength{\abovecaptionskip}{0.1cm}
  \setlength{\belowcaptionskip}{0cm}
  \centering
  \footnotesize
  \stackunder[0.5pt]{\includegraphics[scale=0.19]{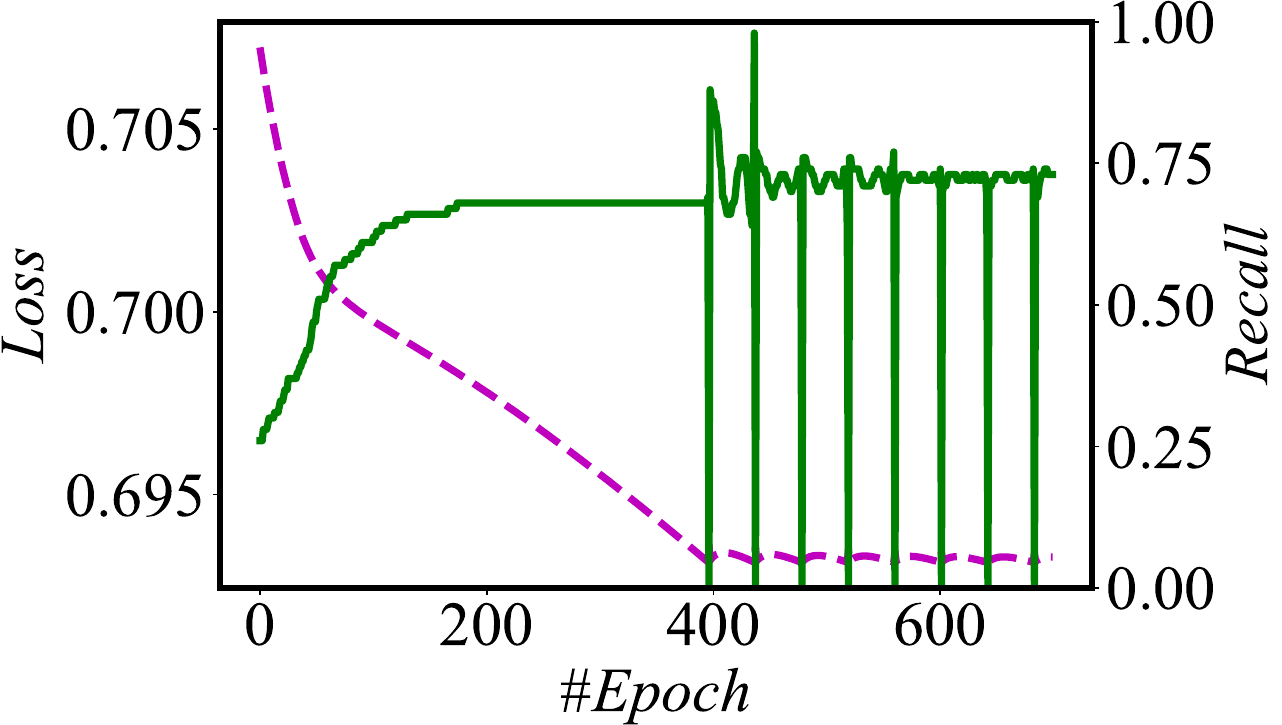}}{(a) $|N^{-}|=1$}\hspace{0.2cm}
  \stackunder[0.5pt]{\includegraphics[scale=0.19]{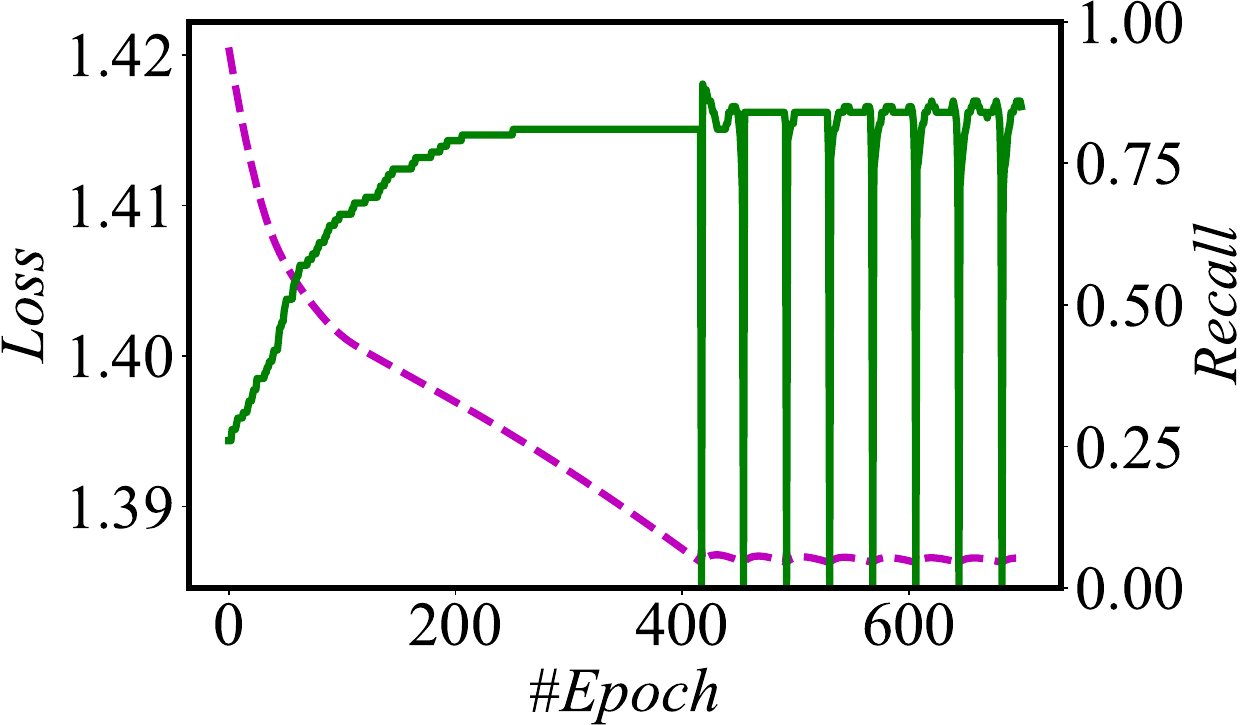}}{(b) $|N^{-}|=2$}
  \newline
  \stackunder[0.5pt]{\includegraphics[scale=0.19]{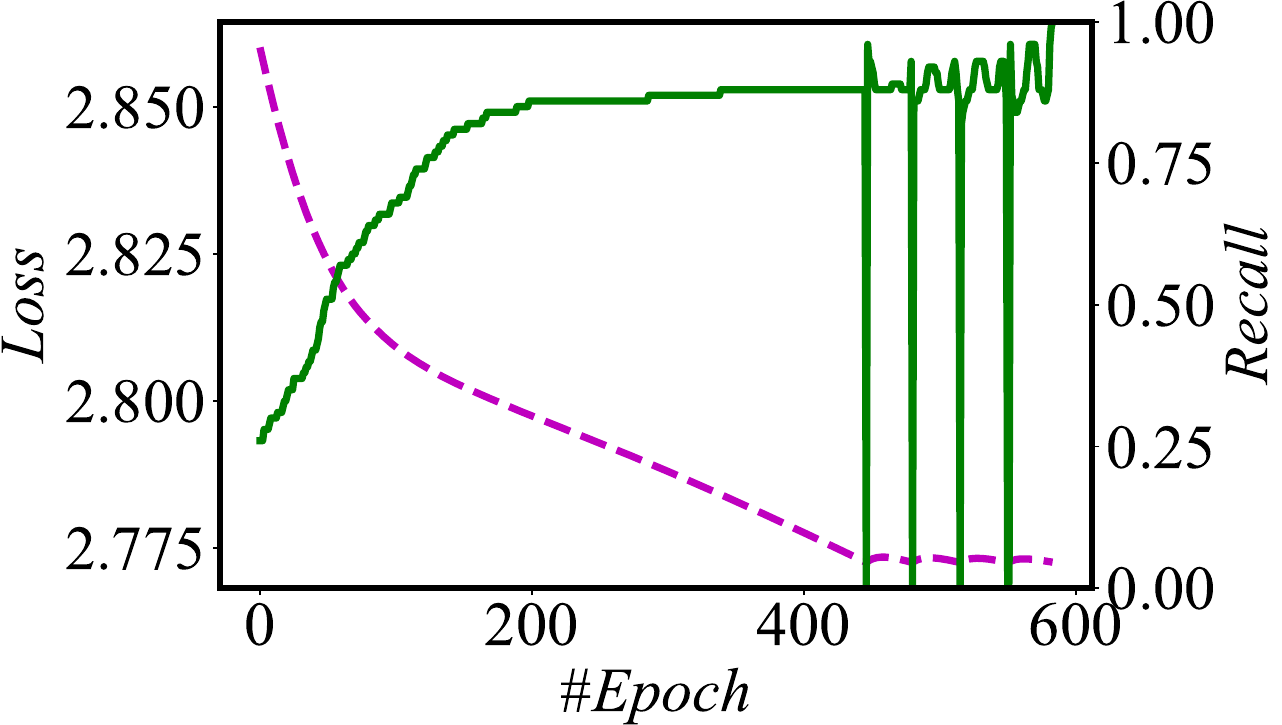}}{(c) $|N^{-}|=4$}\hspace{0.2cm}
  \stackunder[0.5pt]{\includegraphics[scale=0.19]{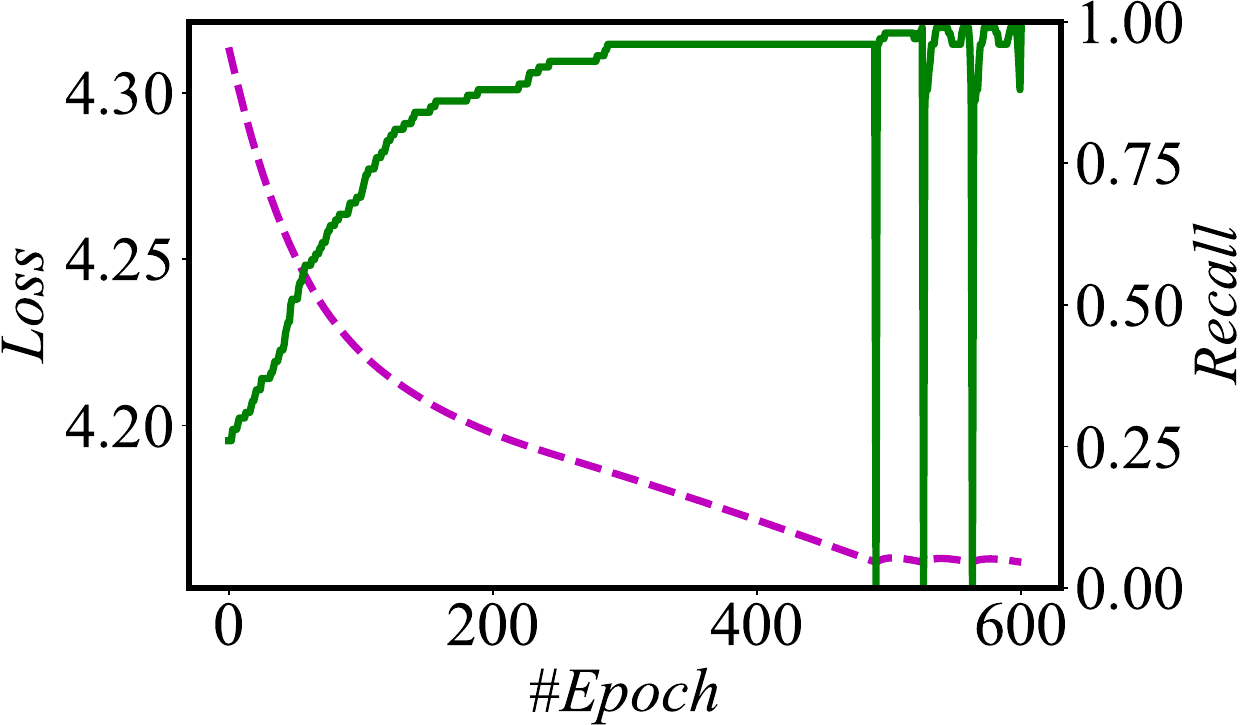}}{(d) $|N^{-}|=6$}
  \newline
  \stackunder[0.5pt]{\includegraphics[scale=0.19]{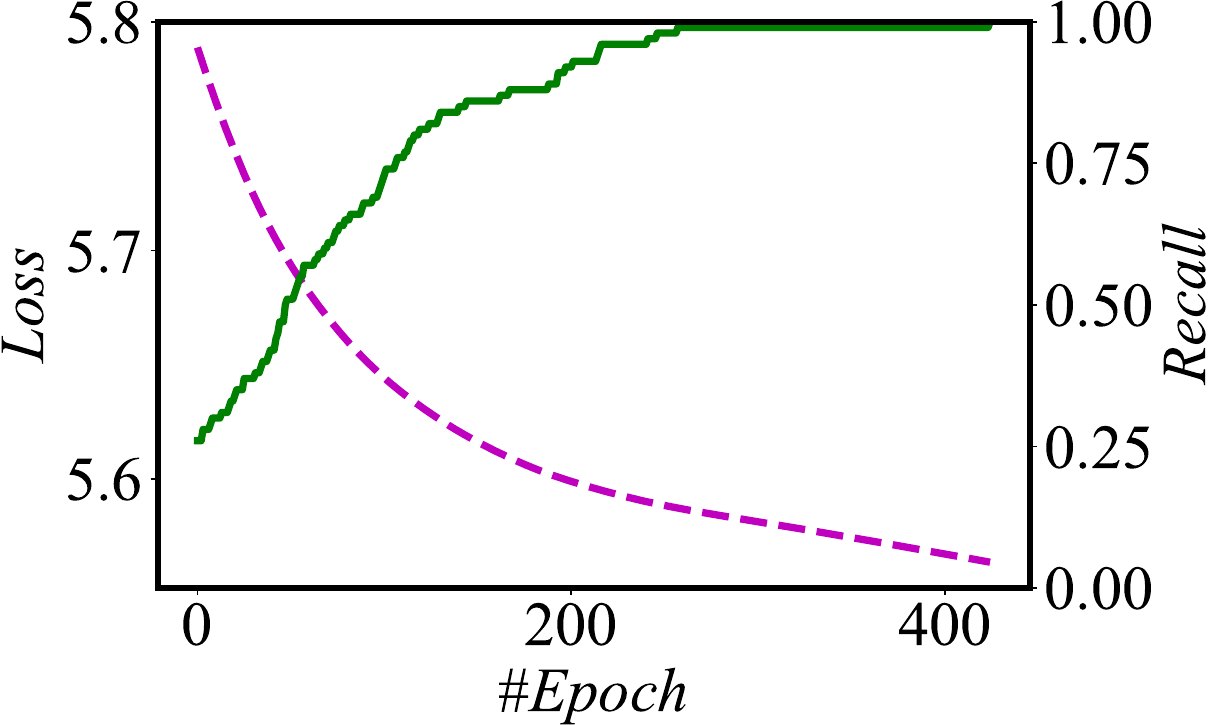}}{(e) $|N^{-}|=8$}\hspace{0.2cm}
  \stackunder[0.5pt]{\includegraphics[scale=0.19]{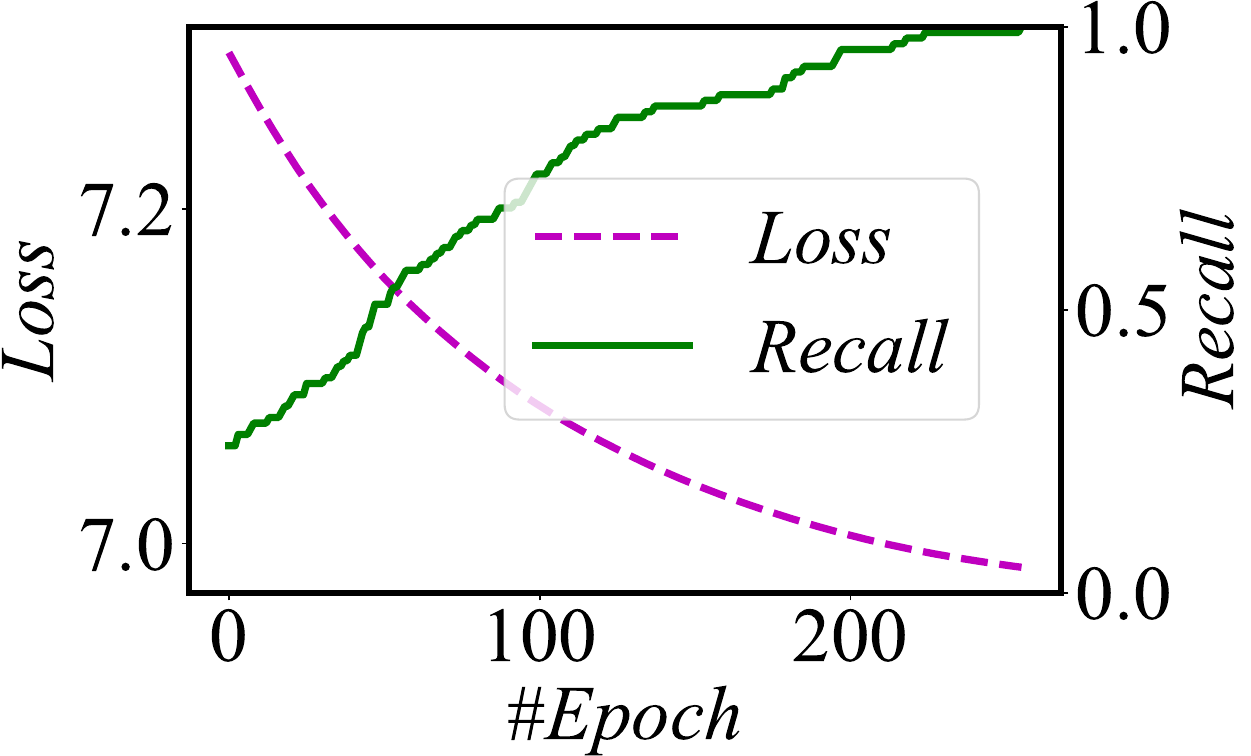}}{(f) $|N^{-}|=10$}
  \newline
  \caption{Effect of different number of negatives in vector weight learning.}
  \label{fig: appendix example number}
\end{figure}

\subsection{More Details of Setup and Parameters}
\label{appendix: setup and parameters}
We adopt the same training hyperparameters as the original papers of the encoders to obtain the embedding vectors. The encoder configuration remains consistent across all three frameworks. Our training pipeline is implemented in PyTorch for the vector weight learning model, and we utilize the Pybind library to invoke the vector similarity search kernel written in C++. The learning rate was set to 0.002, and the training was conducted for 700 iterations by default.

The indexing components and search codes were implemented in C++ using CGraph \cite{CGraph} and compiled with g++6.5. We built all indexes in parallel using 64 threads and executed the search procedure using a single thread, following the common setting in related work \cite{graph_survey_vldb2021,NSG}.

All experiments were conducted on a Linux server equipped with an Intel(R) Xeon(R) Gold 6248R CPU running at 3.00GHz and 755G memory. We performed three repeated trials and reported the average results for all evaluation metrics. The learned weights can be found in Appendix \ref{appendix: weights setting}.

\subsection{Study of the Number of Negative Examples}
\label{appendix: number negative examples}
In our study, we investigate the impact of the number of negative examples (e.g., $|N^{-}|$) on model training. We evaluate this effect by examining the loss and recall rate curves on the hard negatives using the ImageText1M dataset. As shown in Fig. \ref{fig: appendix example number}, we observe that increasing the number of negative examples generally leads to better training results. This means that including more negative examples during training helps improve the model's performance in identifying positive examples and distinguishing them from hard negatives.

However, it is important to consider the trade-off between training quality and training efficiency. As the number of negative examples increases, the training time also tends to increase. Therefore, it is necessary to find a proper balance between the training quality and efficiency by selecting an appropriate value for $|N^{-}|$.

By analyzing the loss and recall rate curves for different values of $|N^{-}|$, we can determine the optimal number of negative examples that achieves a satisfactory training effect without excessively increasing the training time. This ensures that the model is trained effectively while taking into account practical considerations.

\subsection{Parameters for Fused Index}
\label{appendix: index parameters}
The construction of the fused index relies on two important parameters: the maximum number of neighbors ($\gamma$) and the maximum iterations ($\varepsilon$). These parameters have an impact on the index size, index build time, and search performance. In this section, we evaluate these parameters in more detail.

\begin{figure}
  \setlength{\abovecaptionskip}{0.1cm}
  \setlength{\belowcaptionskip}{0cm}
  \centering
  \footnotesize
  \stackunder[0.5pt]{\includegraphics[scale=0.22]{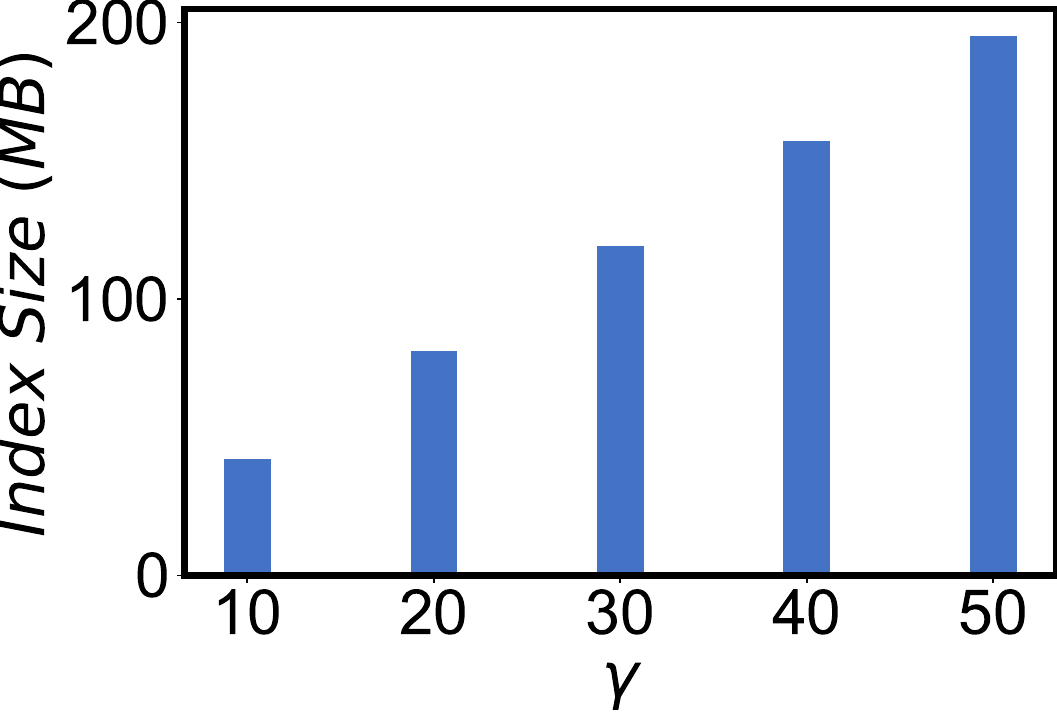}}{(a) ImageText1M}\hspace{0.2cm}
  \stackunder[0.5pt]{\includegraphics[scale=0.22]{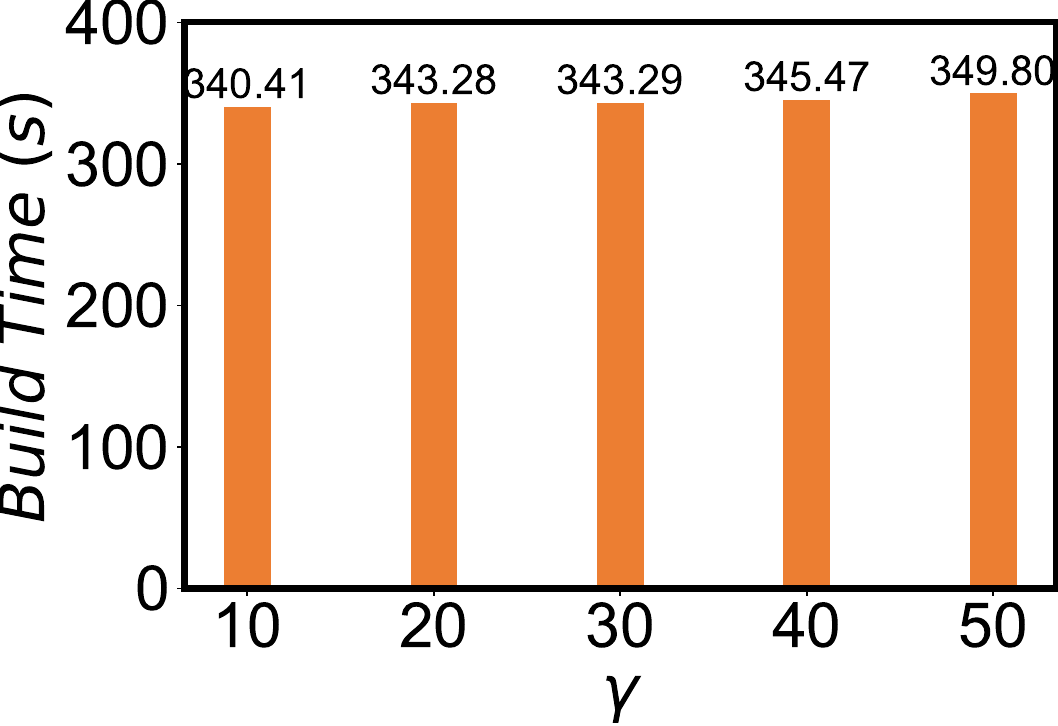}}{(b) ImageText1M}
  \newline
  \caption{Index size and index build time under different values of $\gamma$.}
  \label{fig: appendix neighbors parameters index}
\end{figure}

\begin{figure}[!t]
  \setlength{\abovecaptionskip}{0.1cm}
  \setlength{\belowcaptionskip}{0cm}
  \centering
  \footnotesize
  \stackunder[0.5pt]{\includegraphics[scale=0.22]{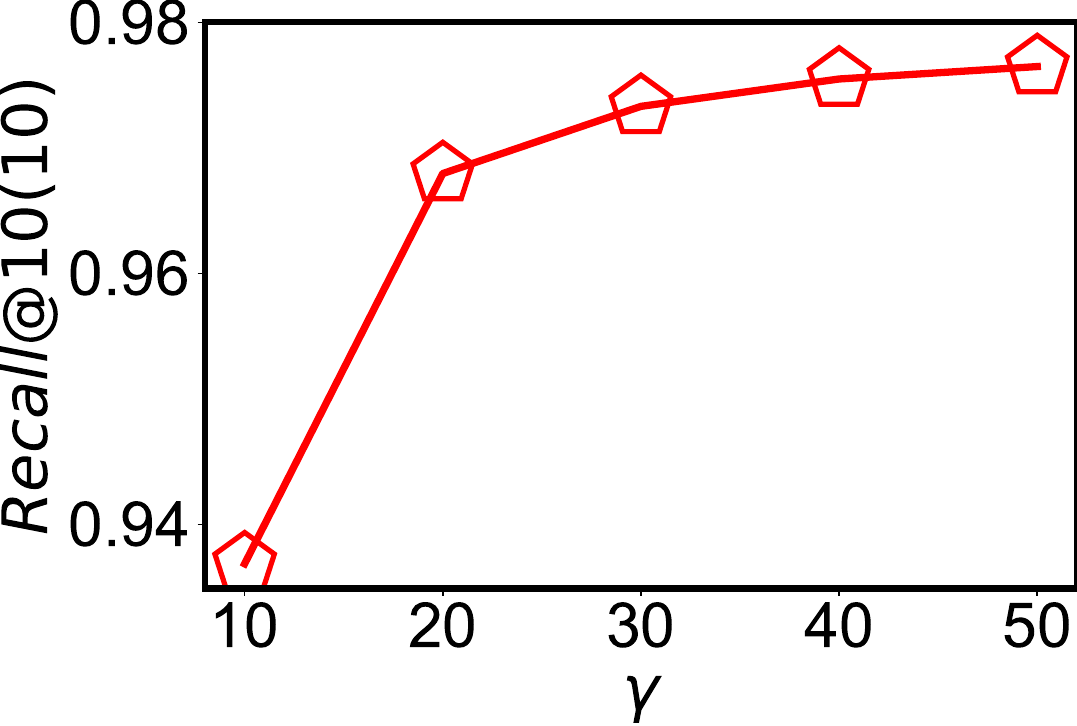}}{(a) ImageText1M ($l$=4000)}\hspace{0.2cm}
  \stackunder[0.5pt]{\includegraphics[scale=0.22]{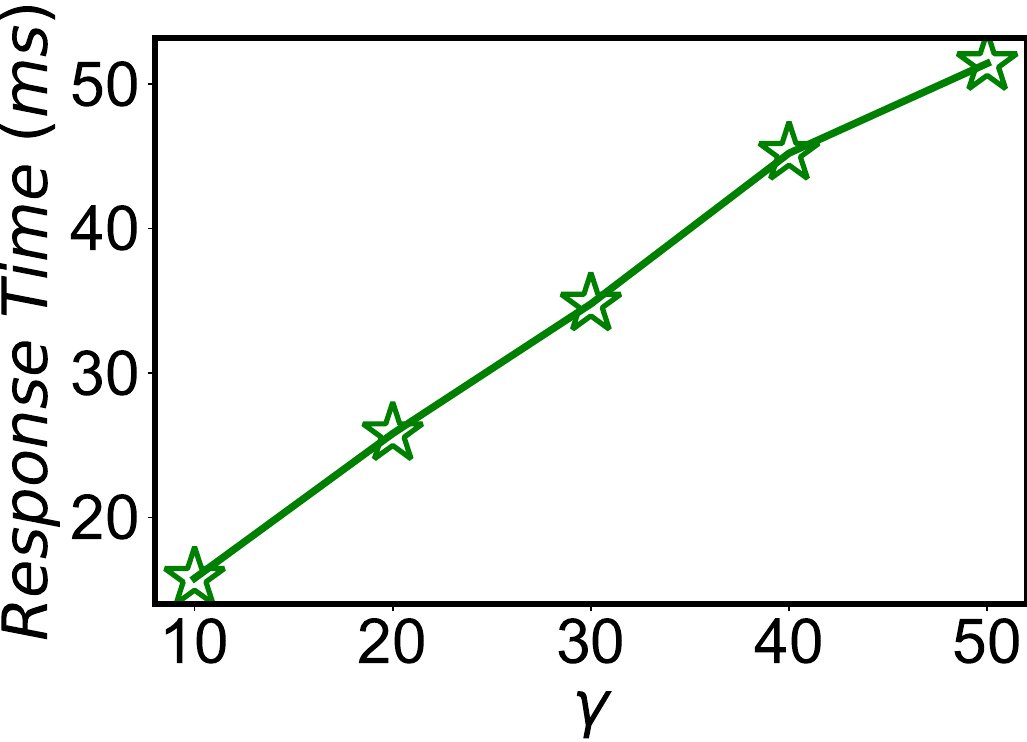}}{(b) ImageText1M ($l$=4000)}
  \newline
  \caption{Search performance under different values of $\gamma$.}
  \label{fig: appendix neighbors parameters search}
\end{figure}

\noindent\textit{Maximum Number of Neighbors ($\gamma$)}. Fig. \ref{fig: appendix neighbors parameters index} illustrates the relationship between $\gamma$ and the index size as well as the index build time. As $\gamma$ increases, both the index size and the index build time also increase. This is because a larger $\gamma$ requires handling more vertices in the \textit{Initialization}, \textit{Candidate acquisition}, and \textit{Neighbor selection} components, which increases the complexity of index construction.

In Fig. \ref{fig: appendix neighbors parameters search}, we maintain the other parameters unchanged (including the search parameters $k$ and $l$) and analyze the search performance for different values of $\gamma$. The result shows that the recall rate improves as $\gamma$ increases. This is because a larger $\gamma$ allows for visiting more neighbors of a vertex during the search process, leading to improved search accuracy. However, when $\gamma$ becomes very large, the computation of vector distances increases, resulting in lower efficiency. Therefore, it is crucial to strike a balance between accuracy and efficiency by adjusting $\gamma$ in practical scenarios. In our experiments, we set $\gamma$ to a default value of 30.

\begin{table}[!tb]
  \centering
  \setlength{\abovecaptionskip}{0.05cm}
  \setlength{\belowcaptionskip}{0.2cm}
  \setstretch{0.8}
  \fontsize{6.5pt}{3.3mm}\selectfont
  \caption{Graph quality under different number of iterations.}
  \label{tab: appendix graph quality iteration}
  \setlength{\tabcolsep}{.008\linewidth}{
  \begin{tabular}{|l|l|l|l|}
    \hline
    \textbf{\# Iterations ($\varepsilon$) $\downarrow$} & ImageText1M & AudioText1M & VideoText1M \\
    \hline
    1 & 0.0094 & 0.0088 & 0.0096 \\
    \hline
    2 & 0.7795 & 0.7945 & 0.7842 \\
    \hline
    3 & 0.9900 & 0.9900 & 0.9900 \\
    \hline
  \end{tabular}
  }
\end{table}

\noindent\textit{Maximum iterations ($\varepsilon$)}. The graph quality is defined as the mean ratio of $\gamma$ neighbors of a vertex over the top-$\gamma$ nearest neighbors based on joint similarity \cite{graph_survey_vldb2021}. Tab. \ref{tab: appendix graph quality iteration} presents the changes in graph quality with different numbers of iterations ($\varepsilon$). The results demonstrate that graph quality increases with $\varepsilon$ and reaches a value close to 1 when $\varepsilon$ is set to 3. Therefore, we set $\varepsilon$ to a default value of 3 in all our experiments.

By evaluating these parameters, we gain insights into their impact on the fused index construction. The findings suggest that selecting appropriate values for $\gamma$ and $\varepsilon$ can optimize index size, build time, and search performance, striking a balance between accuracy and efficiency in practice.

\subsection{Parameter for Joint Search}
\label{appendix: search parameters}
The parameter that affects the accuracy and efficiency of the search process in the joint search algorithm (please refer to Algorithm \hyperref[alg: search]{2} of the main text) is the result set size, denoted as $l$. In this section, we evaluate how the recall rate and response time change with different values of $l$.

Tab. \ref{tab: appendix result set size search} presents the results of the evaluation, showing the recall rate and response time for different $l$ values. It can be observed that both the recall rate and response time increase as $l$ becomes larger. The increase in recall rate is expected because a larger $l$ allows for visiting more vertices and considering more potential results during the search process. This leads to a higher likelihood of retrieving relevant objects, resulting in an improved recall rate. However, the response time also increases with larger $l$ values. This is because a larger result set size requires visiting more vertices and performing more vector calculations, which adds computational overhead and increases the overall response time. 
Therefore, selecting an appropriate value for $l$ involves a trade-off between recall rate and response time. A larger $l$ can improve recall but at the cost of increased response time. Researchers and practitioners need to consider the specific requirements of their application and strike a balance between accuracy and efficiency when choosing the value of $l$.

\begin{table}[!tb]
  \centering
  \setlength{\abovecaptionskip}{0.05cm}
  \setlength{\belowcaptionskip}{0.2cm}
  \setstretch{0.8}
  \fontsize{6.5pt}{3.3mm}\selectfont
  \caption{Search performance under different values of $l$ ($\gamma=30$).}
  \label{tab: appendix result set size search}
  \setlength{\tabcolsep}{.008\linewidth}{
  \begin{tabular}{|l|l|l|l|l|l|}
    \hline
    \textbf{$l$ $\rightarrow$} & 700 & 1000 & 1500 & 2000 & 4000 \\
    \hline
    $\boldsymbol{Recall@10(10)}$ & 0.506100 & 0.637260 & 0.766190 & 0.856250 & 0.973310 \\
    \hline
    \textbf{Response Time (ms)} & 5 & 7 & 11 & 15 & 35 \\
    \hline
  \end{tabular}
  }
\end{table}

\subsection{Datasets}
\label{appendix: datasets}
We utilize nine datasets with varying modalities and cardinalities obtained from public sources, as presented in Tab. \ref{tab: Dataset} of the main text. CelebA \cite{celeba}, MIT-States \cite{mit-states}, Shopping \cite{shopping}, and MS-COCO \cite{neculai2022probabilistic} are four real-world multimodal datasets \cite{jandial2022sac,vo2019composing}. For instance, CelebA comprises two modalities for each object: a facial image and a corresponding text description. In these datasets, we employ the original query samples, and each query contains one or more ground-truth objects. Since there are no publicly available datasets with up to four modalities, we simulated two additional modalities for CelebA using different encoders. This led to the creation of the CelebA+ dataset, with four vectors for each object, simulating four modalities. To evaluate performance at a large scale, we added the text modality to four single-modal datasets (DEEP and SIFT for images \cite{deep16m,sift1m}, MSONG for audio \cite{msong1m}, and UQ-V for video \cite{uqv1m}) using the same method described in \cite{Milvus_sigmod2021}. As a result, we formed four large-scale multimodal datasets: ImageText16M, ImageText1M, AudioText1M, and VideoText1M.

\vspace{0.1cm}
\noindent\textbf{CelebA.} 
CelebFaces Attributes Dataset (CelebA) \cite{celeba} is a comprehensive dataset of celebrity images, consisting of over 200,000 images. Each image is annotated with 40 attributes, 5 landmark locations, and a face identity. The dataset covers a wide range of poses and backgrounds, offering rich diversity, quantity, and annotations. CelebA can be utilized for various computer vision tasks, including face attribute recognition, face recognition, face detection, landmark localization, and face editing \& synthesis.

\vspace{0.1cm}
\noindent\textbf{MIT-States.} 
MIT-States \cite{mit-states} is a dataset comprising approximately 60,000 images, each labeled with an object/noun and a state/adjective (e.g., ``red tomato'' or ``new camera''). The dataset includes 245 nouns and 115 adjectives, with an average of around 9 adjectives per noun. MIT-States is commonly employed to evaluate image retrieval and image classification tasks in the field of computer vision.

\vspace{0.1cm}
\noindent\textbf{Shopping.} 
The Shopping100k dataset \cite{shopping} consists of 101,021 images of clothing items extracted from various e-commerce providers for fashion studies. It was developed to address limitations in existing fashion-related datasets, which often feature posed images with occlusion issues. Each image in the dataset is represented with general and special attributes, with the special attributes being more suitable for attribute manipulation and fashion searches. Shopping encompasses various categories, such as T-shirts and bottoms.

\vspace{0.1cm}
\noindent\textbf{MS-COCO.} 
The MS-COCO dataset \cite{neculai2022probabilistic} is a widely used dataset in computer vision research, standing for Microsoft Common Objects in Context. It comprises over 330,000 images with more than 2.5 million labeled object instances, annotated with object bounding boxes and belonging to 80 object categories. The dataset is designed to facilitate research on object detection, segmentation, captioning, and other related tasks, serving as a challenging benchmark for evaluating computer vision models due to its scale, diversity, and complexity.

\vspace{0.1cm}
\noindent\textbf{ImageText1M.} 
ImageText1M is a semi-synthetic dataset that combines real-world images with text. It consists of 1 million SIFT vectors with a dimension of 128 \cite{sift1m}. Each vector represents an image and is augmented with a text modality to form a multimodal dataset.

\vspace{0.1cm}
\noindent\textbf{AudioText1M.} 
AudioText1M is a semi-synthetic dataset that combines real-world audio with text. It comprises 1 million contemporary popular music tracks, each represented by 420 dimensions of audio features and metadata \cite{msong1m}. Similar to ImageText1M, each audio vector is paired with a text modality to create a multimodal dataset.

\vspace{0.1cm}
\noindent\textbf{VideoText1M.} 
VideoText1M is a semi-synthetic dataset that combines real-world videos with text. It extracts 256 dimensions of local features from keyframes of each video \cite{uqv1m}. These video vectors are then combined with a text modality, resulting in a multimodal dataset.

\vspace{0.1cm}
\noindent\textbf{ImageText16M.} 
ImageText16M is a semi-synthetic dataset that merges real-world images with text. It encompasses 16 million data points, with each point represented by 96 dimensions of deep neural codes derived from a convolutional neural network \cite{deep16m}. Similar to the other multimodal datasets, a text modality is added to each image vector.

\begin{table}[!tb]
  \centering
  \setlength{\abovecaptionskip}{0.05cm}
  \setlength{\belowcaptionskip}{0cm}
  \setstretch{0.8}
  \fontsize{6.5pt}{3.3mm}\selectfont
  \caption{Output weights of module for MIT-States dataset.}
  \label{tab: appendix weights setting mit-states}
  \setlength{\tabcolsep}{.008\linewidth}{
  \begin{tabular}{|l|l|l|l|}
    \hline
    \textbf{Encoder $\downarrow$} & $\omega_0^2$ (modality 0) & $\omega_1^2$ (modality 1) \\
    \hline
    ResNet17+LSTM & 0.3000 & 0.7000  \\
    \hline
    ResNet50+LSTM & 0.0012 & 1.4291 \\
    \hline
    ResNet17+Transformer & 0.1172 & 0.2669 \\
    \hline
    ResNet50+Transformer & 0.5000 & 0.5000 \\
    \hline
    TIRG+LSTM & 0.5000 & 0.5000 \\
    \hline
    TIRG+Transformer & 0.0295 & 0.0224 \\
    \hline
    CLIP+LSTM & 0.5000 & 0.5000 \\
    \hline
    CLIP+Transformer & 0.0670 & 0.0432 \\
    \hline
  \end{tabular}
  }
\end{table}

\begin{table}[!tb]
  \centering
  \setlength{\abovecaptionskip}{0.05cm}
  \setlength{\belowcaptionskip}{0.2cm}
  \setstretch{0.8}
  \fontsize{6.5pt}{3.3mm}\selectfont
  \caption{Output weights of module for CelebA dataset.}
  \label{tab: appendix weights setting celeba}
  \setlength{\tabcolsep}{.008\linewidth}{
  \begin{tabular}{|l|l|l|l|}
    \hline
    \textbf{Encoder $\downarrow$} & $\omega_0^2$ (modality 0) & $\omega_1^2$ (modality 1) \\
    \hline
    ResNet17+Encoding & 0.0007 & 0.9526  \\
    \hline
    ResNet50+Encoding & 0.0848 & 1.1855 \\
    \hline
    TIRG+Encoding & 0.1064 & 0.6414 \\
    \hline
    CLIP+Encoding & 0.1089 & 0.8551 \\
    \hline
  \end{tabular}
  }
\end{table}

\begin{table}[!tb]
  \centering
  \setlength{\abovecaptionskip}{0.05cm}
  \setlength{\belowcaptionskip}{0.2cm}
  \setstretch{0.8}
  \fontsize{6.5pt}{3.3mm}\selectfont
  \caption{Output weights of module for Shopping dataset.}
  \label{tab: appendix weights setting shopping}
  \setlength{\tabcolsep}{.008\linewidth}{
  \begin{tabular}{|l|l|l|l|}
    \hline
    \textbf{Encoder $\downarrow$} & $\omega_0^2$ (modality 0) & $\omega_1^2$ (modality 1) \\
    \hline
    ResNet17+Encoding & 0.0262 & 1.2124  \\
    \hline
    TIRG+Encoding & 0.0092 & 1.2042 \\
    \hline
  \end{tabular}
  }
\end{table}

\begin{table}[!tb]
  \centering
  \setlength{\abovecaptionskip}{0.05cm}
  \setlength{\belowcaptionskip}{0.2cm}
  \setstretch{0.8}
  \fontsize{6.5pt}{3.3mm}\selectfont
  \caption{Output weights of module for MS-COCO dataset.}
  \label{tab: appendix weights setting coco}
  \setlength{\tabcolsep}{.008\linewidth}{
  \begin{tabular}{|l|l|l|l|l|}
    \hline
    \textbf{Encoder $\downarrow$} & $\omega_0^2$ & $\omega_1^2$ & $\omega_2^2$ \\
    \hline
    MPC+GRU+ResNet50 & 0.0083 & 0.0342 & 0.0123  \\
    \hline
    ResNet50+GRU+ResNet50 & 0.0091 & 0.0233 & 0.0144  \\
    \hline
  \end{tabular}
  }
\end{table}

\begin{table}[!tb]
  \centering
  \setlength{\abovecaptionskip}{0.05cm}
  \setlength{\belowcaptionskip}{0.2cm}
  \setstretch{0.8}
  \fontsize{6.5pt}{3.3mm}\selectfont
  \caption{Output weights of module for CelebA+ dataset.}
  \label{tab: appendix weights setting shopping}
  \setlength{\tabcolsep}{.008\linewidth}{
  \begin{tabular}{|l|l|l|l|l|l|}
    \hline
    \textbf{Encoder $\downarrow$} & $\omega_0^2$ & $\omega_1^2$ & $\omega_2^2$ & $\omega_3^2$ \\
    \hline
    CLIP+Encoding+ResNet17+ResNet50 & 0.4092 & 3.1363 & 0.0721 & 0.0290  \\
    \hline
  \end{tabular}
  }
\end{table}

\begin{table}[!tb]
  \centering
  \setlength{\abovecaptionskip}{0.05cm}
  \setlength{\belowcaptionskip}{0.2cm}
  \setstretch{0.8}
  \fontsize{6.5pt}{3.3mm}\selectfont
  \caption{Output weights of module for ImageText1M, AudioText1M, VideoText1M, and ImageText16M datasets.}
  \label{tab: appendix weights setting semi-synthetic datasets}
  \setlength{\tabcolsep}{.008\linewidth}{
  \begin{tabular}{|l|l|l|l|}
    \hline
    \textbf{Dataset $\downarrow$} & $\omega_0^2$ (modality 0) & $\omega_1^2$ (modality 1) \\
    \hline
    ImageText1M & 0.1199 & 0.5572  \\
    \hline
    AudioText1M & 0.0453 & 0.8589  \\
    \hline
    VideoText1M & 0.3106 & 0.4440  \\
    \hline
    ImageText16M & 0.1123 & 0.8742  \\
    \hline
  \end{tabular}
  }
\end{table}

\subsection{Weights Setting}
\label{appendix: weights setting}
In {\name}, we employ a vector weight learning module to capture the significance of various modalities. The specific weights utilized for constructing indexes and performing query processing on different datasets and encoders are presented in Tab. \ref{tab: appendix weights setting mit-states} to \ref{tab: appendix weights setting semi-synthetic datasets}.

\begin{table}[!tb]
  \centering
  \setlength{\abovecaptionskip}{0.05cm}
  \setlength{\belowcaptionskip}{0.2cm}
  \setstretch{0.8}
  \fontsize{6.5pt}{3.3mm}\selectfont
  \caption{Search accuracy with target modality input only.}
  \label{tab: accuracy unimodal}
  \setlength{\tabcolsep}{.008\linewidth}{
  \begin{tabular}{|c|l|l|l|l|}
    \hline
    \textbf{Dataset} & \textbf{Encoder} & \textit{\textbf{Recall@1(1)}} & \textit{\textbf{Recall@5(1)}} & \textit{\textbf{Recall@10(1)}} \\
    \hline
    \multirow{2}*{\textbf{MIT-States}} & ResNet17 & 0.0268 & 0.1103 & 0.1822 \\
    \cline{2-5}
    ~ & ResNet50 & 0.0363 & 0.1393 & 0.2257 \\
    \hline
    \multirow{2}*{\textbf{CelebA}} & ResNet17 & 0.1499 & 0.4055 & 0.4913 \\
    \cline{2-5}
    ~ & ResNet50 & 0.1475 & 0.3785 & 0.4519 \\
    \hline
    \multirow{1}*{\tabincell{c}{\textbf{Shopping} \textbf{(T-shirt)}}} & ResNet17 & 0 & 0.0192 & 0.0399  \\
    \hline
  \end{tabular}
  }
\end{table}

\begin{table}[!tb]
  \centering
  \setlength{\abovecaptionskip}{0.05cm}
  \setlength{\belowcaptionskip}{0.2cm}
  \setstretch{0.8}
  \fontsize{6.5pt}{3.3mm}\selectfont
  \caption{Search accuracy with auxiliary modality only.}
  \label{tab: accuracy auxiliary modality}
  \setlength{\tabcolsep}{.008\linewidth}{
  \begin{tabular}{|c|l|l|l|l|}
    \hline
    \textbf{Dataset} & \textbf{Encoder} & \textbf{Recall@1(1)} & \textbf{Recall@5(1)} & \textbf{Recall@10(1)} \\
    \hline
    \multirow{2}*{\textbf{MIT-States}} & LSTM & 0.2747 & 0.4343 & 0.4844 \\
    \cline{2-5}
    ~ & Transformer & 0.2601 & 0.2641 & 0.2824 \\
    \hline
    \multirow{1}*{\textbf{CelebA}} & Encoding & 0.0377 & 0.0936 & 0.1291 \\
    \hline
    \multirow{1}*{\tabincell{c}{\textbf{Shopping} \textbf{(T-shirt)}}} & Encoding & 0.0964 & 0.4126 & 0.5362  \\
    \hline
  \end{tabular}
  }
\end{table}

\subsection{Search Accuracy Using a Single Modality}
\label{appendix: accuracy one modality}
In this section, we analyze the search accuracy achieved when using only the target modality input or the auxiliary modality input. Tab. \ref{tab: accuracy unimodal} presents the search accuracy using only the target modality input, while Tab. \ref{tab: accuracy auxiliary modality} displays the search accuracy using only the auxiliary modality input, both evaluated on three real-world datasets. Generally, these unimodal approaches exhibit lower performance compared to methods that combine multiple modalities. However, in certain cases, they outperform the JE framework, which combines the features of all modality inputs into a joint embedding. One possible explanation for this observation is that the JE framework introduces a significant encoder error, while the auxiliary modality can accurately describe an object in some datasets. This finding further emphasizes the necessity of utilizing multiple vectors from different encoders to effectively represent an object.

\begin{table}[!tb]
  \centering
  \setlength{\abovecaptionskip}{0.05cm}
  \setlength{\belowcaptionskip}{0.2cm}
  \setstretch{0.8}
  \fontsize{6.5pt}{3.3mm}\selectfont
  \caption{Search accuracy on Shopping (Bottoms).}
  \label{tab: appendix accuracy shopping bottoms}
  \setlength{\tabcolsep}{.006\linewidth}{
  \begin{tabular}{|c|l|l|l|l|}
    \hline
    \textbf{Framework} & \textbf{Encoder} & \textbf{Recall@1(1)} & \textbf{Recall@5(1)} & \textbf{Recall@10(1)} \\
    \hline
    \textbf{JE} & TIRG & 0.0905 & 0.2715 & 0.3924  \\
    \hline
    \multirow{2}*{\textbf{MR}} & ResNet17+Encoding & 0.0107 & 0.0551 & 0.0995  \\
    \cline{2-5}
    ~ & TIRG+Encoding & 0.0596 & 0.2552 & 0.3850 \\
    \hline
    \multirow{2}*{\tabincell{c}{\boldsymbol{\name}\\ \textbf{(ours)}}} & ResNet17+Encoding & \textbf{0.4840($\uparrow$434.8\%)} & 0.7960 & 0.8887 \\
    \cline{2-5}
    ~ & TIRG+Encoding & 0.4784 & \textbf{0.8162($\uparrow$200.6\%)} & \textbf{0.8999($\uparrow$129.3\%)} \\
    \hline
  \end{tabular}
  }
\end{table}

\subsection{Search Accuracy on Shopping (Bottoms)}
\label{appendix: search accuracy shopping bottoms}
In this section, we present the search accuracy specifically for the ``bottoms'' category of the Shopping dataset, as shown in Tab. \ref{tab: appendix accuracy shopping bottoms}. It is worth noting that different categories within the Shopping dataset share the same output weights. This observation highlights the generalization capability of the vector weight learning module employed in our framework.

\subsection{Discussion of Accuracy, Efficiency, and Scalability}
\label{appendix: discussion}
\noindent\textbf{Accuracy.} In the context of {\problem}, unimodal search results tend to be inaccurate. Two baseline methods, namely Joint Embedding ({\bii}) and Multi-streamed Retrieval ({\bi}), aim to improve accuracy by incorporating multimodal information. {\bii} creates a composition vector by combining different modalities, while {\bi} merges the results of multiple unimodal searches. One way to optimize these baselines is to use {\bii} as one of the separate search methods within {\bi}. However, this optimization is still limited by {\bi}'s inability to capture the importance of different modalities. {\name} addresses this limitation by employing a hybrid fusion mechanism that combines modalities at multiple levels. By utilizing a weight-learning module, it effectively captures the importance of different modalities and achieves the highest accuracy.

\noindent\textbf{Efficiency.} Performing {\problem} involves computationally intensive tasks such as building an index and performing approximate queries. {\bi} builds an index for each modality and searches them separately, which can be slow in large-scale scenarios, especially when the number of query results increases. {\name} achieves higher efficiency and recall rate by constructing a fused index for different modal vectors and performing joint searches. The main reason for the efficiency difference lies in the quality of the graph index. In the main text, Fig. \ref{fig: index show} provides a visualization of three neighbors for an object in the CelebA dataset. In {\name}'s index, the vertex and its neighbors balance the importance of different modalities and exhibit better joint similarity. On the other hand, in {\bi}'s indexes, the vertex and its neighbors only consider similarity within a single modality.

\noindent\textbf{Scalability.} As the number of modalities and the data scale increase, {\bi} requires more and larger indexes, resulting in high storage costs and low indexing efficiency. In contrast, the size of {\name}'s index and its construction time grow slightly (almost logarithmically) with the data size and remain the same regardless of the number of modalities. It is important to note that as more data is added, the merging operation in {\bi} becomes more complex, limiting its search efficiency and accuracy. {\name} proves to be more efficient and accurate in large-scale scenarios and does not require a merging operation. Additionally, {\name} handles multiple modalities more effectively than {\bi}.

\noindent\textbf{Highlight.} We introduce the concept of modality importance mining in the context of {\problem}, thereby opening up new research avenues in multimodal search. We anticipate that this will drive advancements in representation learning, vector indexing, and search algorithms.

\subsection{Result Examples on Real-world Datasets}
\label{appendix: recall examples}

\begin{figure*}[!tbh]
  \centering
  \setlength{\abovecaptionskip}{0.1cm}
  \setlength{\belowcaptionskip}{-0.2cm}
  \includegraphics[width=\linewidth]{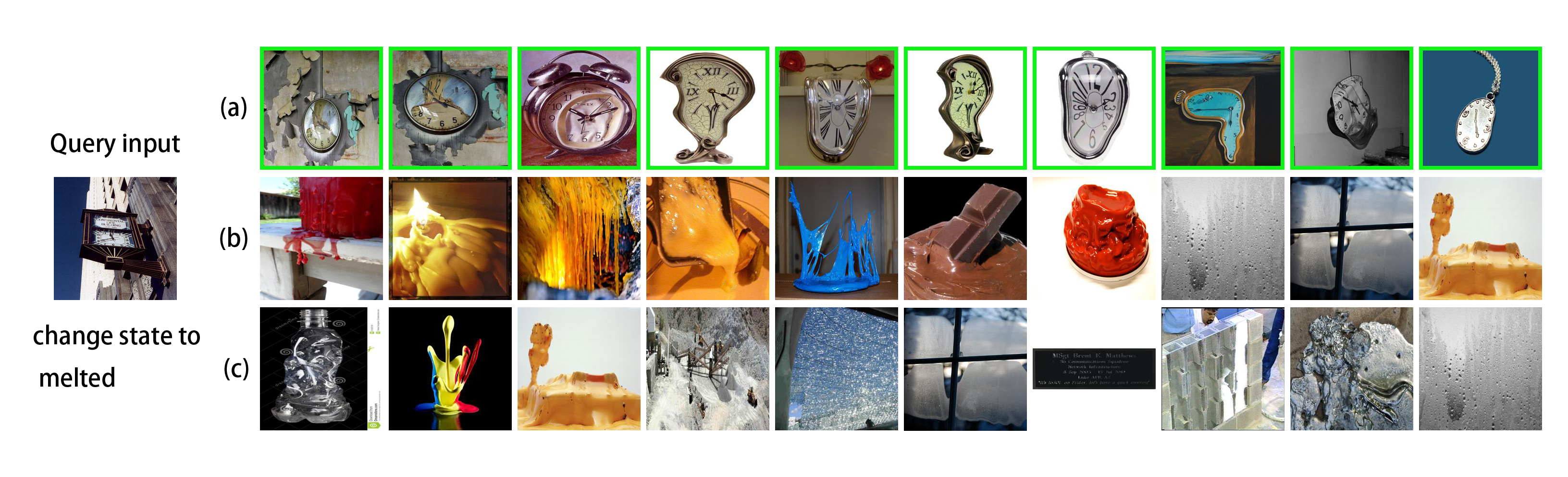}
  \caption{Some result examples of different frameworks with optimal encoders on MIT-States dataset. (a), (b), and (c) are the top-10 results of {\name}, {\bi}, and {\bii}, respectively. The green box marks the target objects.}
  \label{fig: mitstates_case1}
\end{figure*}

Fig. \ref{fig: mitstates_case1} shows the top-10 search results obtained using different frameworks for a query input consisting of a clock image and the text description ``change state to melted''. The results demonstrate that {\name} successfully retrieves all the ground-truth objects. However, {\bi} and {\bii} returns many unrelated objects. This discrepancy occurs because {\bi} combines dissimilar objects with low ranks from each candidate set, and the high encoding error of {\bii} leads to a larger inner product (IP) between the query input’s composition vector and the vectors of dissimilar objects in the target modality.

\begin{figure*}[!tbh]
  \centering
  \setlength{\abovecaptionskip}{0.1cm}
  \setlength{\belowcaptionskip}{-0.2cm}
  \includegraphics[width=\linewidth]{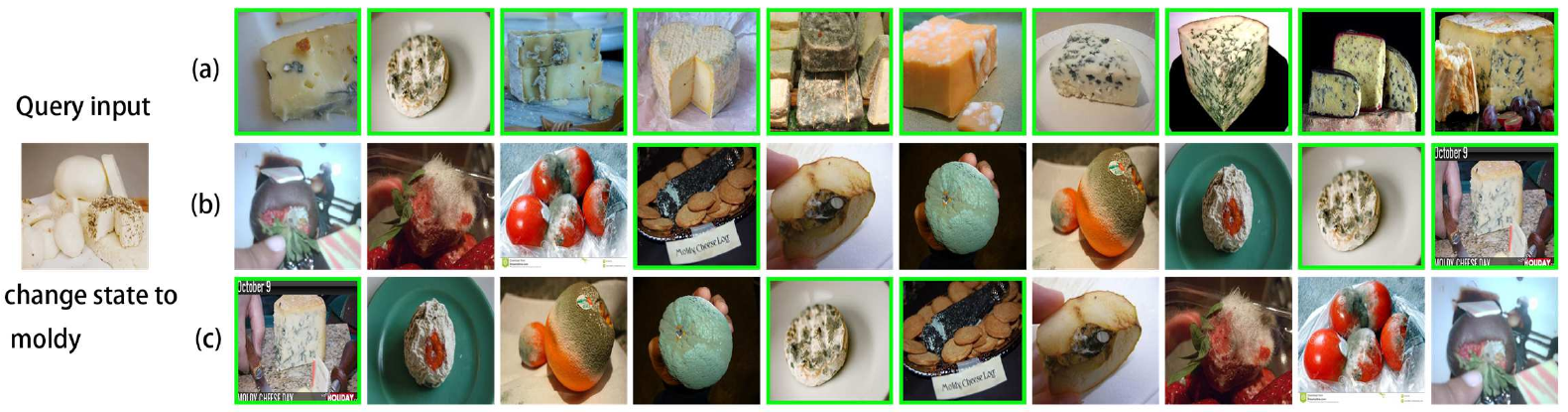}
  \caption{Some result examples of different frameworks with optimal encoders on MIT-States dataset. (a), (b), and (c) are the top-10 results of {\name}, {\bi}, and {\bii}, respectively. The green box marks the target objects.}
  \label{fig: mitstates_case2}
\end{figure*}

Fig. \ref{fig: mitstates_case2} displays the top-10 search results obtained using different frameworks for a query input consisting of a fresh cheese image and the text description ``change state to moldy''. The results indicate that {\name} successfully recalls all the ground-truth objects, whereas {\bi} and {\bii} only retrieve a few ground-truth objects, and many of the results only match the text description.

\begin{figure*}[!tbh]
  \centering
  \setlength{\abovecaptionskip}{0.1cm}
  \setlength{\belowcaptionskip}{-0.2cm}
  \includegraphics[width=\linewidth]{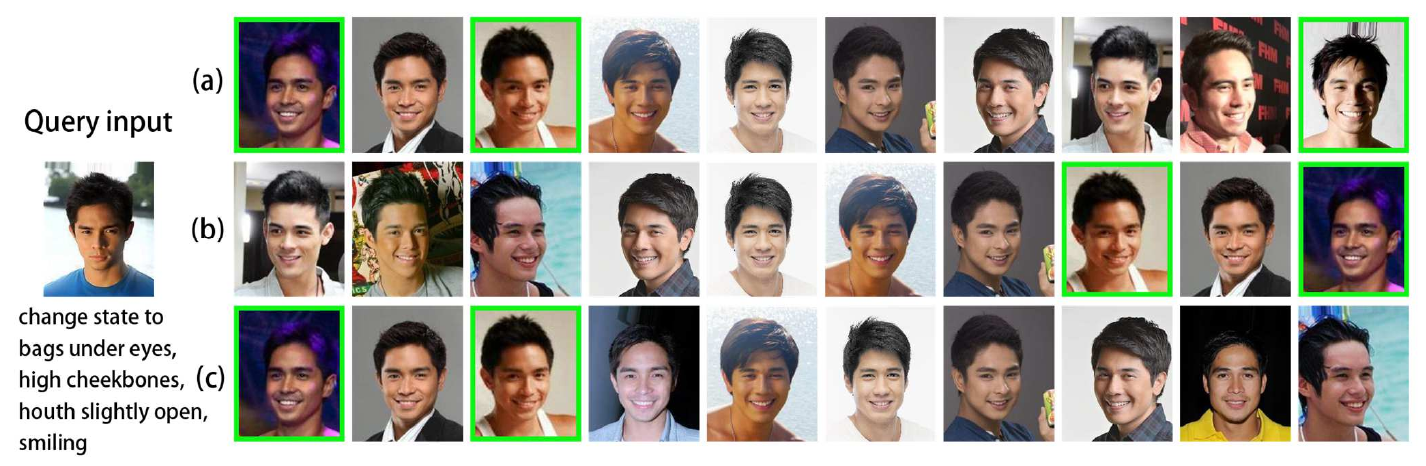}
  \caption{Some result examples of different frameworks with optimal encoders on CelebA dataset. (a), (b), and (c) are the top-10 results of {\name}, {\bi}, and {\bii}, respectively. The green box marks the target objects.}
  \label{fig: celeba_case1}
\end{figure*}

\begin{figure*}[!th]
  \centering
  \setlength{\abovecaptionskip}{0.1cm}
  \setlength{\belowcaptionskip}{-0.2cm}
  \includegraphics[width=\linewidth]{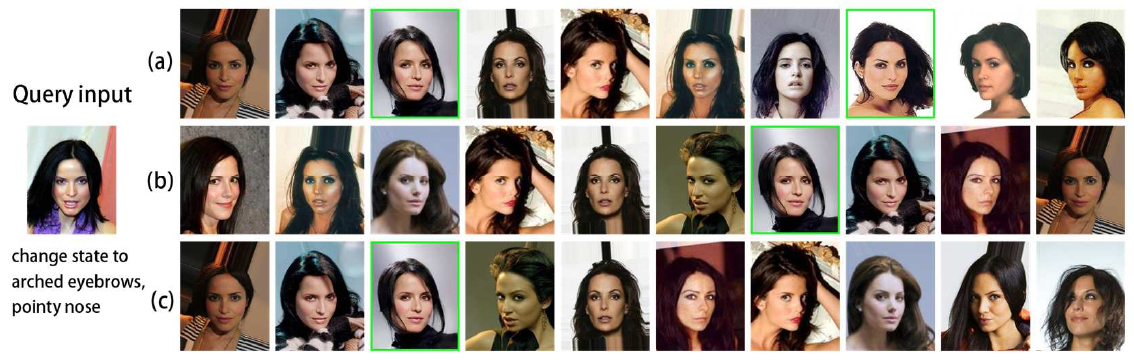}
  \caption{Some result examples of different frameworks with optimal encoders on CelebA dataset. (a), (b), and (c) are the top-10 results of {\name}, {\bi}, and {\bii}, respectively. The green box marks the target objects.}
  \label{fig: celeba_case2}
\end{figure*}

\begin{figure*}[!th]
  \centering
  \setlength{\abovecaptionskip}{0.1cm}
  \setlength{\belowcaptionskip}{-0.2cm}
  \includegraphics[width=\linewidth]{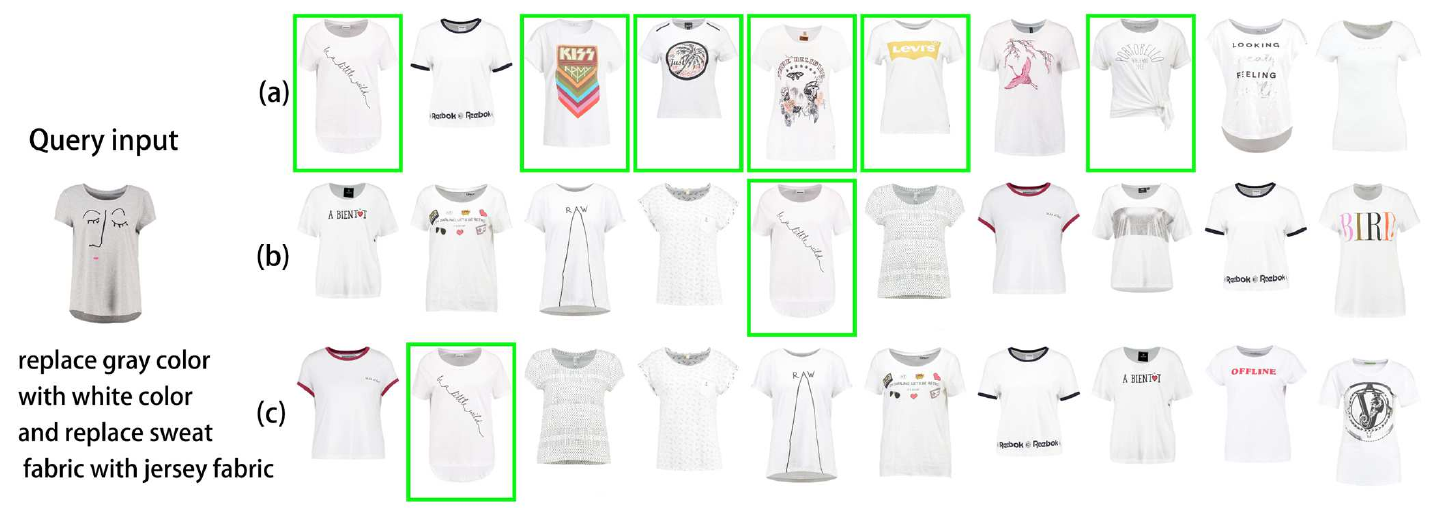}
  \caption{Some result examples of different frameworks with optimal encoders on Shopping dataset. (a), (b), and (c) are the top-10 results of {\name}, {\bi}, and {\bii}, respectively. The green box marks the target objects.}
  \label{fig: shopping_case1}
\end{figure*}

\begin{figure*}[!th]
  \centering
  \setlength{\abovecaptionskip}{0.1cm}
  \setlength{\belowcaptionskip}{-0.2cm}
  \includegraphics[width=\linewidth]{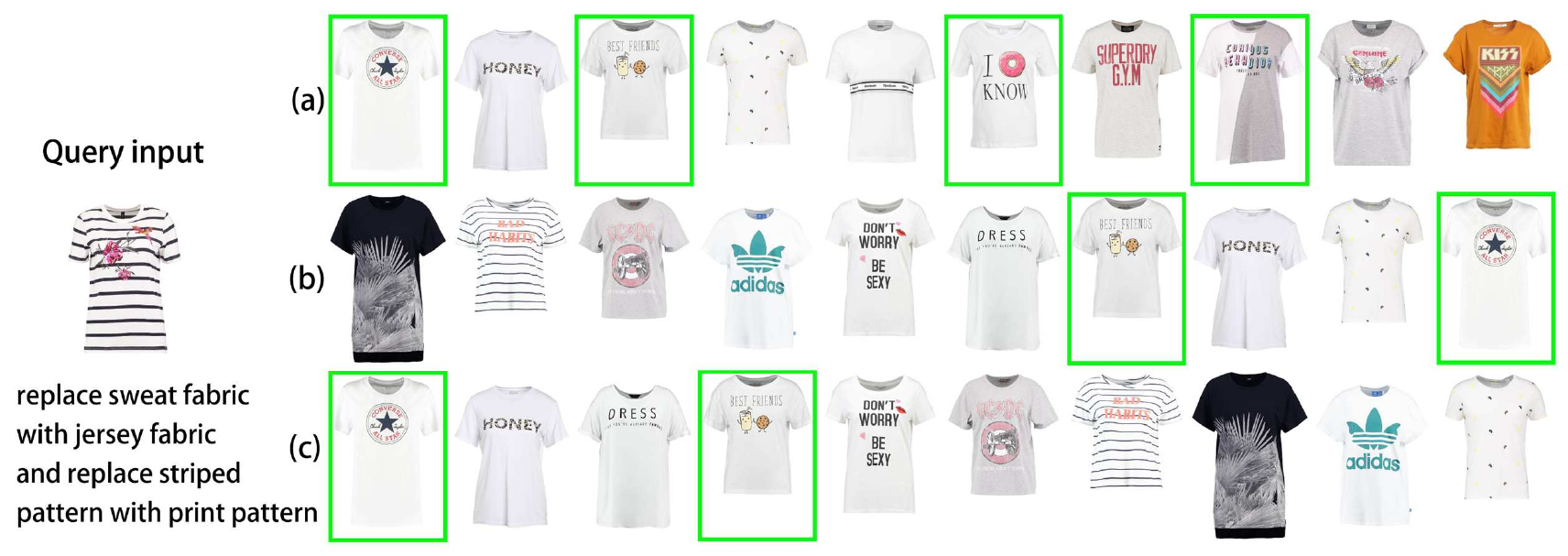}
  \caption{Some result examples of different frameworks with optimal encoders on Shopping dataset. (a), (b), and (c) are the top-10 results of {\name}, {\bi}, and {\bii}, respectively. The green box marks the target objects.}
  \label{fig: shopping_case2}
\end{figure*}

Fig. \ref{fig: celeba_case1} showcases the top-10 search results obtained using different frameworks for a query input consisting of a male face image and the text description ``change state to bags under eyes, high cheekbones, mouth slightly open, and smiling''. It is important to note that in addition to matching the query input face and text description, we also require that the result face and the query input face belong to the same identity. The results demonstrate that {\name} retrieves more ground-truth faces compared to other frameworks.

Fig. \ref{fig: celeba_case1} presents the top-10 search results obtained using different frameworks for a query input consisting of a female face image and the text description ``change state to arched eyebrows and pointy nose''. Similar to the previous example, we require that the result face and the query input face belong to the same identity. This requirement makes it more challenging to obtain the target face. Nevertheless, the results demonstrate that {\name} retrieves more ground-truth faces compared to other frameworks.

Fig. \ref{fig: shopping_case1} illustrates the top-10 search results obtained using different frameworks for a query input consisting of a T-shirt image and the text description ``replace gray color with white color and replace sweat fabric with jersey fabric''. The results indicate that {\name} retrieves more ground-truth T-shirts compared to other frameworks.

Fig. \ref{fig: shopping_case2} displays the top-10 search results obtained using different frameworks for a query input consisting of a T-shirt image and the text description ``replace sweat fabric with jersey fabric and replace striped pattern with print pattern''. The results demonstrate that {\name} retrieves more ground-truth T-shirts compared to other frameworks.

\end{document}